\documentclass[letterpaper]{article}
\usepackage{arxiv}
\usepackage[utf8]{inputenc}
\usepackage[T1]{fontenc}
\usepackage{hyperref}
\usepackage{url}
\usepackage{booktabs}
\usepackage{amsfonts}
\usepackage{nicefrac}
\usepackage{microtype}
\usepackage{lipsum}
\usepackage{graphicx}
\usepackage{natbib}
\usepackage{doi}
\usepackage{subcaption}
\usepackage{pdflscape}
\usepackage{amsmath,amssymb,amsthm}
\usepackage{algorithm}
\usepackage{multirow}
\usepackage{algpseudocode}
\usepackage{epstopdf}
\usepackage{tikz}
\usepackage{pgfplots}
\usepackage{booktabs}
\usepackage{multirow}
\usepackage{amsmath}
\usepackage{mathtools}
\usepackage{floatpag}
\mathtoolsset{showonlyrefs}
\usepgfplotslibrary{fillbetween}
\usetikzlibrary{arrows.meta}
\DeclareMathOperator*{\argmin}{arg\,min}

\newcommand{\x}{\mathbf{x}}

\newcommand{\A}{\mathbf{A}}

\newcommand{\y}{\mathbf{y}}

\newcommand{\R}{\mathbb{R}}

\def\tsc#1{\csdef{#1}{\textsc{\lowercase{#1}}\xspace}}
\tsc{WGM}
\tsc{QE}

\newtheorem{theorem}{Theorem}
\newtheorem{proposition}{Proposition}
\newtheorem{definition}{Definition}

\newtheorem{remark}{Remark}
\newtheorem{assumption}{Assumption}

\title{Self-Tuning Regularization for Image Scanning Microscopy}

\date{} 					

\author{
	{Sofia Agostoni} \\
	MaLGa Center, DIBRIS, University of Genoa\\
	Genoa, Italy \\
	\texttt{sofia.agostoni@edu.unige.it} \\
	\And
	{Lisa Cuneo} \\
	MMS, Istituto Italiano di Tecnologia (IIT)\\
    CSML, Istituto Italiano di Tecnologia (IIT)\\
	Genoa, Italy \\
    \texttt{lisa.cuneo@iit.it}
	\And
	{Christian Daniele} \\
	MaLGa Center, DIBRIS, University of Genoa\\
	Genoa, Italy \\
    \texttt{christian.daniele@edu.unige.it}
	\And
	{Giacomo Garré} \\
	MMS, Istituto Italiano di Tecnologia (IIT)\\
	Genoa, Italy \\
    \texttt{giacomo.garrè@iit.it}
	\And
	{Laurent Le} \\
	MMS, Istituto Italiano di Tecnologia (IIT)\\
	Genoa, Italy \\
    \texttt{laurent.le@iit.it}
	\And
	{Alessandro Zunino} \\
	MMS, Istituto Italiano di Tecnologia (IIT)\\
	Genoa, Italy \\
    \texttt{alessandro.zunino@iit.it}
	\And
	{Giuseppe Vicidomini} \\
	MMS, Istituto Italiano di Tecnologia (IIT)\\
	Genoa, Italy \\
    \texttt{giuseppe.vicidomini@iit.it}
	\And
	{Luca Calatroni} \\
	MaLGa Center, DIBRIS, University of Genoa\\
	MMS, Istituto Italiano di Tecnologia (IIT)\\
	Genoa, Italy \\
    \texttt{luca.calatroni@unige.it}
}

 \renewcommand{\shorttitle}{\textit{arXiv} preprint}

\begin{document}

\let\WriteBookmarks\relax
\def\floatpagepagefraction{1}
\def\textpagefraction{.001}





\maketitle


\begin{abstract}
Image Scanning Microscopy (ISM) and its super-resolution sectioning extension, s$^2$ISM, enable photon-efficient super-resolved fluorescence imaging by combining detector-array acquisition with computational reconstruction. Their standard reconstruction methods, Multi-Image Deconvolution (MID) and s$^2$ISM inversion, rely on Richardson–Lucy-type iterations and therefore suffer from semi-convergence: while early iterations recover fine spatial details, later ones amplify noise and artifacts. In practice, reconstruction quality depends on empirical early stopping, which is difficult to tune and unreliable in the absence of ground truth. We propose a self-tuning explicit regularization framework for MID and s$^2$ISM that replaces implicit early stopping with a Bayesian maximum a posteriori (MAP) formulation. The resulting variational models combine a multi-frame Poisson data fidelity term with sparsity-promoting priors, illustrated through $\ell_1$ and smoothed total variation regularization. To make the approach fully automatic, we adapt the Residual Whiteness Principle to the multi-frame Poisson setting, yielding a ground-truth-free rule for selecting the regularization weight. The framework is compatible with different regularizers and first-order solvers, including proximal-gradient and mirror-descent schemes with adaptive backtracking. Numerical experiments on simulated and real ISM data show that the proposed self-tuning framework stabilizes MID and s$^2$ISM reconstruction, mitigates Richardson–Lucy semi-convergence, and improves super-resolution and optical sectioning performance, particularly in low-photon regimes.
\end{abstract}




\keywords{Image Scanning Microscopy  \and  Poisson inverse problems \and Bayesian reconstruction \and Automatic parameter selection}

\section{Introduction}

Understanding the structure and dynamics of biological systems at the subcellular scale demands imaging techniques capable of resolving features at or below the optical diffraction limit, while minimizing illumination intensity to preserve sample viability. Laser Scanning Microscopy (LSM) forms the architectural backbone of numerous fluorescence imaging techniques widely deployed in biology \cite{pawley2006handbook}, thanks to its high spatiotemporal resolution and quantitative imaging capabilities. Among LSM techniques, Confocal Laser Scanning Microscopy (CLSM) became especially popular for its ability to reject out-of-focus light and achieve superior lateral resolution through the use of a pinhole spatial filter placed in front of the detector. While CLSM can in principle achieve a lateral resolution twice better than the optical diffraction limit in the limiting case of a point-like pinhole \cite{sheppard2006signal}, fully closing the pinhole drastically reduces the number of detected photons and compromises the signal-to-noise ratio (SNR), making this theoretical resolution gain inaccessible in practice.

Image Scanning Microscopy (ISM) \cite{bertero1982resolution,Sheppard1988Super-resolutionImaging,muller2010image} resolves this fundamental trade-off by replacing the single-element detector with an array of detectors — for example, a 5$\times$5 square single-photon avalanche diode (SPAD) array — each element acting as an independent small-pinhole confocal channel \cite{castello2019robust}. Following a complete bidimensional raster scan of the sample, the detector array generates a four-dimensional dataset consisting of two spatial dimensions corresponding to the scanning coordinates and two spatial dimensions corresponding to the detector reference system. Equivalently, the dataset can be interpreted as a collection of 25 confocal-like scanned images, one for each detector element, corresponding to spatially shifted views of the same specimen. Because every detector element contributes to photon collection and no pinhole is used to reject photons coming from the out-of-focus planes, the full SNR is preserved. At the same time the geometric diversity across elements encodes complementary spatial information that enables super-resolution reconstruction. ISM thus achieves sub-diffraction lateral resolution and enhanced SNR simultaneously, without increasing illumination intensity or requiring substantial hardware modifications beyond replacing the conventional single-point detector with a detector array. These features make ISM particularly well suited for the live-cell, low-phototoxicity imaging demanded by modern biomedical research.

Extracting a single, high-quality super-resolved image from the ISM dataset is the central computational challenge of the technique. A simple approach, known as pixel reassignment (PR), theoretically estimates the shift-vectors between detector images and sums the registered channels \cite{
sheppard2013superresolution}. Adaptive pixel reassignment (APR) improves upon this by estimating shift-vectors directly from the data via cross-correlation \cite{castello2019robust}, providing robustness to optical aberrations and misalignments \cite{sheppard2019pixel, fersini2025wavefront}. However, both methods build on the assumption that the detector images are shifted and intensity-rescaled replicas of one another i.e., that they share the same point spread function (PSF). This approximation degrades in realistic imaging conditions, where the PSF varies across detector elements. Moreover, neither PR nor APR explicitly exploits the underlying image formation model or the statistical properties of the measurement noise.

A more rigorous approach, introduced in \cite{castello2019robust,Zunino_2023}, frames ISM reconstruction as a multi-channel (or multi-image) inverse problem. Each detector image is modeled as the convolution of the unknown object with a detector-specific PSF, corrupted by Poisson photon-counting noise. Because SPAD array detectors are essentially free from readout noise and are predominantly limited by photon shot noise, a principled statistical likelihood model can be formulated. Maximum-likelihood estimation via minimization of the Kullback–Leibler (KL) divergence leads to a multiplicative iterative algorithm — Multi-Image Deconvolution (MID) — which generalizes the classical Richardson–Lucy algorithm to the multi-detector ISM setting \cite{Zunino_2023}. By fully accounting for the image formation process and the Poissonian noise statistics, MID outperforms APR in both resolution and SNR recovery. In \cite{Zunino_2023} it was further demonstrated that the redundancy intrinsic to the ISM dataset can be exploited to relax the Nyquist–Shannon sampling criterion by a factor of two, enabling a fourfold improvement in acquisition speed without information loss.

A further significant advance was introduced in \cite{s2ISM} with the simultaneous super-resolution and optical sectioning ISM framework (s$^2$ISM). The key observation underpinning s$^2$ISM is that imaging with a detector array inherently embeds axial information \cite{tortarolo2022focus}: the fingerprint function, which encodes the brightness distribution of each PSF across the detector elements, varies with the axial position of the emitter. Specifically, in-focus and out-of-focus fluorescence distribute differently across the detector array, with central elements collecting primarily in-focus light and peripheral elements increasingly collecting out-of-focus contributions. By modeling the ISM dataset as the superposition of contributions from two discrete axial planes — an in-focus plane of interest and an out-of-focus background plane — and inverting the resulting forward model under a Poisson likelihood, s$^2$ISM simultaneously achieves super-resolution, enhanced SNR, and optical sectioning from a single-plane acquisition, i.e., a single two-dimensional raster scan, without any modification to the optical system. 
Notably, MID can also be generalized to three-dimensional imaging by acquiring an additional axial scan and reconstructing a volumetric image stack. In this case, each reconstructed plane can achieve spatial resolution and optical sectioning performance comparable to s$^2$ISM. However, this approach requires a full three-dimensional acquisition, resulting in increased acquisition time and phototoxicity, which substantially limits its applicability to fast live-cell imaging.

Despite these advances, a fundamental limitation of both s$^2$ISM and MID remains unresolved. These methods inherit the well-known semi-convergent behavior of Richardson–Lucy deconvolution: while early iterations progressively recover high-frequency spatial details and improve image sharpness, continued iteration amplifies noise and leads to overfitted, artifact-laden reconstructions that progressively diverge from the true object. This behavior was explicitly demonstrated in \cite{Zunino_2023} where the KL divergence between the reconstruction and the ground truth was shown to initially decrease, reach a minimum, and then rise monotonically. In practice, this instability was managed by adopting a conservative early stopping criterion, chosen empirically to prevent artifact generation. Crucially, the authors therein explicitly acknowledged that the inclusion of regularization rules had not been considered, and identified the development of explicit regularization — including sparsity and continuity priors — as a needed future development, noting that careful evaluation of iteration stopping is essential and that integration of regularization approaches into the reconstruction pipeline represents a natural and important extension of the method for improving the robustness and usability of the MID and s$^2$ISM Richardson-Lucy based algorithms.

To overcome these limitations in a principled and mathematically grounded way, we develop in this paper a comprehensive explicit regularization framework for both MID and s$^2$ISM reconstructions. We incorporate prior knowledge about the unknown object directly into the objective function within the Maximum A Posteriori (MAP) estimation framework, combining the multi-frame KL data fidelity term with a regularization functional weighted by a parameter $\lambda>0$ that controls the trade-off between data fit and the prior. We consider two representative and widely used regularization functionals: a smoothed Total Variation (TV) penalty \cite{RUDIN1992259}, which promotes piecewise-smooth reconstructions while preserving sharp edges, and an $\ell_1$ norm penalty \cite{Candes}, which enforces sparsity in the solution. 
Such regularization strategies have been extensively employed in Poisson inverse problems and computational imaging to stabilize reconstruction while preserving relevant image structures, both in variational Total Variation formulations and in sparse-representation approaches, see, e.g., ~\cite{Dupe2009,Figureido2010,carl_blanc_f}. A critical practical challenge associated with explicit regularization is the selection of the regularization parameter, which controls the balance between data fidelity and prior information. In fluorescence microscopy applications, no reference reconstruction is available in practice to directly assess the reconstruction error, making parameter selection a fully unsupervised problem. Poorly chosen values of such parameter may lead either to over-regularization, resulting in the loss of biologically relevant fine structures, or to under-regularization, reintroducing noise amplification and reconstruction artifacts.
To address this issue, we adapt the whiteness-based parameter selection strategies proposed in~\cite{BEVILACQUA2023197,Bevilacqua_2023,bevilacqua2021nearly} to both the ISM and $s^2$ISM settings. Furthermore, in order to make the procedure robust in the dual-plane $s^2$ISM scenario and decouple parameter selection from structured low-frequency artifacts arising from the unregularized background component, we introduce a high-pass spectral filtering strategy together with a robust knee-point detection procedure aimed at improving the numerical stability of the parameter-selection process.

\paragraph{Contributions.}
We formulate MID and s$^2$ISM reconstruction within a Bayesian MAP framework and derive explicitly regularized variational models that overcome the need for empirical early stopping in Richardson–Lucy-type schemes. The framework combines the multi-frame Poisson data fidelity naturally associated with detector-array photon-counting measurements with sparsity-promoting priors, illustrated here through $\ell_1$ and smoothed total variation regularization. We further introduce a fully automatic, ground-truth-free strategy for selecting the regularization parameter by adapting the Residual Whiteness Principle to multi-frame Poisson ISM data. For s$^2$ISM, we complement this criterion with a high-pass spectral residual analysis that reduces the influence of low-frequency correlations induced by the unregularized out-of-focus component. The resulting self-tuning framework is compatible with different regularizers and first-order solvers, and is validated on simulated and real fluorescence microscopy datasets, where it stabilizes MID and s$^2$ISM reconstructions, mitigates noise amplification, and improves super-resolution and optical sectioning in low-photon regimes.

\section{Image Scanning Microscopy with Pinhole: the MID Case} \label{sec1}
In this section, we address the reconstruction of single-plane (two-dimensional) ISM datasets using multi-image deconvolution (MID). In the original ISM formulation, the confocal pinhole is completely removed and replaced by a detector array, thereby maximizing the photon collection efficiency from the focal plane. Under these conditions, when imaging a specific section of a three-dimensional specimen, the detected signal contains not only the in-focus fluorescence originating from the focal plane, but also out-of-focus contributions arising from planes above and below it.
When a full volumetric dataset is acquired, i.e., when scanning is also performed along the optical axis, out-of-focus light does not represent a fundamental limitation for the three-dimensional implementation of MID. Indeed, fluorescence contributions that appear out-of-focus in one section correspond to in-focus signal in neighboring axial planes and can therefore be reassigned to their correct spatial position during the volumetric reconstruction process. As a consequence, three-dimensional MID can simultaneously exploit the full photon budget while preserving optical sectioning. In many practical applications, however, acquiring a complete volumetric dataset is not feasible because of acquisition-time constraints, photobleaching, and phototoxicity, as is often the case in fast live-cell imaging. In these scenarios, only a single focal plane is recorded. Under such conditions, the axial information required to reassign out-of-focus photons is not directly available and, when a two-dimensional implementation of MID is employed, fluorescence originating outside the focal plane accumulates as background in the reconstructed image. For this reason, practical ISM implementations typically still employ a relatively large confocal pinhole, commonly around 1 Airy Unit (A.U.), often effectively defined by the detector array itself. This configuration provides sufficient optical sectioning to suppress most out-of-focus fluorescence contributions while still preserving a large fraction of the in-focus light, allowing two-dimensional MID to effectively improve both the resolution and the signal-to-noise ratio of the reconstructed image. 
\subsection{Image formation and mathematical modeling}
\label{sec:math_mod_oneplane}
The ISM image formation is modeled as a multi-frame Poisson inverse problem. We follow \cite{Bertero} and model the 2D ISM image formation process as
\begin{equation}
\label{for_problem}
    \mathbf{y}_d \sim \operatorname{Poiss}(\A_d \x + \mathbf{b}_d),
    \qquad d \in \{1, \dots, 25\}.
\end{equation}

Here, $\x \in \mathbb{R}^N$ denotes the vectorized image of the unknown object to be recovered, and $\mathbf{y}_d \in \mathbb{R}^N$ represents the measurements collected by the $d$-th detector element, for $d = 1, \ldots, 25$. The operator $\A_d \in \mathbb{R}^{N \times N}$ is the convolution matrix associated with the $d$-th effective point spread function (PSF) $h_d$. Crucially, the PSF is defined as the product of the excitation PSF $h_{\text{exc}}$ and the detection PSF $h_{\text{det}}$ (which accounts for the specific spatial shift of each element relative to the optical axis), such that $h_d(\mathbf{r}) = h_{\text{exc}}(\mathbf{r}) \cdot h_{\text{det}}(\mathbf{r})$. 
For $d = 1, \ldots, 25$, the vector $\mathbf{b}_d = b_d \mathbf{1}_{\mathbb{R}^N}$ denotes the background term, which is assumed to be spatially constant for each detector, such as the dark-count background typical of SPAD.
Physically, they model constant background signals present in the measurements. From a mathematical perspective, the presence of such terms guarantees strict positivity of the forward model, which will be required to establish smoothness properties of the objective functional. Notably, the forward model described above applies to both single-plane and volumetric imaging, provided that the corresponding two-dimensional or three-dimensional PSFs are considered for each detector element. In the case of volumetric imaging, the ISM dataset becomes five-dimensional, comprising three spatial dimensions associated with the scanning coordinates and two dimensions associated with the detector array.

Assuming statistical independence across pixels and detector elements \cite{Bertero}, the likelihood function associated with \eqref{for_problem} is given by
\begin{equation} \label{eq:Pois_likelihood}
\mathcal{P}_{\mathcal{Y}|\mathcal{X}}(\mathbf{y} \vert \x)
=
\prod_{d=1}^{25}
\prod_{i=1}^{N}
\frac{
[\A_d \x + \mathbf{b}_d]_i^{\,y_{d,i}}
e^{-[\A_d \x + \mathbf{b}_d]_i}
}{
y_{d,i}!
}.
\end{equation}

Here, $y_{d,i} \in \mathbb{N}$ denotes the photon count measured at pixel $i$ of detector $d$. The corresponding negative log-likelihood reads
\begin{equation}
\label{eq:likelihood}
-\log \mathcal{P}_{\mathcal{Y}|\mathcal{X}}(\mathbf{y} \vert \x)
=
\sum_{d=1}^{25}
\sum_{i=1}^{N}
\left\{
[\A_d \x + \mathbf{b}_d]_i
-
y_{d,i}\ln\bigl([\A_d \x + \mathbf{b}_d]_i\bigr)
+
\ln(y_{d,i}!)
\right\}.
\end{equation}

Since the factorial terms $\ln(y_{d,i}!)$ do not depend on $\x$, they can be discarded without affecting the minimizer. The maximum likelihood estimate of $\x$ is therefore obtained by solving:
\begin{equation}
\label{ml_min}
\x^*
\in
\operatorname*{argmin}_{\x}
~
\Phi^{\text{ISM}}(\x;\underline{\mathbf{A}}, \underline{\mathbf{y}}, \underline{\mathbf{b}})
:=
\sum_{d=1}^{25}
\operatorname{KL}(\mathbf{y}_d,\A_d\x+\mathbf{b}_d),
\end{equation}
where
$\underline{\A} = [\A_1,\dots,\A_{25}]$,
$\underline{\y} = [\y_1,\dots,\y_{25}]$, $\underline{\mathbf{b}} = [ \mathbf{b}_1, \dots, \mathbf{b}_{25}]$
and $\operatorname{KL}(\cdot,\cdot)$ denotes the generalized Kullback--Leibler divergence:
\[
\operatorname{KL}(\y,\x)
=
\sum_{i=1}^{N}
x_i
-
y_i \ln(x_i)
+
\ln(y_i!).
\]

The resulting multi-frame $\operatorname{KL}$ data fidelity term is convex and, thanks to the presence of the background term, has Lipschitz-continuous gradient, i.e. it is $L$-smooth (Definition~\ref{L-smothness}). Its gradient is given by:
\begin{equation}
\label{eq:grad_ISM}
\nabla \Phi^{\text{ISM}}(\x;\underline{\A},\underline{\y},\mathbf{b})
=
\sum_{d=1}^{25}
\A_d^\top
\left(
\mathbf{1}
-
\frac{\y_d}{\A_d\x+\mathbf{b}_d}
\right),
\end{equation}
where the division is understood entrywise. For the optimization algorithms considered in the following sections, we will require an upper bound on the Lipschitz constant of \eqref{eq:grad_ISM}. As shown in Proposition~\ref{KL_L}, one has
\begin{equation}
L_{\Phi^{\text{ISM}}}
\leq
\sum_{d=1}^{25}
\frac{\max(\mathbf{y}_d)}{\mathbf{b}_d^2}
\cdot
\max(\A_d \mathbf{1})
\cdot
\max(\A_d^\top \mathbf{1}).
\end{equation}

The $L$-smoothness of the data fidelity term $\Phi^{\text{ISM}}$ will be instrumental in establishing convergence guarantees for the optimization algorithms introduced in the following sections.

\section{Image Scanning Microscopy without Pinhole: the s$^2$ISM Case}

In this section, we introduce s$^2$ISM and we show how it overcomes the residual trade-off between photon collection efficiency and optical sectioning that still characterizes traditional single-plane ISM acquisitions. Unlike standard ISM, where the pinhole must remain partially closed to suppress out-of-focus background, s$^2$ISM enables efficient optical sectioning even when the pinhole is fully opened, or equivalently when detector arrays extending well beyond 1 Airy Unit are employed.
The key advantage of s$^2$ISM is that optical sectioning is achieved computationally through higher-order spatial correlations rather than through physical rejection of out-of-focus photons. As a consequence, fluorescence contributions originating outside the focal plane are enables computational suppression without sacrificing photon collection efficiency. This allows the simultaneous recovery of enhanced lateral resolution, improved signal-to-noise ratio, and effective background rejection.
Importantly, these benefits are obtained while still operating in a single-plane acquisition modality, without requiring the registration of a complete three-dimensional stack. Therefore, s$^2$ISM alleviates the compromise that characterizes conventional 2D ISM implementations and enables high-resolution, optically sectioned imaging even in thick specimens under fully open-pinhole detection conditions. 

\subsection{Image formation and mathematical modeling} \label{sec: mathematical modeling for $s^2ISM$}
In the $s^2$ISM framework, the forward model is extended to account 
for axial information by treating the fluorescence signal as a superposition of 
in-focus and out-of-focus contributions. The measurement for each detector $d$ 
is given by
\begin{equation}
\label{eq:s2_forward}
\mathbf{y}_d
\sim
\operatorname{Poisson}
\left(
\A_{d,1}\x_1
+
\A_{d,2}\x_2
+
\mathbf{b}_d
\right),
\qquad
\forall\, d \in \{1,\dots,25\}.
\end{equation}
The variable $\x_1 \in \R^N$ represents the in-focus signal of interest, while 
$\x_2 \in \R^N$ models the out-of-focus contribution, treated as a background 
component to be separated from the reconstruction. Here, $\A_{d,1}$ is the 
convolution matrix defined by the PSF at the focal plane ($z = 0$), denoted 
$h_{d,1}$, and $\A_{d,2}$ is the convolution matrix defined by the PSF at a 
displaced axial plane ($z = \Delta z$), denoted $h_{d,2}$. By incorporating 
these depth-dependent PSFs into the forward operators, the model simultaneously 
achieves super-resolved lateral reconstruction in $\x_1$ and enhanced axial 
sectioning through the separation of the out-of-focus background $\x_2$.

Similarly to Section~\ref{sec:math_mod_oneplane}, the corresponding maximum likelihood estimation problem is formulated as
\begin{equation}
\label{s2ISM}
\underline{\x}^*
=
[\x_1^*,\x_2^*]
\in
\operatorname*{argmin}_{\underline{\x}\in\mathbb{R}^{N\times 2}}
~
\Phi^{^{\text{$s^2$ISM}}}
(
\underline{\x};
[\underline{\mathbf{A}}_1,\underline{\mathbf{A}}_2],
\underline{\mathbf{y}},
\underline{\mathbf{b}}
)
:=
\sum_{d=1}^{25}
\operatorname{KL}
\bigl(
\mathbf{y}_d,
\mathbf{A}_{d,1}\mathbf{x}_1
+
\mathbf{A}_{d,2}\mathbf{x}_2
+
\mathbf{b}_d
\bigr).
\end{equation}

The resulting data fidelity term is convex and $L$-smooth. As shown in Remark~\ref{rmk: lipschitz constant of KL in 2 planes ism}, the Lipschitz constant of its gradient admits the upper bound
\begin{equation}
L_{\Phi{^{\text{$s^2$ISM}}}}
\leq
\sum_{d=1}^{25}
\frac{\max(\mathbf{y}_d)}{\mathbf{b}_d^2}
\cdot
\max\bigl([\mathbf{A}_1,\mathbf{A}_2]_d\mathbf{1}\bigr)
\cdot
\max\bigl([\mathbf{A}_1,\mathbf{A}_2]_d^\top\mathbf{1}\bigr).
\end{equation}

Here,
$
[\mathbf{A}_1,\mathbf{A}_2]_d :
\mathbb{R}^{N\times 2}
\to
\mathbb{R}^N
$
denotes the linear operator defined by
\[
[\mathbf{A}_1,\mathbf{A}_2]_d \underline{\mathbf{x}}
=
\mathbf{A}_{d,1}\mathbf{x}_1
+
\mathbf{A}_{d,2}\mathbf{x}_2,
\qquad
\forall d=1,\dots,25.
\]

\section{Implicit regularization and early stopping: drawbacks}
\label{sec: implicit reg and early stopping}

Problems~\eqref{ml_min} and~\eqref{s2ISM} were previously addressed in~\cite{Zunino_2023} and~\cite{s2ISM}, respectively. In these works, the reconstructed image was obtained through tailored variants of the Richardson--Lucy (RL) algorithm~\cite{Richardson72,Lucy1974}, a classical multiplicative iterative scheme originally introduced to minimize the Kullback--Leibler divergence within a Maximum Likelihood Expectation-Maximization framework~\cite{Shepp1982}.
 In such approaches, reconstruction stability is not enforced through an explicit regularization term added to the optimization problem. Instead, regularization is induced implicitly by stopping the iterative process before convergence. This phenomenon, commonly referred to as \emph{implicit regularization via early stopping}, relies on the observation that the first RL iterations recover meaningful large-scale structures, while later iterations progressively amplify noise and reconstruction artifacts, see also \cite{Liu} for a frequency analysis in the context of microscopy.

Taking as example the forward model~\eqref{for_problem}, the RL iteration associated with the ISM model reads for $k\geq 0$
\begin{equation}
\label{eq: Richardson-Lucy}
\x^{k+1}
=
\x^k
\cdot
\underline{\mathbf{A}}^\top
\left(
\frac{\underline{\mathbf{y}}}
{\underline{\mathbf{A}}\x^k+\underline{\mathbf{b}}}
\right),
\end{equation}
where
\[
\underline{\mathbf{A}}
=
[\mathbf{A}_1,\dots,\mathbf{A}_{25}]
:
\mathbb{R}^N
\to
\mathbb{R}^{25\times N},\quad
\underline{\mathbf{y}}
=
[\mathbf{y}_1,\dots,\mathbf{y}_{25}],
\qquad
\underline{\mathbf{b}}
=
[\mathbf{b}_1,\dots,\mathbf{b}_{25}],
\]
and all products and divisions are understood entrywise.

While RL typically produces reasonable reconstructions during the first iterations, it is inherently affected by the phenomenon of \emph{semi-convergence}: as the iterations progress, the algorithm progressively overfits the noise in the measurements, causing the iterates to drift away from the true solution and yielding increasingly unstable reconstructions. Consequently, the performance of RL critically depends on identifying a suitable stopping iteration --- early enough to preserve meaningful structures, yet before noise amplification becomes dominant. In practice, however, determining such an optimal stopping point is highly nontrivial, particularly in real-world microscopy applications where the ground truth is unavailable and reliable stopping criteria are difficult to define.

This behavior is illustrated in the simulated example of Fig.~\ref{fig:no_reg_conv}. While the multi-frame KL objective function decreases monotonically throughout the iterations --- as guaranteed by the construction of RL --- the $\ell_2$ reconstruction error with respect to the ground truth follows a markedly different trend: it initially decreases, reaches a minimum at an intermediate iteration, and subsequently increases as noise amplification sets in. This discrepancy between objective minimization and reconstruction quality highlights the fundamental limitation of implicitly regularized iterative schemes and motivates the explicit variational regularization framework developed in this work.
\begin{figure}
    \centering
        \begin{subfigure}{0.48\linewidth}
        \centering
        \includegraphics[width=\linewidth]{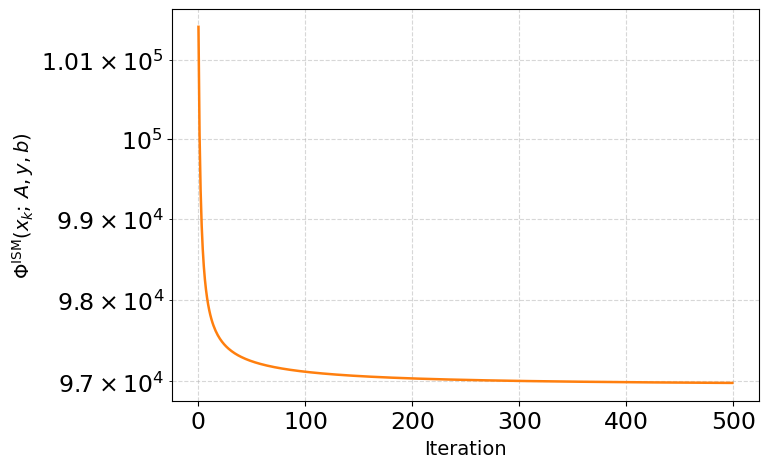}
        \caption{Decrease of $\Phi^{\text{ISM}}$  \eqref{ml_min} along RL iterations.}
    \end{subfigure}
    \hfill
    \begin{subfigure}{0.48\linewidth}
        \centering
        \includegraphics[width=\linewidth]{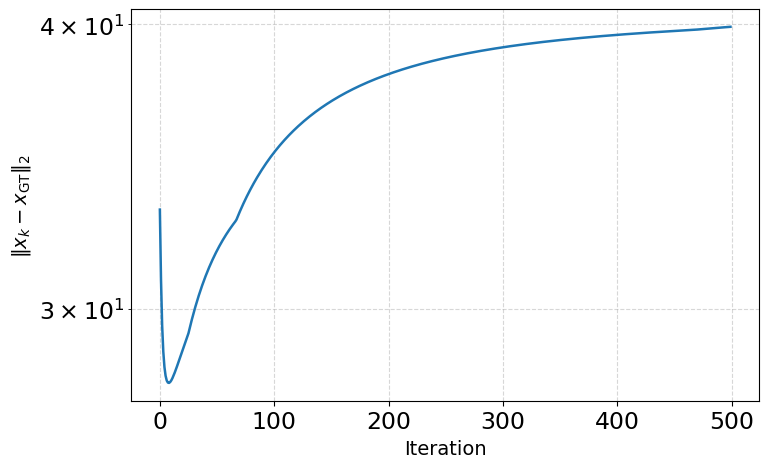}
        \caption{Evolution of the reconstruction error $\| \mathbf{x}^k-\mathbf{x}_{\text{GT}}\|_2$.}
    \end{subfigure}
    \caption{Semi-convergence of the Richardson--Lucy algorithm on a simulated example. 
    While the $\Phi^{\text{ISM}}$ objective (left) decreases monotonically, the $\ell_2$-reconstruction error 
    between RL iterates and the ground truth  reaches a minimum and subsequently grows (right), 
    reflecting progressive noise amplification.}
    \label{fig:no_reg_conv}
\end{figure}

\section{Bayesian formulation and regularization models}
\label{sec2}

To address the limitations discussed above, we develop an explicit regularization framework for robust and stable ISM and s$^2$ISM reconstruction.

\subsection{Maximum A Posteriori (MAP) approach} \label{sec:MAP}

To overcome the limitations of implicitly regularized reconstruction schemes, we adopt a Bayesian Maximum A Posteriori (MAP) formulation. In this framework, prior information about the unknown object is incorporated through a prior probability distribution $\mathcal{P}_{\mathcal{X}}(\x)$, leading naturally to an explicitly regularized variational problem. By Bayes' theorem, the posterior distribution of the unknown image satisfies
\[
\mathcal{P}_{\mathcal{X}|\mathcal{Y}}(\x \mid \mathbf{y})
\propto
\mathcal{P}_{\mathcal{Y}|\mathcal{X}}(\mathbf{y}\mid\x)
\,
\mathcal{P}_{\mathcal{X}}(\x),
\]
where $\mathcal{P}_{\mathcal{Y}|\mathcal{X}}$ denotes the likelihood introduced in~\eqref{eq:likelihood}. Following a standard variational Bayesian approach, we consider priors of the form
\begin{equation}
\label{MAP_eq}
\mathcal{P}_{\mathcal{X}}(\x)
\propto
\exp\bigl(-\lambda R(\x)\bigr),
\end{equation}
where
$
R:\mathbb{R}^N \to \mathbb{R}_{+}
$
is a convex regularization functional and $\lambda>0$ is the regularization parameter controlling the trade-off between data fidelity and prior information.

Taking the negative logarithm of the posterior distribution and neglecting additive constants independent of $\x$, the MAP estimate is obtained by solving the composite optimization problem
\begin{equation}
\label{Map_min}
\x^*
\in
\operatorname*{argmin}_{\x}
\;
\Phi^{\text{ISM}}
(\x;\underline{\A},\underline{\y},\underline{\mathbf{b}})
+
\lambda R(\x)
+
\iota_{\geq0}(\x).
\end{equation}

Here,
$\Phi^{\text{ISM}}$
denotes the multi-frame Poisson data fidelity term introduced in~\eqref{ml_min}, while
$\iota_{\geq0}$
is the indicator function of the non-negative orthant, enforcing the physical non-negativity constraint on the reconstructed fluorescence signal.

In the $s^2$ISM setting, the same variational formulation applies with the data fidelity term $\Phi^{s^2\text{ISM}}$ defined in~\eqref{s2ISM}. Since only the in-focus component is the reconstruction target, regularization is applied only to $\x_1$. The resulting optimization problem thus reads in this case:
\begin{equation} \label{eq:prob_s2ISM}
\underline{\x}^*
=
[\x_1^*,\x_2^*]
\in
\operatorname*{argmin}_{\underline{\x}\in\mathbb{R}^{N\times2}}
\;
\Phi^{s^2\text{ISM}}
(
\underline{\x};
[\underline{\mathbf{A}}_1,\underline{\mathbf{A}}_2],
\underline{\mathbf{y}},
\underline{\mathbf{b}}
)
+
\lambda R(\x_1)
+
\iota_{\geq0}(\underline{\x}).
\end{equation}

\subsection{On the choice of the  regularization function $R$} \label{sec:reg}

We now introduce the representative regularization functionals considered in this work. As discussed above, these choices are intended to illustrate how classical sparsity-promoting explicit regularization strategies can be naturally incorporated into the multi-frame $\operatorname{KL}$ minimization framework arising in ISM and $s^2$ISM reconstruction, but more advanced choices based on higher-order regularization \cite{TGV}, spatial adaptivity \cite{TV_space_variant} and/or data-driven extensions \cite{hertrich2025learning} can be considered too. In the following, we consider
exemplar regularization functionals of the form
$
R : \mathbb{R}^N \to \mathbb{R}_{+},
$
where $R$ is assumed to be convex, proper, lower semi-continuous, and possibly non-smooth. 

\subsubsection{$\ell_1$ regularization}
\label{sec:l1_reg}

A widely used explicit regularization strategy is the $\ell_1$ norm, which promotes sparsity in the reconstructed signal by penalizing the sum of the absolute values of its entries. Originally popularized in the context of compressed sensing~\cite{Candes} and sparse regression~\cite{tibshirani1996regression}, it has since become a standard regularization tool in inverse problems where the underlying solution is expected to exhibit sparse structures. The $\ell_1$ regularizer is defined as
\begin{equation}
\label{eq:l1_norm}
    R(\mathbf{x})
    =
    \|\mathbf{x}\|_1
    =
    \sum_{i=1}^{N} |x_i|.
\end{equation}

\subsubsection{Smoothed Total Variation}
\label{sec:tv_reg}

Another widely used regularization strategy in image reconstruction is \emph{Total Variation} (TV) regularization~\cite{RUDIN1992259}, which is particularly effective at preserving edges while suppressing noise by promoting sparsity in the image gradient. In its classical form, TV is non-differentiable, and its optimization typically requires primal--dual splitting techniques; see, e.g.,~\cite{ChambollePock}.  In this work, we adopt for simplicity a \emph{smoothed} variant of TV in order to make the regularization functional compatible with the gradient-based optimization methods discussed in Sections~\ref{sec:PGD} and~\ref{sec:MD}. The smoothing is introduced through a small parameter $\epsilon > 0$, ensuring differentiability everywhere. The resulting smoothed TV regularizer is defined as
\begin{equation}
\label{eq:TV_smooth}
\operatorname{TV}_{\epsilon}(\mathbf{x})
:=
\sum_{i=1}^{N}
\sqrt{
\left(\nabla_h \mathbf{x}\right)_i^2
+
\left(\nabla_v \mathbf{x}\right)_i^2
+
\epsilon^2
},
\end{equation}
where
$
\nabla
=
\begin{bmatrix}
\nabla_h \\
\nabla_v
\end{bmatrix}
\in
\mathbb{R}^{2N\times N}
$
denotes the discrete finite-difference spatial gradient operator, with $\nabla_h$ and $\nabla_v$ representing its horizontal and vertical components, respectively. The parameter $\epsilon > 0$ controls the degree of smoothing near zero.
The gradient of $\operatorname{TV}_{\epsilon}$ is given by:
\begin{equation}
\label{eq:TV_grad}
    \nabla \operatorname{TV}_{\epsilon}(\mathbf{x}) 
    = -\nabla^{\top} \frac{\nabla \mathbf{x}}{\sqrt{\|\nabla \mathbf{x}\|^2 + \epsilon^2}},
\end{equation}
where the division is understood entrywise. An upper 
bound for the Lipschitz constant of the resulting composite gradient is derived in Proposition \ref{prop:lip_F}.

\section{Automatic parameter selection via Residual Whiteness Principle}
\label{secRPW}

Selecting an appropriate value of the regularization parameter $\lambda$ in~\eqref{MAP_eq}, balancing data fidelity and regularization, is a long-standing problem in inverse problems and variational imaging. Classical approaches include discrepancy-principle-based strategies, originally introduced for additive Gaussian noise~\cite{morozov2012methods} and later adapted to Poisson inverse problems in~\cite{Bertero_2010,bevilacqua2021nearly}. These methods select $\lambda$ so that the residual matches the expected noise level.

Other widely used parameter selection techniques include the L-curve criterion~\cite{hansen1992analysis}, which identifies a compromise between data fidelity and regularization by analyzing the curvature of the corresponding trade-off curve, and unbiased risk estimation approaches such as SURE~\cite{Stein1981} (Stein's Unbiased Risk Estimator), originally developed for Gaussian noise and subsequently extended to Poisson statistics in the form of PURE~\cite{PURE} (see also \cite{Massa_2021}).

Despite their effectiveness, these approaches may be difficult to apply in practice, either because they require accurate knowledge of the noise statistics or because they involve computationally demanding parameter-search procedures. This challenge is particularly critical in fluorescence microscopy, where the ground truth is unavailable and the noise statistics may significantly vary across detector channels and acquisition settings.

To address this issue, we focus on a different class of parameter selection strategies based on the statistical analysis of the reconstruction residuals. Rather than enforcing consistency with a prescribed noise level, these approaches assess whether the residual behaves as an uncorrelated realization of the underlying noise process. In particular, we consider the family of residual whiteness principles, firstly applied in a number of papers for Gaussian statistics \cite{Lanza_RWP1,Lanza_RWP2,PragliolaRWP} and later used in \cite{BEVILACQUA2023197,Bevilacqua_2023,bevilacqua2021nearly} in the case of Poisson statistics.

\smallskip

Such an approach provides a fully automatic parameter selection strategy by exploiting the statistical structure of the reconstruction residuals, rather than relying on prior knowledge of the noise variance or on ground-truth information. More precisely, the underlying idea is that, for a properly regularized reconstruction, the residual should behave as an uncorrelated realization of the underlying Poisson noise process. 

In the ISM setting, each acquisition consists of a collection of detector-dependent images rather than a single observation. Since the forward model assumes conditional independence of photon counts across both spatial pixels and detector elements, the likelihood naturally factorizes over the complete measurement index set \eqref{eq:Pois_likelihood}. Motivated by this multi-frame statistical model, we extend the residual whiteness criterion to the full observation space by evaluating residual correlations jointly across the spatial and detector dimensions. Although detector images share common object structures through the forward model, the proposed criterion is applied to the standardized residuals, where the common structural content is expected to be largely explained by the detector-dependent forward operators. This allows the parameter-selection strategy to exploit the information contained in the full detector-array acquisition.

Let us now describe how this principle can be adapted to the multi-frame ISM setting. Let the observed data be organized as a matrix
$
\mathbf{Y}\in\mathbb{R}^{N\times D},
$
with scalar entries $y_{i,d}$, where
$
(i,d)\in\mathcal I
:=
\{1,\dots,N\}\times\{1,\dots,D\}
$
indexes pixels and detector frames, respectively. We model each measurement as an independent Poisson random variable:
\begin{equation}
y_{i,d}
\sim
\operatorname{Poiss}(\nu_{i,d}),
\qquad
\nu_{i,d}
:=
[\mathbf A_d\mathbf x+\mathbf b_d]_i,
\end{equation}
where $\nu_{i,d}$ denotes the expected photon count at pixel $i$ and detector frame $d$. Collecting these quantities into the matrix
$
\boldsymbol\nu\in\mathbb R^{N\times D},
$
we can write compactly
$
\mathbf Y \sim \operatorname{Poiss}(\boldsymbol\nu)
$
entrywise.
We introduce the multi-frame  lag domain
\[
\mathcal L
:=
\{-(N-1),\dots,N-1\}
\times
\{-(D-1),\dots,D-1\}.
\]
and recall the following standard definitions.

\begin{definition}[Weakly stationary and white random fields]
A random field
$
\mathcal Z
=
\{Z_{i,d}\}_{(i,d)\in\mathcal I}
$
is said to be \emph{weakly stationary} if its mean $\mu_{\mathcal Z}$ and variance $\sigma_{\mathcal Z}^2$ are constant over $\mathcal I$, and its autocorrelation depends only on the lag $(l,m)\in\mathcal L$.
A weakly stationary field is said to be \emph{white} if $\mu_{\mathcal Z}=0$ and its normalized autocorrelation
$
\mathbf C[\mathcal Z]
=
\{c_{l,m}[\mathcal Z]\}
$
satisfies
\begin{equation}
\label{eq:autocorr}
c_{l,m}[\mathcal Z]
:=
\frac{
\operatorname{Cov}(Z_{i,d},Z_{i+l,d+m})
}{
\sigma_{\mathcal Z}^2
}
=
\begin{cases}
1,
& \text{if } (l,m)=(0,0),\\
0,
& \text{otherwise},
\end{cases}
\end{equation}
where
$
\operatorname{Cov}(U,V)
:=
\mathbb E[(U-\mathbb E[U])(V-\mathbb E[V])].
$
\end{definition}

In the Poisson setting, the variance of the measurements depends on the underlying intensity itself. Consequently, raw residuals are not directly comparable across pixels or detector channels characterized by different expected photon counts. To remove this signal dependence, we introduce the following standardized version of a Poisson random variable.

\begin{definition}[Standardized Poisson Random Variable]
\label{def:standardized_poisson}
Let $Y \sim \operatorname{Poiss}(\nu)$ with $\nu > 0$. Its \textbf{standardized} 
counterpart is defined as:
\begin{equation}
    Z := \frac{Y - \nu}{\sqrt{\nu}}.
    \label{eq:standardization}
\end{equation}
By construction, $\mathbb{E}[Z] = 0$ and $\operatorname{Var}[Z] = 1$. Consequently, 
any field of independent standardized Poisson variables forms a white random field.
\end{definition}

In our setting, given a reconstruction $\mathbf{x}$, we define the associated standardized residual field
$
\mathbf Z
=
\{z_{i,d}\}
\in
\mathbb R^{N\times D}
$
entrywise as
\begin{equation}
\label{eq:residual_field}
z_{i,d}
:=
\frac{
y_{i,d}-\nu_{i,d}
}{
\sqrt{\nu_{i,d}}
},
\qquad
(i,d)\in\mathcal I.
\end{equation}
When $\mathbf{x}$ coincides with the ground-truth image $\mathbf{x}_{GT}$, the residual field $\mathbf Z$ consists of independent standardized Poisson random variables and therefore forms a white random field by Definition~\ref{def:standardized_poisson}.

To quantify how far a realization of $\mathbf Z$  is from being white and assuming that the residual field is approximately centered, we introduce the empirical normalized autocorrelation
$
\hat{\mathbf C}(\mathbf Z)
=
\{\hat c_{l,m}(\mathbf Z)\},
$
whose components read
\begin{equation}
\label{eq:sample_corr}
\hat c_{l,m}(\mathbf Z)
:=
\frac{1}{\|\mathbf Z\|_F^2}
\sum_{(i,d)\in\mathcal I}
z_{i,d}\,
z_{i+l,d+m},
\end{equation}
where the sum is taken over all admissible index pairs and $\|\cdot\|_F$ denotes the Frobenius norm. Equation~\eqref{eq:sample_corr} represents the empirical counterpart of the normalized autocorrelation introduced in~\eqref{eq:autocorr}.

The overall departure of the residual field from whiteness is quantified through the whiteness functional
$
\mathcal W :
\mathbb R^{N\times D}
\to
\mathbb R_{+},
$
defined as
\begin{equation}
\label{eq:whiteness_def}
\mathcal W(\mathbf Z)
:=
\sum_{(l,m)\in\mathcal L}
\hat c_{l,m}(\mathbf Z)^2.
\end{equation}
Intuitively, $\mathcal W(\mathbf Z)$ measures the amount of residual correlation present across spatial and detector dimensions: the closer the residual field is to an ideal white field, the smaller the value of $\mathcal W$. As shown in~\cite{wp_lanza}, under periodic boundary conditions the whiteness functional can be evaluated efficiently in
$
\mathcal O(ND\log(ND))
$
operations using the two-dimensional Discrete Fourier Transform (DFT). Namely, denoting by
$
\widetilde{\mathbf Z}
=
\{\tilde z_{i,d}\}
$
the 2D DFT of $\mathbf Z$, one indeed obtains the following handy expression of $\mathcal{W}$
\begin{equation}
\label{eq:whiteness_fft}
\mathcal W(\mathbf Z)
=
\frac{
\displaystyle
\sum_{(i,d)\in\mathcal I}
|\tilde z_{i,d}|^4
}{
\left(
\displaystyle
\sum_{(i,d)\in\mathcal I}
|\tilde z_{i,d}|^2
\right)^2
},
\end{equation}
which can therefore be evaluated on the residual fields associated with reconstructions computed for different regularization parameters $\lambda$, allowing for fully automatic parameter selection through the minimization of residual correlations. More precisely, for a given regularization parameter $\lambda > 0$, let
$\mathbf{x}^\lambda \in \mathbb{R}_+^N$
denote the reconstruction obtained at convergence by one of the proposed algorithms, and let
$\boldsymbol{\nu}^\lambda \in \mathbb{R}^{N \times D}$
be the corresponding estimated Poisson mean, with entries
\begin{equation}
    \nu_{i,d}^\lambda
    :=
    \left[
    \mathbf{A}_d \mathbf{x}^\lambda + \mathbf{b}_d
    \right]_i,
    \qquad
    (i,d)\in\mathcal I.
\end{equation}

The associated standardized residual field
$
\mathbf Z^\lambda
=
\{z_{i,d}^\lambda\}
$
is then defined as
\begin{equation}
    z_{i,d}^\lambda
    :=
    \frac{
    y_{i,d}-\nu_{i,d}^\lambda
    }{
    \sqrt{\nu_{i,d}^\lambda}
    },
    \qquad
    (i,d)\in\mathcal I.
\end{equation}

The optimal regularization parameter is selected by solving
\begin{equation}
\label{masked_wp}
    \lambda^*
    \in
    \argmin_{\lambda>0}
    \mathcal W(\mathbf Z^\lambda).
\end{equation}

Intuitively, when the regularization parameter is properly chosen, the reconstruction
$\mathbf x^\lambda$
is close to the ground-truth image
$\mathbf x_{GT}$,
the residual field
$\mathbf Z^\lambda$
behaves approximately as a white random field, and consequently
$\mathcal W(\mathbf Z^\lambda)$
is minimized.

\begin{remark}
The interpretation of the residual field $\mathbf{Z}^\lambda$ as an approximately white random field relies on the assumption that $\mathbb{E}[\mathbf{Z}^\lambda] \approx \mathbf{0}$. Indeed, the normalized autocorrelation estimator introduced in \eqref{eq:sample_corr} is meaningful under the assumption of approximately centered residuals. This condition is expected to hold whenever the reconstruction $\mathbf{x}^\lambda$ is sufficiently close to the ground-truth image $\mathbf{x}_{GT}$.  In practice, we empirically verify this assumption through a simulation utilizing the PGD-TV algorithm. Specifically, we generate 50 independent noise realizations of the simulated 2D ISM dataset and compute the regularized reconstructions across a logarithmically spaced grid of regularization parameters. By tracking the spatial mean of the corresponding standardized residual field $\mathbf{Z}^\lambda$ for each realization, we confirm that the average residual mean remains negligibly small across the tested parameter range. Crucially, the deviation from zero is virtually non-existent for the smaller parameter values typically selected by the whiteness principle, thereby supporting the validity and robustness of the proposed parameter selection strategy in our setting (see Fig \ref{fig:z_mean_pgdtv}).

\begin{figure}
    \centering

    \includegraphics[width=0.4\linewidth]{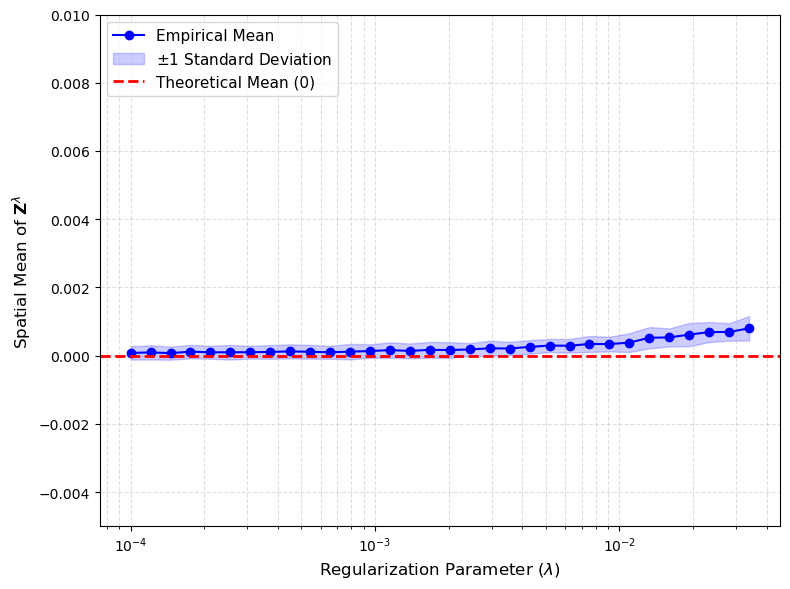} 
    \caption{Empirical validation of the zero-mean residual assumption for the PGD-TV algorithm. The solid blue line tracks the spatial mean of the standardized residual field $\mathbf{Z}^\lambda$ as a function of the regularization parameter $\lambda$, averaged across 50 independent noise realizations of the simulated 2D ISM dataset. The shaded light-blue region represents $\pm 1$ standard deviation, while the red dashed line indicates the ideal theoretical zero value. The empirical mean remains bounded and negligibly small, particularly in the lower $\lambda$ regime where optimal parameters are typically selected.}
    \label{fig:z_mean_pgdtv}
\end{figure}

\end{remark}

At this stage, the whiteness functional provides a criterion for selecting the regularization parameter. In principle, this can be approached in two different ways: either by evaluating $\mathcal W(\mathbf Z^\lambda)$ over a discrete set of candidate parameters and selecting the minimizer through a grid-search strategy, or by directly optimizing with respect to $\lambda$ within a bilevel variational framework \cite{SantambrogioWP,PragliolaBilevelWP}, where the reconstruction problem acts as a lower-level constraint. In this work, we adopt the former approach for simplicity and computational convenience.

\subsection{Masking}
\label{sec:masked_wp}

While effective, the procedure described above does not account for the fact that, in fluorescence microscopy applications, a significant fraction of pixels may record zero photon counts, particularly in low-flux regimes. These zero-valued pixels introduce a systematic bias into the whiteness functional~\eqref{eq:whiteness_def}, since they carry little information about the underlying noise correlations; see~\cite{carl_blanc_f,Bevilacqua_2023}. To address this issue, we adopt a \emph{masked} variant of the Residual Whiteness Principle, originally proposed in~\cite{Bevilacqua_2023}, restricting the residual analysis to the subset of pixels where at least one photon has been detected. Here, we adapt this approach to the multi-frame and dual-plane ISM setting.

The active domain used for the residual statistics is defined by the mask
$
    \mathcal M
    :=
    \{
    (i,d)\in\mathcal I
    \mid
    y_{i,d}>0
    \}.
$

Conditioning on $y_{i,d}>0$ modifies the underlying distribution of the observations: the relevant statistical model is no longer a standard Poisson distribution, but rather a \emph{zero-truncated Poisson} distribution. Its conditional mean $\mu_{i,d}^+$ and variance $(\sigma_{i,d}^+)^2$ are given by
\begin{equation}
    \mu_{i,d}^{+}
    :=
    \frac{\nu_{i,d}^\lambda}
    {1-e^{-\nu_{i,d}^\lambda}},
    \qquad
    (\sigma_{i,d}^{+})^2
    :=
    \mu_{i,d}^{+}
    \left(
    1-\mu_{i,d}^{+}e^{-\nu_{i,d}^\lambda}
    \right),
\end{equation}
where
$
\nu_{i,d}^\lambda
=
[\mathbf A_d\mathbf x^\lambda+\mathbf b_d]_i
$
denotes the estimated Poisson mean; see Appendix~\ref{masked_wh_proof} for the derivation of these quantities. The corresponding masked standardized residual field
$
\mathbf Z^{+}
=
\{z_{i,d}^{+}\}
\in
\mathbb R^{N\times D}
$
is then defined as
\begin{equation}
    z_{i,d}^{+}
    :=
    \begin{cases}
        \dfrac{
        y_{i,d}-\mu_{i,d}^{+}
        }{
        \sigma_{i,d}^{+}
        },
        & \text{if } (i,d)\in\mathcal M,
        \\[8pt]
        0,
        & \text{otherwise.}
    \end{cases}
\end{equation}

The optimal regularization parameter is then selected by minimizing the whiteness functional associated with the masked residual field:
\begin{equation}
    \lambda^*
    \in
    \argmin_{\lambda>0}
    \mathcal W(\mathbf Z^{+}_\lambda),
\end{equation}
where the dependence on $\lambda$ has been explicitly reinstated for clarity.

\subsubsection{High-pass filtering for $s^2$ISM} \label{sec:wp_highpass}

We now discuss the adaptation of the masked Residual Whiteness Principle to the $s^2$ISM model~\eqref{s2ISM}. In this setting, the Poisson mean associated with detector frame $d$ is given by
\begin{equation}
    \nu_{i,d}^\lambda
    :=
    \left[
    \mathbf A_{d,1}\mathbf x_1
    +
    \mathbf A_{d,2}\mathbf x_2
    +
    \mathbf b_d
    \right]_i,
\end{equation}
where $\mathbf x_1$ denotes the foreground component, explicitly regularized through the variational model, while $\mathbf x_2$ represents the background component, constrained only by non-negativity. Since the operators $\mathbf A_{d,2}$ are inherently low-pass and the component $\mathbf x_2$ is not explicitly regularized, inaccuracies in the estimated background may introduce structured low-frequency correlations into the residual field $\mathbf Z^+$. As a consequence, the whiteness functional may become dominated by low-frequency modeling errors unrelated to the regularization quality of the foreground reconstruction $\mathbf x_1$. To mitigate this effect, we evaluate the whiteness functional only on the high-frequency components of the residual field. Let
$
P(\boldsymbol\omega)
:=
|\mathcal F(\mathbf Z^+)(\boldsymbol\omega)|^2
$
denote the power spectrum of $\mathbf Z^+$, where $\mathcal F$ is the 2D DFT and $\boldsymbol\omega$ indexes the frequency components. We introduce a high-pass spectral mask by suppressing a low-frequency region
$
\mathcal B
$
centered at $\boldsymbol\omega=\mathbf 0$, defining the filtered power spectrum
\begin{equation}
    \widetilde P(\boldsymbol\omega)
    :=
    \begin{cases}
        0,
        & \text{if } \boldsymbol\omega\in\mathcal B,
        \\
        P(\boldsymbol\omega),
        & \text{otherwise.}
    \end{cases}
\end{equation}

By thus defining
$
P_\lambda(\boldsymbol\omega)
:=
|\mathcal F(\mathbf Z_\lambda^+)(\boldsymbol\omega)|^2
$
to denote the power spectrum of the masked residual field, and letting
$
\widetilde P_\lambda(\boldsymbol\omega)
$
to be its high-pass filtered counterpart, we can thus define the whiteness functional valuated using the filtered spectrum as
\begin{equation}
    \mathcal W(\mathbf Z_\lambda^+)
    :=
    \frac{
    \displaystyle
    \sum_{\boldsymbol\omega}
    \widetilde P_\lambda(\boldsymbol\omega)^2
    }{
    \left(
    \displaystyle
    \sum_{\boldsymbol\omega}
    \widetilde P_\lambda(\boldsymbol\omega)
    \right)^2
    }.
\end{equation}
By such procedure, the selection of $\lambda^*$ becomes primarily driven by the high-frequency structure of the foreground reconstruction and the associated noise statistics, while remaining substantially insensitive to low-frequency modeling errors arising from the unregularized background component $\mathbf x_2$.

Since the structured low-frequency correlations in $\mathbf Z^+$ arise from the low-pass action of $\mathbf A_{d,2}$ on the unregularized component $\mathbf x_2$, we set the extent of $\mathcal B$ according to the frequency support of the corresponding out-of-focus operator $\mathbf{A}_{d,2}$. Let $h_{d,2}$ denote the out-of-focus PSF of channel $d$, and let $\widehat h_{d,2} = \mathcal F(h_{d,2})$ be its optical transfer function (OTF). We take the cutoff radius $\omega_c$ to be a fraction of the spectral support $\omega_{\mathrm{supp}}(\widehat h_{d,2})$ of $\widehat h_{d,2}$, in particular:
\begin{equation}
    \mathcal B
    :=
    \{\,\boldsymbol\omega \;:\; \|\boldsymbol\omega\| < \omega_c\,\},
    \qquad
    \omega_c := \rho\,\omega_{\mathrm{supp}}(\widehat h_{d,2}),
\end{equation}
where 
$\rho \in (0,1]$ controls the extent of the suppressed band. In all our experiments we set $\rho = 0.5$ and select the central detector $d = 12$ for choosing such cut-off, which appeared a stable choice.

\section{Optimization Algorithms}
\label{sec:algorithms}

To solve the variational reconstruction problems~\eqref{Map_min} and~\eqref{eq:prob_s2ISM}, we consider two representative first-order optimization strategies: Proximal Gradient Descent (PGD)~\cite{proximal_gd} and Mirror Descent (MD)~\cite{nemirovskii_yudin_1983}. The objective function combines the multi-frame KL  data fidelity term, a regularization functional, and a non-negativity constraint on the reconstructed signal. The structure of the optimization scheme therefore depends on the smoothness properties of the chosen regularization model discussed in Section~\ref{sec:reg}. For simplicity, we focus primarily on the single-plane ISM setting and discuss the modifications required for the $s^2$ISM case whenever relevant.

\subsection{Proximal Gradient Descent}
\label{sec:PGD}

Proximal Gradient Descent (PGD) is a widely adopted algorithm in the field of signal/image-processing~\cite{proximal_gd,beck}, particularly effective 
for solving composite optimization problems where the objective function splits into a $L$-smooth (see Definition \ref{L-smothness}) and a non-smooth component. Applied to problem~\eqref{Map_min}, the PGD iteration reads:
\begin{equation}
\label{eq:PGD}
\tag{PGD}
    \x^{k+1} = \operatorname{prox}_{\alpha_k (\iota_{\geq 0} + \lambda R)} 
    \left( \x^k - \alpha_k \nabla \Phi^{\text{ISM}} (\x^k; \underline{\mathbf{A}},
    \underline{\mathbf{y}}, \underline{\mathbf b}) \right) 
    = \Pi_{\geq 0} \left\{ \operatorname{prox}_{\alpha_k \lambda R} 
    \left( \x^k - \alpha_k \nabla \Phi^{\text{ISM}} (\x^k; \underline{\mathbf{A}},
    \underline{\mathbf{y}}, \underline{\mathbf{b}}) \right) \right\},
\end{equation}
where $\Pi_{\geq 0}$ denotes the projection onto the non-negative orthant, enforcing the 
constraint encoded by $\iota_{\geq 0}$, $\operatorname{prox}_{\alpha_k \lambda R}$ is the 
proximal operator associated with the regularization term $R$, and $\alpha_k > 0$ is the stepsize. Note that the equivalence in~\eqref{eq:PGD} does not always hold: the proximal operator of a sum of two functions does not generally reduce to the composition of their individual proximal operators. However, 
when one of the two functions is an indicator function, the decomposition is valid by virtue of 
Proposition~II.2 in~\cite{sum_of_prox}.

When the regularization functional $R$ is differentiable with Lipschitz-continuous gradient, computing its proximal operator is unnecessary. Instead, one can directly incorporate the 
gradient of $R$ into the forward step, reducing the scheme to a standard Projected Gradient 
Descent iteration where the only non-smooth operation is the non-negativity projection:
\begin{equation}
\label{eq:ProjGD}
\tag{ProjGD}
    \x^{k+1} = \Pi_{\geq 0}\left(\x^k - \alpha_k \nabla \left( 
    \Phi^{\text{ISM}} (\x^k; \underline{\mathbf{A}},\underline{\mathbf{y}}, \underline{\mathbf{b}}) 
    + \lambda R (\x^k)\right)\right).
\end{equation}

Convergence of both schemes is guaranteed provided the stepsize satisfies $\alpha_k\in(0,2/L)$, where $L$ is the Lipschitz constant of the gradient of the smooth part of the objective. When $R$ is smooth, this smooth part includes the full objective $\Phi^{ISM}(\x; \underline{\mathbf{A}}, \underline{\mathbf{y}},{\underline{\mathbf{b}}}) + \lambda R(\x)$; when $R$ is non-smooth and handled via its proximal operator, only the data fidelity term $\Phi_{\mathbf{b}}^{ISM}(\x; \underline{\mathbf{A}}, \underline{\mathbf{y}})$ contributes to the smooth part, so that $L=L_{\Phi^{\text{ISM}}}$. Since the Lipschitz constant is typically large, leading to prohibitively small stepsizes, we adopt an adaptive backtracking strategy that selects $\alpha_k$ automatically at each iteration $k$ (see Section~\ref{backtracking}). A detailed analysis of the convergence properties of both~\eqref{eq:PGD} and~\eqref{eq:ProjGD} in the specific ISM setting is provided in Appendix~\ref{app: app_conv}.

\paragraph{The case of \text{$s^2$ISM}.}
In the case of problem~\eqref{eq:prob_s2ISM}, minimization is performed over the joint variable
$
\underline{\x}
=
[\x_1,\x_2],
$
with regularization acting only on the foreground component $\x_1$. Defining the gradient step
\[
\underline{\mathbf z}^k
:=
\underline{\x}^k
-
\alpha_k
\nabla
\Phi^{s^2\text{ISM}}
(
\underline{\x}^k;
[\underline{\mathbf A}_1,\underline{\mathbf A}_2],
\underline{\mathbf y},
\underline{\mathbf b}
),
\]
the proximal update decomposes as
$
\underline{\x}^{k+1}
=
\left[
\operatorname{prox}_{\alpha_k\lambda R}(\mathbf z_1^k),
\,
\mathbf z_2^k
\right].
$

When $R$ is the $\ell_1$ regularizer~\eqref{eq:l1_norm}, the proximal operator reduces to the soft-thresholding mapping,
$
\operatorname{prox}_{\alpha_k\lambda R}(\x)
=
\mathcal T_{\lambda\alpha_k}(\x),
$
and the resulting method corresponds to the classical Iterative Shrinkage-Thresholding Algorithm (ISTA), summarized in Algorithm~\ref{alg:ISTA}.  When $R$ is the smoothed Total Variation functional~\eqref{eq:TV_smooth}, its gradient can instead be incorporated directly into the smooth component of the objective, yielding the projected gradient scheme~\eqref{eq:ProjGD}, summarized in Algorithm~\ref{alg:PGD-TV}.

\begin{algorithm}
\caption{Iterative Shrinkage-Thresholding Algorithm (ISTA) for ISM and $s^2$ISM}
\label{alg:ISTA}
\begin{algorithmic}[1]
\Require Modality $\mathsf{M}\in\{\mathrm{ISM},s^2\mathrm{ISM}\}$, regularization parameter $\lambda>0$

\If{$\mathsf{M}=\mathrm{ISM}$}
    \State Set $\mathbf \x^0=\sum_{d=1}^{25}\mathbf A_d^\top \mathbf y_d \in \mathbb R_+^N$
    \State Set $\Phi=\Phi^{\mathrm{ISM}}(\cdot;\underline{\mathbf A},\underline{\mathbf y},\underline{\mathbf b})$ and $L=L_{\Phi^{\mathrm{ISM}}}$
\ElsIf{$\mathsf{M}=s^2\mathrm{ISM}$}
    \State Set $\underline{\mathbf x}^0=
    \left[
    \sum_{d=1}^{25}\mathbf A_{d,1}^\top\mathbf y_d,\,
    \sum_{d=1}^{25}\mathbf A_{d,2}^\top\mathbf y_d
    \right]\in\mathbb R_+^{N\times2}$
    \State Set $\Phi=\Phi^{s^2\mathrm{ISM}}(\cdot;[\underline{\mathbf A}_1,\underline{\mathbf A}_2],\underline{\mathbf y},\underline{\mathbf b})$ and $L=L_{\Phi^{s^2\mathrm{ISM}}}$
\EndIf

\Require Stepsize $\alpha_k$ chosen with Algorithm \ref{alg:backtracking_adaptive}

\For{$k=0,1,2,\dots$}
    \If{$\mathsf{M}=\mathrm{ISM}$}
        \State $\mathbf{x}^{k+1}
        =
        \Pi_{\geq0}\!\left(
        \mathcal T_{\lambda\alpha_k}
        \left(
        \mathbf{x}^k
        -
        \alpha_k\nabla\Phi(\mathbf{x}^k)
        \right)
        \right)$
    \ElsIf{$\mathsf{M}=s^2\mathrm{ISM}$}
        \State $\underline{\mathbf{x}}^{k+1}
        =
        \left[
        \Pi_{\geq0}\!\left(
        \mathcal T_{\lambda\alpha_k}
        \left(
        \mathbf{x}_1^k
        -
        \alpha_k
        \nabla_{\mathbf{x}_1}\Phi(\underline{\mathbf{x}}^k)
        \right)
        \right),
        \,
        \Pi_{\geq0}\!\left(
        \mathbf{x}_2^k
        -
        \alpha_k
        \nabla_{\mathbf{x}_2}\Phi(\underline{\mathbf{x}}^k)
        \right)
        \right]
        $
    \EndIf
\EndFor
\end{algorithmic}
\end{algorithm}
\begin{algorithm}
\caption{Projected Gradient Descent with smoothed TV for ISM and $s^2$ISM}
\label{alg:PGD-TV}
\begin{algorithmic}[1]

\Require Modality $\mathsf{M}\in\{\mathrm{ISM},s^2\mathrm{ISM}\}$, regularization parameter $\lambda>0$

\If{$\mathsf{M}=\mathrm{ISM}$}
    \State Set $\mathbf x^0=\sum_{d=1}^{25}\mathbf A_d^\top \mathbf y_d \in \mathbb R_+^N$
    \State Set
$
    \Phi
    =
    \Phi^{\mathrm{ISM}}
    (\cdot;\underline{\mathbf A},\underline{\mathbf y},\underline{\mathbf b})
$
\ElsIf{$\mathsf{M}=s^2\mathrm{ISM}$}
    \State Set
$
    \underline{\mathbf x}^0
    =
    \left[
    \sum_{d=1}^{25}\mathbf A_{d,1}^\top\mathbf y_d,\,
    \sum_{d=1}^{25}\mathbf A_{d,2}^\top\mathbf y_d
    \right]
    \in
    \mathbb R_+^{N\times2}
$
    \State Set
$
    \Phi
    =
    \Phi^{s^2\mathrm{ISM}}
    (\cdot;[\underline{\mathbf A}_1,\underline{\mathbf A}_2],\underline{\mathbf y},\underline{\mathbf b})
$
\EndIf

\Require Stepsize $\alpha_k$ chosen with Algorithm \ref{alg:backtracking_adaptive}

\For{$k=0,1,2,\dots$}

    \If{$\mathsf{M}=\mathrm{ISM}$}
        \State
        $
        \mathbf{x}^{k+1}
        =
        \Pi_{\geq0}
        \left(
        \mathbf{x}^k
        -
        \alpha_k
        \nabla
        \left(
        \Phi(\mathbf{x}^k)
        +
        \lambda\operatorname{TV}_\epsilon(\mathbf{x}^k)
        \right)
        \right)
        $

    \ElsIf{$\mathsf{M}=s^2\mathrm{ISM}$}
        \State
        $
        \underline{\mathbf{x}}^{k+1}
        =
        \left[
        \Pi_{\geq0}
        \left(
        \mathbf{x}_1^k
        -
        \alpha_k
        \left(
        \nabla_{\mathbf{x}_1}\Phi(\underline{\mathbf{x}}^k)
        +
        \lambda\nabla\operatorname{TV}_\epsilon(\mathbf{x}_1^k)
        \right)
        \right),
        \,
        \Pi_{\geq0}
        \left(
        \mathbf{x}_2^k
        -
        \alpha_k
        \nabla_{\mathbf{x}_2}\Phi(\underline{\mathbf{x}}^k)
        \right)
        \right]
        $
    \EndIf

\EndFor
\end{algorithmic}
\end{algorithm}

\subsubsection{Mirror Descent}
\label{sec:MD}

Mirror Descent (MD) can be viewed as a generalization of gradient descent to a 
different geometry suited to deal with the potential constraints imposed on the desired solution,
thus avoiding the introduction of an explicit projection step. To get an intuition of how the optimization schemes change in this setting, we consider the case where the regularization functional $R$ is differentiable on the positive orthant. In this case, the Mirror Descent (MD) \cite{nemirovskii_yudin_1983} iteration associated with problem~\eqref{Map_min} reads:
\begin{equation}
\mathbf{x}^{k+1}
\in
\argmin_{\mathbf{x}\geq0}
\left\langle
\nabla\Phi^{\mathrm{ISM}}
(\mathbf{x}^k;\underline{\mathbf A},\underline{\mathbf y},\underline{\mathbf{b}})
+
\lambda\nabla R(\mathbf{x}^k),
\,
\mathbf{x}-\mathbf{x}^k
\right\rangle
+
\frac{1}{\alpha_k}
D_h(\mathbf{x},\mathbf{x}^k), \label{eq:implicit MD for KL}
\end{equation}
where $h$ is the function which takes into account problem constraints (see Definition~\ref{def:Legendre function}) and induces the underlying geometry,  $D_h$ is the associated Bregman divergence generalizing the standard $\ell^2$ distance to such non-Euclidean case (Definition~\ref{Bregman divergence}). Intuitively, the function $h$ acts as a geometric barrier: its gradient diverges near the boundary of its domain, so that the algorithmic iterates are naturally constrained to remain inside the feasible set. 
Motivated by previous work \cite{Bolte2018}, for the particular constraint set considered in this work we consider the logarithmic barrier (so-called Burg's entropy):
\begin{equation}
    h(\mathbf{x})
=
-\sum_{i=1}^{N}\log(x_i), \label{eq: burg's entropy}
\end{equation}
which automatically enforces positivity of the iterates and eliminates the need for an explicit projection onto the non-negative orthant. Under this choice, the implicit update~\eqref{eq:implicit MD for KL} admits the multiplicative form~\cite{nolip}
\begin{equation}
\label{eq:explicit_md_step}
\mathbf{x}^{k+1}
=
\frac{
\mathbf{x}^{k}
}{
\mathbf 1
+
\alpha_k\,\mathbf{x}^k
\cdot
\left(
\nabla\Phi^{\mathrm{ISM}}
(\mathbf{x}^k;\underline{\mathbf A},\underline{\mathbf y},\underline{\mathbf{b}})
+
\lambda\nabla R(\mathbf{x}^k)
\right)
},
\end{equation}
where all operations are understood component-wise. The resulting iteration is multiplicative in nature and closely resembles the Richardson--Lucy scheme~\eqref{eq: Richardson-Lucy}.

Convergence of \eqref{eq:implicit MD for KL} is guaranteed by choosing a stepsize adapted to the geometry induced by $h$, typically proportional to the inverse of a generalized Lipschitz smoothness constant satisfying a suitable  condition (Definition~\ref{NoLip}). Such constant  plays an analogous role to the Lipschitz constant $L$ in PGD. As in~\eqref{eq:PGD}, large values of such generalized $L$ may lead to prohibitively small stepsizes, hence we adopt an adaptive backtracking strategy that selects $\alpha_k$ automatically at each iteration $k$ (see Section~\ref{backtracking}). Additional convergence details are provided in Appendix~\ref{app: app_conv}.

Note that when $R$ is the $\ell_1$ regularizer~\eqref{eq:l1_norm}, the positivity constraint induced by the logarithmic barrier makes $R$ differentiable on the feasible set. Consequently, the regularization term can be incorporated directly into the forward step without requiring any proximal evaluation, leading to the scheme reported in Algorithm~\ref{alg:Mirror-L1-Combined}.  When $R$ is the smoothed Total Variation functional~\eqref{eq:TV_smooth}, the corresponding Mirror Descent scheme is summarized in Algorithm~\ref{alg:Mirror-TV-Combined}.


\begin{algorithm}
\caption{Mirror Descent with $\ell_1$ regularization for ISM and $s^2$ISM}
\label{alg:Mirror-L1-Combined}
\begin{algorithmic}[1]

\Require Modality $\mathsf{M}\in\{\mathrm{ISM},s^2\mathrm{ISM}\}$, regularization parameter $\lambda>0$

\If{$\mathsf{M}=\mathrm{ISM}$}
    \Require Initial point $\mathbf{x}^0=\sum_{d=1}^{25}\mathbf A_d^\top\mathbf y_d\in\mathbb R^N_{>0}$, stepsize $\alpha_k$ chosen with (\ref{alg:backtracking_adaptive})
    \State Set $\Phi=\Phi^{\mathrm{ISM}}(\cdot;\underline{\mathbf A},\underline{\mathbf y},\underline{\mathbf b})$
\ElsIf{$\mathsf{M}=s^2\mathrm{ISM}$}
    \Require Initial point $\underline{\mathbf x}^0=
    \left[
    \sum_{d=1}^{25}\mathbf A_{d,1}^\top\mathbf y_d,\,
    \sum_{d=1}^{25}\mathbf A_{d,2}^\top\mathbf y_d
    \right]
    \in
    \mathbb R^{N\times2}_{>0}$,
    stepsize $\alpha_k$ chosen with Algorithm \ref{alg:backtracking_adaptive}
    \State Set $\Phi=\Phi^{s^2\mathrm{ISM}}(\cdot;[\underline{\mathbf A}_1,\underline{\mathbf A}_2],\underline{\mathbf y},\underline{\mathbf b})$
\EndIf

\For{$k=0,1,2,\dots$}

    \If{$\mathsf{M}=\mathrm{ISM}$}
        \State $\mathbf x^{k+1}=
        \frac{
        \mathbf x^k
        }{
        \mathbf 1
        +
        \alpha_k\,\mathbf x^k\cdot
        \left(
        \nabla\Phi(\mathbf x^k)
        +
        \lambda\mathbf 1
        \right)
        }$
    \ElsIf{$\mathsf{M}=s^2\mathrm{ISM}$}
        \State $\underline{\mathbf x}^{k+1}=
        \left[
        \frac{
        \mathbf x_1^k
        }{
        \mathbf 1
        +
        \alpha_k\,\mathbf x_1^k\cdot
        \left(
        \nabla_{\mathbf x_1}\Phi(\underline{\mathbf x}^k)
        +
        \lambda\mathbf 1
        \right)
        },
        \,
        \dfrac{
        \mathbf x_2^k
        }{
        \mathbf 1
        +
        \alpha_k\,\mathbf x_2^k\cdot
        \nabla_{\mathbf x_2}\Phi(\underline{\mathbf x}^k)
        }
        \right]$
    \EndIf

\EndFor
\end{algorithmic}
\end{algorithm}

\begin{algorithm}
\caption{Mirror Descent with smoothed TV for ISM and $s^2$ISM}
\label{alg:Mirror-TV-Combined}
\begin{algorithmic}[1]

\Require Modality $\mathsf{M}\in\{\mathrm{ISM},s^2\mathrm{ISM}\}$, regularization parameter $\lambda>0$

\If{$\mathsf{M}=\mathrm{ISM}$}
    \State Set
$
    \mathbf{x}^0
    =
    \sum_{d=1}^{25}\mathbf A_d^\top\mathbf y_d
    \in
    \mathbb R^N_{>0}
$
    \State Set
$
    \Phi
    =
    \Phi^{\mathrm{ISM}}
    (\cdot;\underline{\mathbf A},\underline{\mathbf y},\underline{\mathbf b})
$
\ElsIf{$\mathsf{M}=s^2\mathrm{ISM}$}
    \State Set
$
    \underline{\mathbf x}^0
    =
    \left[
    \sum_{d=1}^{25}\mathbf A_{d,1}^\top\mathbf y_d,\,
    \sum_{d=1}^{25}\mathbf A_{d,2}^\top\mathbf y_d
    \right]
    \in
    \mathbb R^{N\times2}_{>0}
$
    \State Set
$
    \Phi
    =
    \Phi^{s^2\mathrm{ISM}}
    (\cdot;[\underline{\mathbf A}_1,\underline{\mathbf A}_2],\underline{\mathbf y},\underline{\mathbf b})
$
\EndIf

\Require Stepsize $\alpha_k$ chosen with Algorithm \ref{alg:backtracking_adaptive}

\For{$k=0,1,2,\dots$}

    \If{$\mathsf{M}=\mathrm{ISM}$}
        \State
 $
        \mathbf x^{k+1}
        =
        \frac{
        \mathbf x^k
        }{
        \mathbf 1
        +
        \alpha_k
        \left(
        \nabla\Phi(\mathbf x^k)
        +
        \lambda\nabla\operatorname{TV}_\epsilon(\mathbf x^k)
        \right)
        }
$
    \ElsIf{$\mathsf{M}=s^2\mathrm{ISM}$}
        \State
$
        \underline{\mathbf x}^{k+1}
        =
        \left[
        \frac{
        \mathbf x_1^k
        }{
        \mathbf 1
        +
        \alpha_k
        \left(
        \nabla_{\mathbf x_1}\Phi(\underline{\mathbf x}^k)
        +
        \lambda
        \nabla\operatorname{TV}_\epsilon(\mathbf x_1^k)
        \right)
        },
        \,
        \frac{
        \mathbf x_2^k
        }{
        \mathbf 1
        +
        \alpha_k
        \nabla_{\mathbf x_2}\Phi(\underline{\mathbf x}^k)
        }
        \right]
$
    \EndIf

\EndFor
\end{algorithmic}
\end{algorithm}

Notice that the initial vector $\x_0$, obtained by initializing with the transpose, always provides a rough first estimate of $\x^*$, meaning we start closer to the solution whether the problem is convex or strictly convex.

\subsection{Stepsize selection via adaptive backtracking}
\label{backtracking}

The algorithms introduced above are presented using fixed stepsizes determined by the inverse of a global smoothness constant, namely the Lipschitz constant $L$ in the Euclidean setting, or its generalized counterpart in the Mirror Descent framework (see Definition~\ref{NoLip}). While such choices guarantee convergence, they are often overly conservative in practice. In the ISM/s$^2$ISM Poisson setting, global bounds typically correspond to worst-case estimates (see, e.g., \eqref{eq:L_ISM} and \eqref{eq:L_ISM_TV}), leading to very small admissible stepsizes and consequently slow numerical convergence.

To improve the practical efficiency of the schemes, we employ adaptive nonmonotone backtracking strategies that automatically adjust the stepsize during the iterations. In the Euclidean setting, such strategies are classical for proximal and projected gradient methods, where they are used to avoid relying on global smoothness estimates while still ensuring sufficient descent of the objective functional. The variants considered here are inspired by the adaptive strategies proposed in~\cite{CalatroniChambolle2019,RebegoldiCalatroni2022} for accelerated first-order schemes. At iteration $k$, a tentative stepsize larger than the previously accepted one is first proposed through an expansion factor $\delta\in(0,1)$. A candidate iterate is then computed and tested against a suitable descent condition. If the condition is satisfied, the stepsize is accepted; otherwise, it is progressively reduced through a contraction factor $\eta\in(0,1)$ until the descent criterion holds.

While such backtracking rules are standard in Euclidean proximal-gradient settings, their extension to the non-Euclidean Mirror Descent geometry is, to the best of our knowledge, less explored theoretically. In this work, we therefore adopt the corresponding MD backtracking rule mainly from a practical perspective, motivated by its favorable empirical behavior.

The resulting adaptive routine is summarized in Algorithm~\ref{alg:backtracking_adaptive}. There, $\mathbf u$ denotes either $\mathbf x$ in the ISM case or $\underline{\mathbf x}$ in the $s^2$ISM case. The functional $\mathcal J$ denotes the full objective minimized by the chosen algorithm, while $f$ denotes its smooth part.

\begin{algorithm}
\caption{Adaptive backtracking (AB)}
\label{alg:backtracking_adaptive}
\begin{algorithmic}[1]
\Require Algorithm $\mathcal A$, objective $\mathcal J$, initial point $\mathbf u^0$, initial stepsize $\alpha^0>0$, expansion factor $\delta\in(0,1)$, contraction factor $\eta\in(0,1)$
\For{$k=0,1,2,\dots$}
    \State $\alpha \gets \alpha^k/\delta$
    \Repeat
        \State Compute $\mathbf u^{\mathrm{cand}}=\mathcal A(\mathbf u^k,\alpha)$
        \If{$\mathcal A$ is PGD / ProjGD}
            \State $\Delta \gets \frac{1}{2\alpha}\|\mathbf u^{\mathrm{cand}}-\mathbf u^k\|_2^2
            +\left\langle \nabla f(\mathbf u^k),\mathbf u^{\mathrm{cand}}-\mathbf u^k\right\rangle$
        \ElsIf{$\mathcal A$ is MD}
            \State $\Delta \gets -\frac{d}{\alpha}D_h(\mathbf u^{\mathrm{cand}},\mathbf u^k)$, with $d>0$
        \EndIf
        \If{$\mathcal J(\mathbf u^{\mathrm{cand}})\leq \mathcal J(\mathbf u^k)+\Delta$}
            \State $\alpha^{k+1}\gets\alpha$
            \State $\mathbf u^{k+1}\gets\mathbf u^{\mathrm{cand}}$
        \Else
            \State $\alpha\gets\eta\alpha$
        \EndIf
    \Until{$\alpha$ is accepted}
\EndFor
\end{algorithmic}
\end{algorithm}

Algorithm~\ref{alg:backtracking_adaptive} adapts the sufficient decrease condition to the underlying optimization geometry: Euclidean for PGD-type schemes and Bregman-based for MD. The same strategy applies directly to the $s^2$ISM setting. 

The empirical behavior of the proposed adaptive backtracking procedure is illustrated in Section \ref{real_results} , where Figure~\ref{fig:backtracking} shows the evolution of the stepsize $\alpha_k$ for the different algorithms. As shown, the adaptive strategy consistently finds stepsizes that are several orders of magnitude larger than the pessimistic global theoretical lower bounds $1/L$, validating the practical efficiency of Algorithm~\ref{alg:backtracking_adaptive}. 

\section{Numerical Results on simulated ISM data}

In this section, we present reconstruction results obtained on simulated aberrated tubulin data. Simulations are performed using the BrightEyes-ISM simulator \cite{zunino2023open}, freely available online\footnote{https://github.com/VicidominiLab/BrightEyes-ISM}.

From now on, for the sake of visualization, images marked with the label ``(sat)'' are displayed with a saturated intensity scale: the upper limit of the colormap is set to the $99.99$-th percentile of the intensity values, rather than to the actual maximum. This saturation suppresses the visual impact of a few extremely bright outlier pixels, which would otherwise compress the dynamic range and hide the finer structures of the reconstruction. The saturation affects only the display and not the quantitative metrics computed.

\subsection{The MID case}

\subsubsection{Simulation setup}
We consider a normalized ground-truth image
$
\tilde{\mathbf x}_{GT}\in[0,1],
$
scaled by a flux factor $F>0$ chosen to produce realistic photon-count levels. The measurements
$
\mathbf y_1,\dots,\mathbf y_{25}
$
are generated according to the forward model
\begin{equation}
    \mathbf y_d
    =
    \operatorname{Poiss}
    \left(
    \mathbf A_d(F\tilde{\mathbf x}_{GT})
    +
    \mathbf b_d
    \right),
    \qquad
    d=1,\dots,25,
\end{equation}
where $\mathbf A_d$ denotes the simulated PSF operator associated with the $d$-th detector element, and $\mathbf b_d$ is a detector-dependent spatially constant background term defined as
\begin{equation}
\label{eq:bg_index}
    \eta_d
    =
    \frac{f_d}{f_{13}},
    \qquad
    b_d
    =
    \beta\eta_d,
    \qquad
    d\in\{1,\dots,25\},
\end{equation}
where $f_d$ denotes the fingerprint of the $d$-th detector, defined as the sum of all entries of the kernel associated with $\mathbf A_d$, and $\beta=10^{-1}$. This choice makes the background level proportional to the average signal intensity collected by each detector element.

For the experiments below, we simulate a $256\times256$ tubulin structure with flux factor
$
F=20,
$
shown in Figure~\ref{fig:gt_tub2d}. The ISM dataset is generated by convolving the ground-truth image with the 25 detector-dependent PSFs shown in Figure~\ref{fig:psf_sim}, adding the corresponding constant background terms, and applying Poisson noise independently to each detector channel. An example of resulting low photon-count noisy measurements is shown in Figure~\ref{fig:noisy_tub}.

\begin{figure}
    \centering
    \begin{subfigure}{0.32\textwidth}
        \centering
        \includegraphics[width=\linewidth]{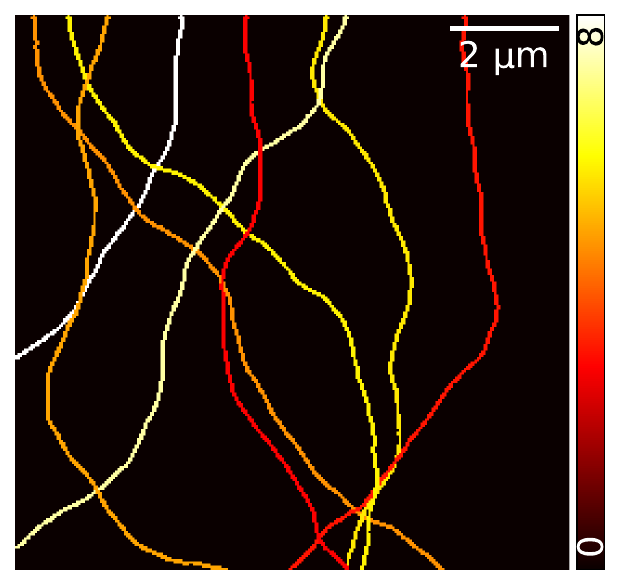}
        \caption{$\tilde{\mathbf{x}}_{GT}$.}
        \label{fig:gt_tub2d}
    \end{subfigure}
    \hfill
    \begin{subfigure}{0.32\textwidth}
        \centering
        \includegraphics[width=0.98\linewidth]{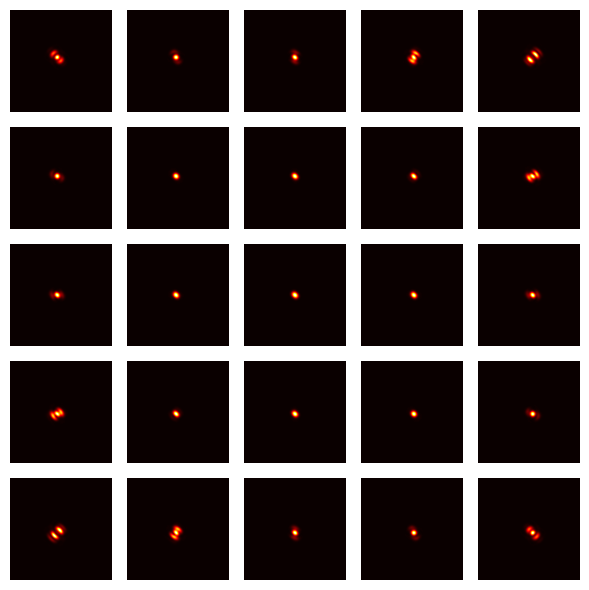}
        \caption{Simulated ISM PSFs.}
        \label{fig:psf_sim}
    \end{subfigure}
    \hfill
    \begin{subfigure}{0.32\textwidth}
    \centering
    \includegraphics[width=\linewidth]{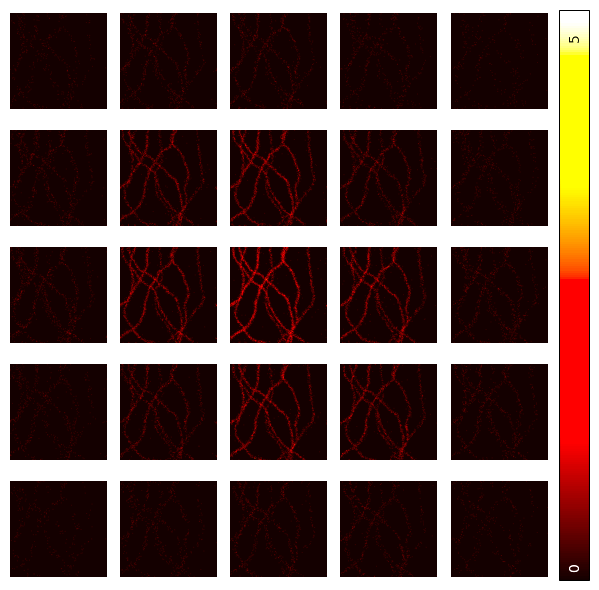}
    \caption{Simulated $\mathbf{y}_1,\ldots,\mathbf{y}_{25}$.}
            \label{fig:noisy_tub}
\end{subfigure}
    \caption{Simulation setup: (a) Ground truth tubulin structure. (b) Simulated ISM PSFs. For visualization purposes, each PSF is individually normalized rather than globally normalized across the dataset. (c) Simulated low-photon count ISM measurements.}
\end{figure}

\subsubsection{Evaluation protocol}

Reconstruction quality is assessed using two standard metrics: the Peak Signal-to-Noise Ratio (PSNR) and the Improved Structural Similarity for Comparing Microscopy Data (MicroSSIM)~\cite{microssim}, a similarity metric with values in $[0,1]$ specifically designed for microscopy applications. For consistency, all reconstructed images are rescaled to $[0,1]$ before comparison with $\tilde{\mathbf{x}}_{GT}$. Convergence is declared when the relative change between consecutive iterates satisfies
\begin{equation}
\label{conv_crit}
    \frac{\|\mathbf{x}^{k+1} - \mathbf{x}^k\|}
    {\|\mathbf{x}^{k+1}\|}
    \leq
    \texttt{tol},
    \qquad
    \texttt{tol}=10^{-5}.
\end{equation}
We compute reconstructions by solving problem~\eqref{Map_min} using Algorithms~\ref{alg:ISTA}, \ref{alg:PGD-TV}, \ref{alg:Mirror-L1-Combined}, and~\ref{alg:Mirror-TV-Combined}. Since the objective functional is convex but not strictly convex, minimizers are not necessarily unique. Consequently, different optimization schemes may converge to different numerical minimizers even under the same regularization model, motivating the comparison between both PGD- and MD-based solvers.

The regularization parameter $\lambda$ is selected independently for each reconstruction algorithm using the RWP strategy described in Section~\ref{secRPW}. Indeed, the residual whiteness functional depends implicitly on the reconstruction $\mathbf{x}^\lambda$, which in turn depends on the optimization dynamics and on the specific minimizer reached by the corresponding algorithm.
 For each method, we evaluate the whiteness functional $\mathcal W(\lambda)$ over a logarithmically distributed grid of $150$ candidate parameters:
\begin{equation}
    \Lambda
    =
    \{\lambda_i\}_{i=1}^{150},
    \qquad
    \lambda_i\in[\lambda_{\text{min}},\lambda_{\text{max}}].
\end{equation}
For every $\lambda\in\Lambda$, the corresponding algorithm is run until numerical convergence and the whiteness metric is evaluated on the masked residual field as described in Section~\ref{secRPW}. The optimal parameter $\lambda^*$ is then selected by minimizing $\mathcal W(\lambda)$.

Finally, visual reconstruction results and quantitative scores are reported in Figure~\ref{fig:sim_result_2d}. TV regularization provides the best overall reconstruction quality, effectively suppressing the point-like artifacts appearing in the $\ell_1$ reconstructions and achieving the highest quantitative scores. We also observe that RL and MD-$\ell_1$ reconstructions exhibit values outside the expected measurement range, reflecting the multiplicative nature of both schemes.

\begin{figure}
    \centering

    \newcommand{\methodhead}[3]{%
        \shortstack{\footnotesize #1\\[1pt]
                    \scriptsize PSNR=#2\\[-1pt]
                    \scriptsize MICROSSIM=#3}}
    \setlength{\tabcolsep}{1.5pt}
    \renewcommand{\arraystretch}{0.3}
    \newcommand{\tw}{0.225\linewidth} 

    \begin{tabular}{@{}c cccc@{}}
        & \methodhead{Noisy sum}{16.73}{0.221}
        & \methodhead{RL ($k=5$)}{16.14}{0.71}
        & \methodhead{RL ($k=10000$) (sat)}{15.95}{0.76}
        & \methodhead{ProjGD ($\lambda=0$)}{16.40}{0.74} \\[3pt]

        \rotatebox{90}{\scriptsize Full} &
        \includegraphics[width=\tw]{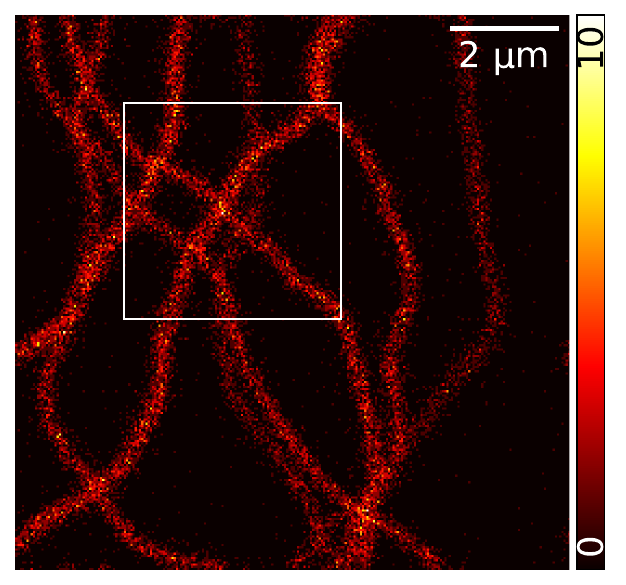} &
        \includegraphics[width=\tw]{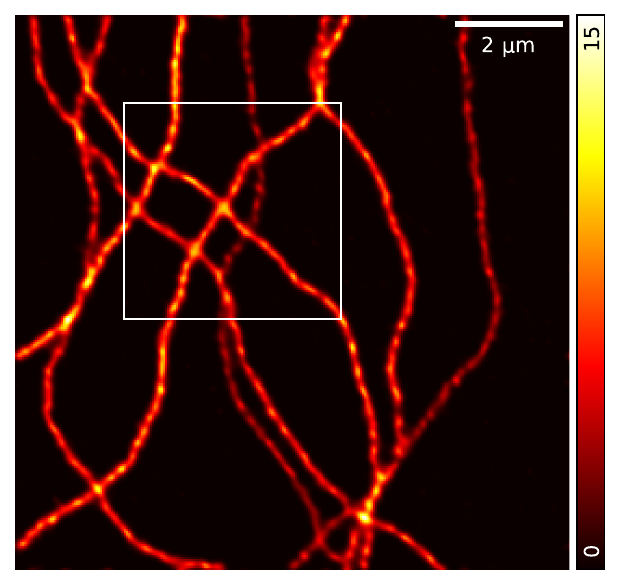} &
        \includegraphics[width=\tw]{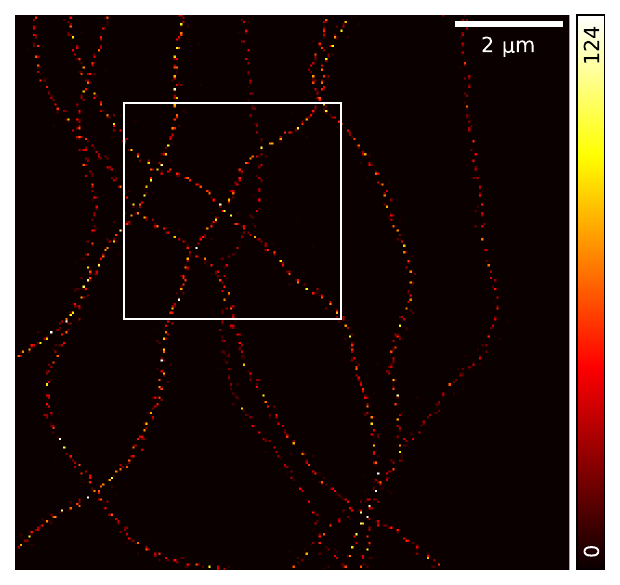} &
        \includegraphics[width=\tw]{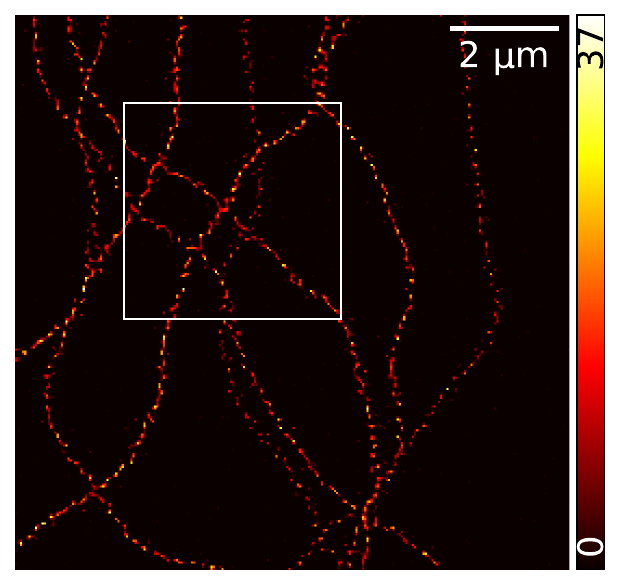} \\

        \rotatebox{90}{\scriptsize Crop} &
        \includegraphics[width=\tw]{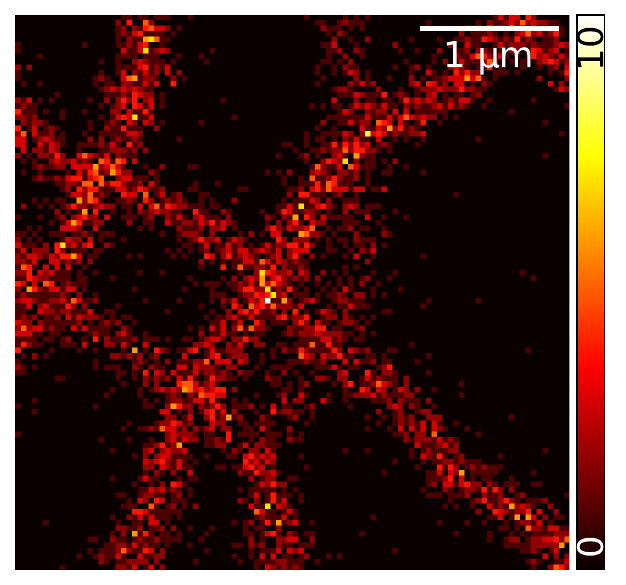} &
        \includegraphics[width=\tw]{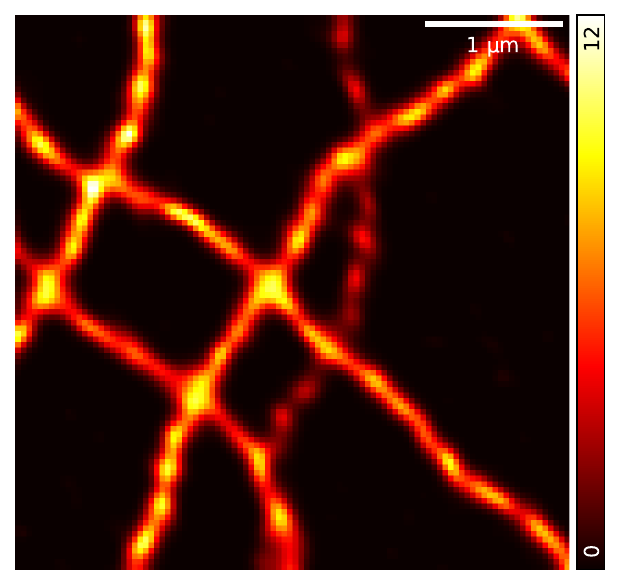} &
        \includegraphics[width=\tw]{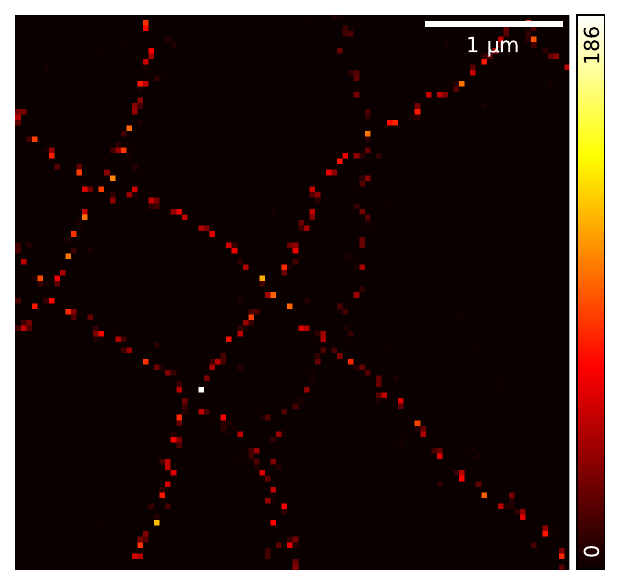} &
        \includegraphics[width=\tw]{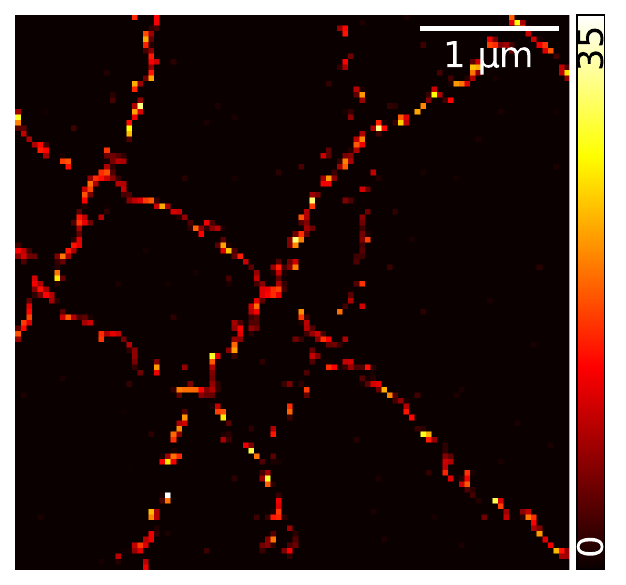} \\
    \end{tabular}

    \vspace{0.3cm}

    \begin{tabular}{@{}c cccc@{}}
        & \methodhead{ProxGD-$\ell_1$ ($\lambda=5.37\!\times\!10^{-2}$)}{17.61}{0.78}
        & \methodhead{MD-$\ell_1$ ($\lambda=8.05\!\times\!10^{-2}$) (sat)}{15.98}{0.822}
        & \methodhead{PGD-TV ($\lambda=6.72\!\times\!10^{-3}$)}{19.40}{0.84}
        & \methodhead{MD-TV ($\lambda=8.05\!\times\!10^{-2}$) (sat)}{17.45}{0.81} \\[3pt]

        \rotatebox{90}{\scriptsize Full} &
        \includegraphics[width=\tw]{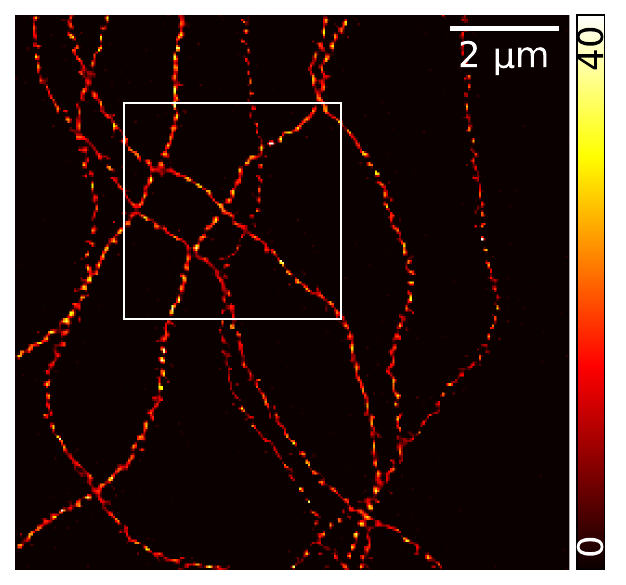} &
        \includegraphics[width=\tw]{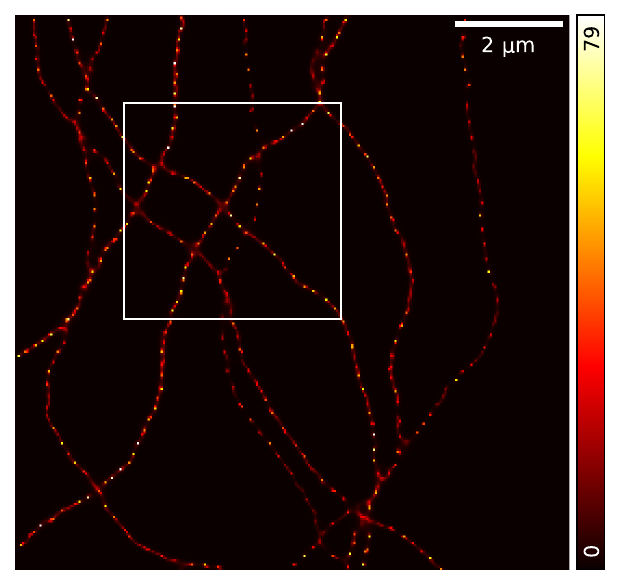} &
        \includegraphics[width=\tw]{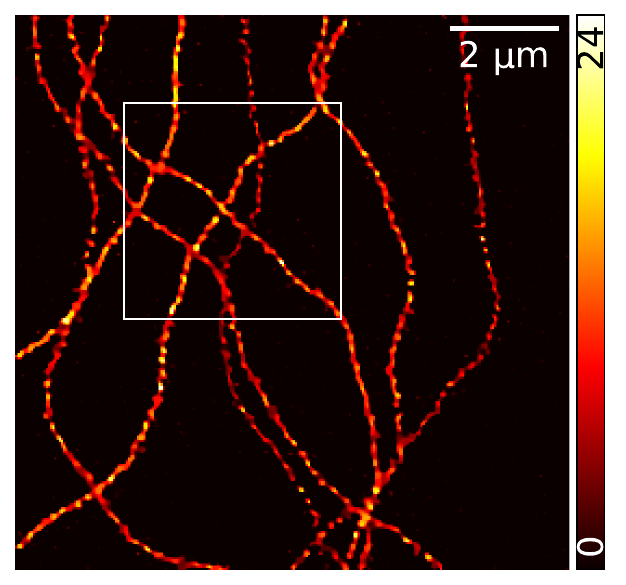} &
        \includegraphics[width=\tw]{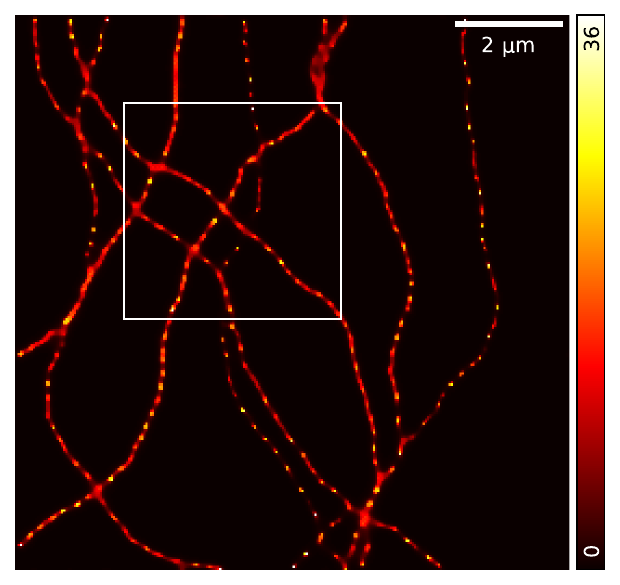} \\

        \rotatebox{90}{\scriptsize Crop} &
        \includegraphics[width=\tw]{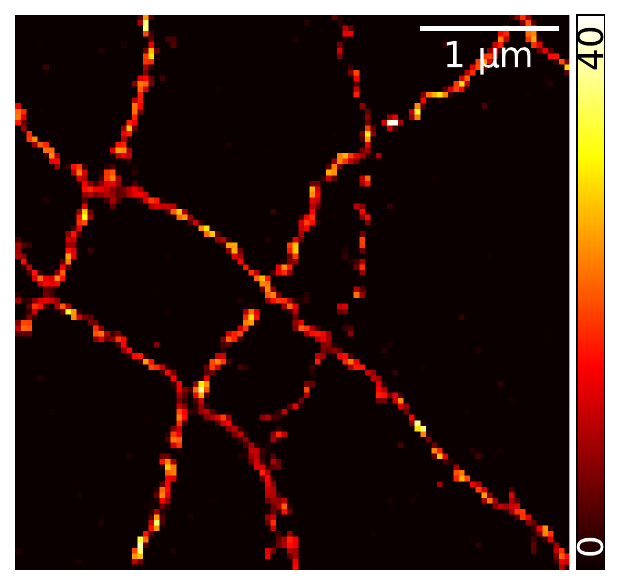} &
        \includegraphics[width=\tw]{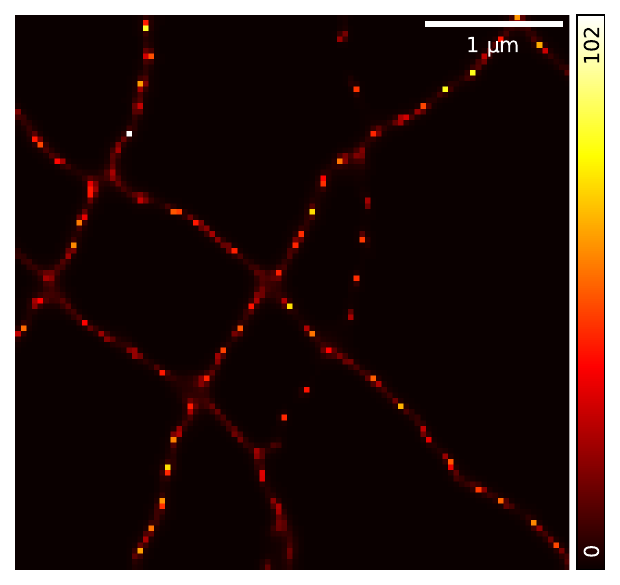} &
        \includegraphics[width=\tw]{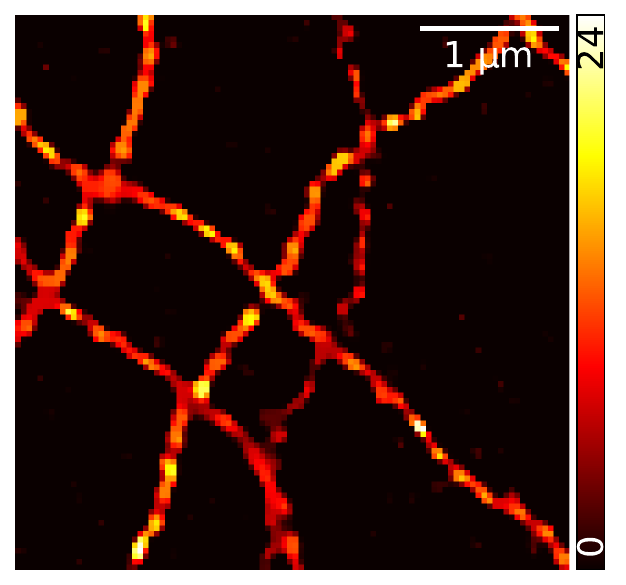} &
        \includegraphics[width=\tw]{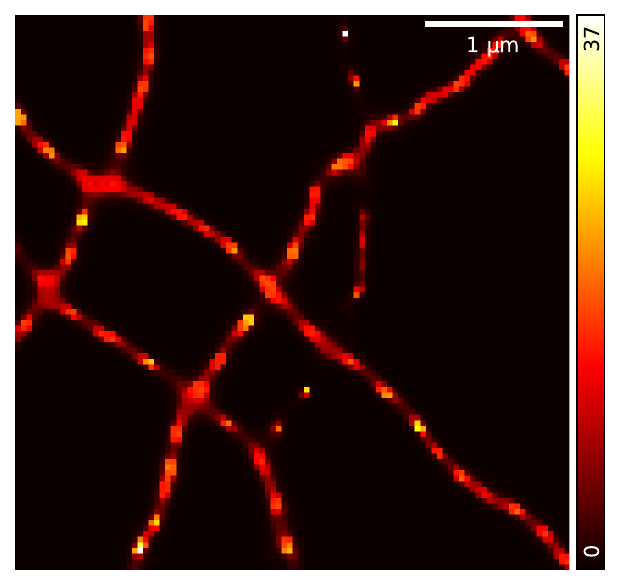} \\
    \end{tabular}

    \caption{Simulated tubulin reconstruction results. Rows 1--2: noisy sum $\bar{\y} = \sum_{d=1}^{25}\y_d$ and 
    Richardson--Lucy at two different stopping iterations, alongside the unregularized PGD 
    ($\lambda=0$) baseline. Rows 3--4: regularized reconstructions obtained with 
    PGD-TV, MD-TV, ProxGD-$\ell_1$, and MD-$\ell_1$, with optimal values $\lambda^*$ selected via the RWP criterion.}
    \label{fig:sim_result_2d}
\end{figure}

\subsection{The $s^2$ISM case}

We now evaluate the proposed reconstruction framework in a dual-plane imaging scenario, considering the $s^2$ISM acquisition model introduced in Section~\ref{sec: mathematical modeling for $s^2ISM$}. The goal of this setting is to recover the in-focus signal while suppressing out-of-focus fluorescence contributions. To simulate this acquisition geometry, we construct a composite scene consisting of two distinct tubulin structures: a foreground (in-focus) component, denoted by $\x_{1,\mathrm{GT}}$, and a background (out-of-focus) component, denoted by $\x_{2,\mathrm{GT}}$. The two ground-truth images, shown in Figure~\ref{fig:gt_3d_tubulin}, are independently scaled using the same flux-scaling procedure adopted in the single-plane experiments, with flux factor $F=20$.

\begin{figure}
  \centering
  \begin{subfigure}{0.32\linewidth}
    \centering
    \includegraphics[width=\linewidth]{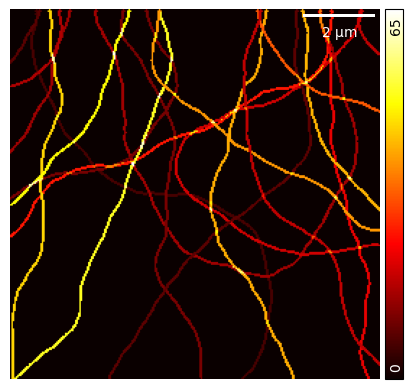} 
    \caption{In-focus ground truth $\x_{1,\mathrm{GT}}$}
    \label{fig:tub_forw}
  \end{subfigure}\hfill
  \begin{subfigure}{0.32\linewidth}
    \centering
    \includegraphics[width=\linewidth]{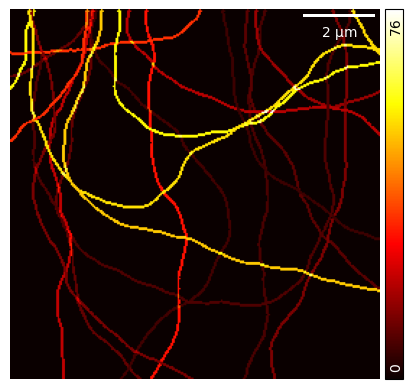}
    \caption{Out-of-focus ground truth $\x_{2,\mathrm{GT}}$}
    \label{fig:tub_back}
  \end{subfigure}
  \begin{subfigure}{0.32\linewidth}
    \centering
    \includegraphics[width=\linewidth]{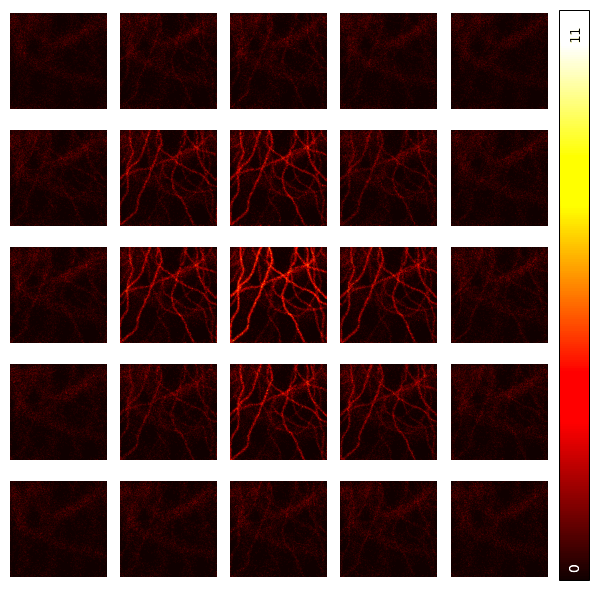}
    \caption{Simulated $s^2$ISM acquisitions.}
    \label{fig:noisy_3d_tub}
\end{subfigure}
  \caption{Left/center: simulated tubulin structures corresponding to the foreground (1) and background (2) planes. Right: simulated $s^2$ISM data generated from the superposition of foreground and background fluorescence contributions under Poisson noise.}
  \label{fig:gt_3d_tubulin}
\end{figure}

The optical system is then modeled through two distinct sets of detector-dependent PSFs: 25 in-focus PSFs associated with the foreground plane and 25 out-of-focus PSFs associated with the background contribution. These simulated PSFs are shown in Figure~\ref{fig:psfs_3d}.

\begin{figure}
  \centering
  \begin{subfigure}{0.48\linewidth}
    \centering
    \includegraphics[width=0.7\linewidth]{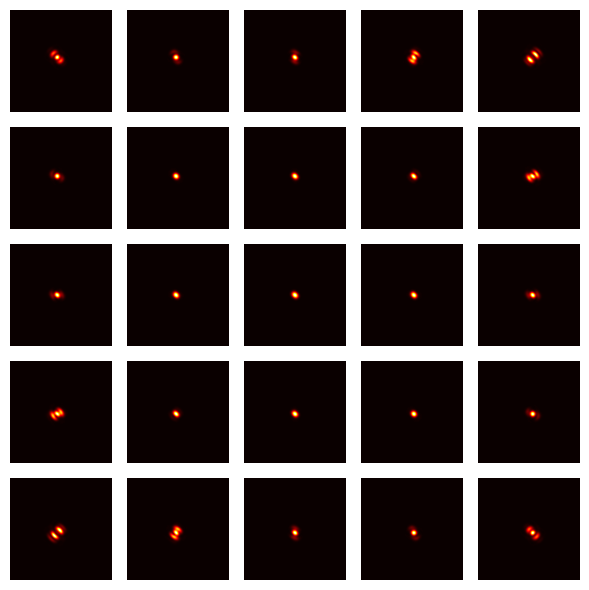} 
    \caption{In-focus PSFs}
    \label{fig:psf_in}
  \end{subfigure}\hfill
  \begin{subfigure}{0.48\linewidth}
    \centering
    \includegraphics[width=0.7\linewidth]{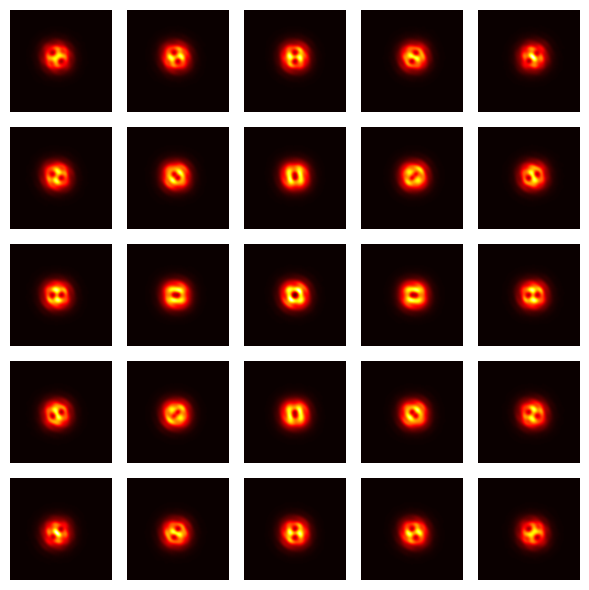}
    \caption{Out-of-focus PSFs}
    \label{fig:psf_out}
  \end{subfigure}
  \caption{Simulated detector-dependent PSFs associated with the foreground (in-focus, left) and background (out-of-focus, right) planes.}
  \label{fig:psfs_3d}
\end{figure}

The measurements are generated by convolving the two planes with their respective PSFs, summing the corresponding contributions, and finally applying Poisson noise independently to each detector channel according to the forward model described in Section~\ref{sec: mathematical modeling for $s^2ISM$}. An example of the resulting noisy acquisition is shown in Figure~\ref{fig:noisy_3d_tub}.

As in the single-plane setting, the regularization parameter $\lambda$ is selected using the masked RWP strategy described in Section~\ref{secRPW}, combined here with the high-pass spectral filtering procedure introduced in Section~\ref{sec:wp_highpass}. This modification prevents low-frequency modeling errors associated with the unregularized background component from dominating the parameter-selection process.

Figure~\ref{fig:sim_result_3d} reports the reconstruction results obtained in this dual-plane setting. RL reconstruction struggles to suppress the background contribution, leading to partial loss of the tubulin structures. In contrast, the proposed explicitly regularized schemes achieve substantially improved separation between foreground and background components.

Among the considered methods, PGD-TV provides the best overall reconstruction quality, accurately recovering the morphology of the tubulin filaments while effectively suppressing the out-of-focus signal. The $\ell_1$-based reconstructions successfully remove the background contribution but tend to produce sparser and less continuous structures. Finally, although the MD-based schemes effectively suppress the background plane, they also attenuate part of the foreground signal, leading to a partial loss of fine structural information.

\begin{figure}
    \centering
    \newcommand{\methodhead}[3]{%
        \shortstack{\scriptsize #1\\[1pt]
                    \tiny PSNR #2\\[-1pt]
                    \tiny MICROSSIM #3}}
    \setlength{\tabcolsep}{1.5pt}
    \renewcommand{\arraystretch}{0.3}
    \newcommand{\tw}{0.225\linewidth} 

    \begin{tabular}{@{}c cccc@{}}
        & \methodhead{Noisy sum}{16.73}{0.191}
        & \methodhead{RL ($k=5$)}{17.02}{0.53}
        & \methodhead{RL ($k=1000$) (sat)}{17.11}{0.52}
        & \methodhead{ProjGD ($\lambda=0$)}{17.11}{0.53} \\[3pt]

        \rotatebox{90}{\scriptsize In-focus} &
        \includegraphics[width=\tw]{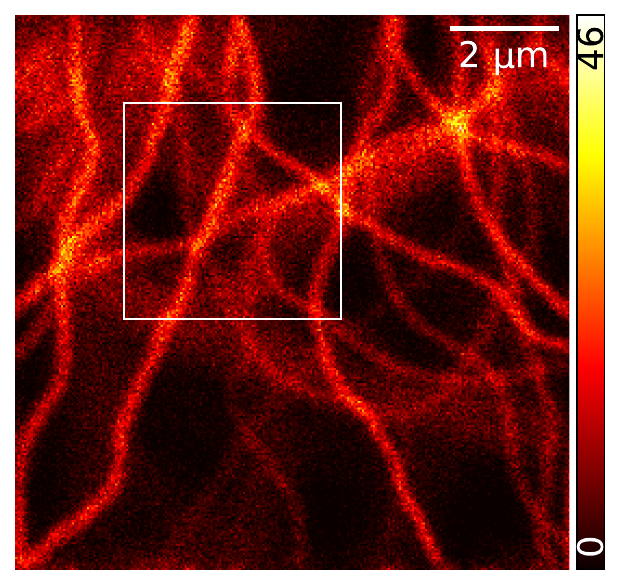} &
        \includegraphics[width=\tw]{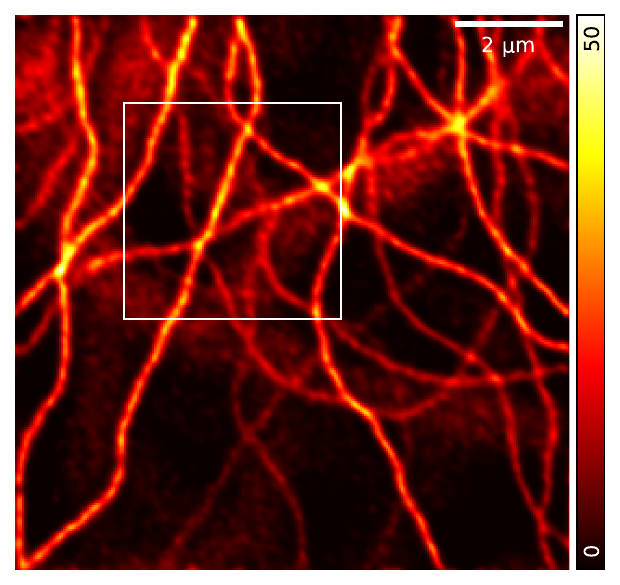} &
        \includegraphics[width=\tw]{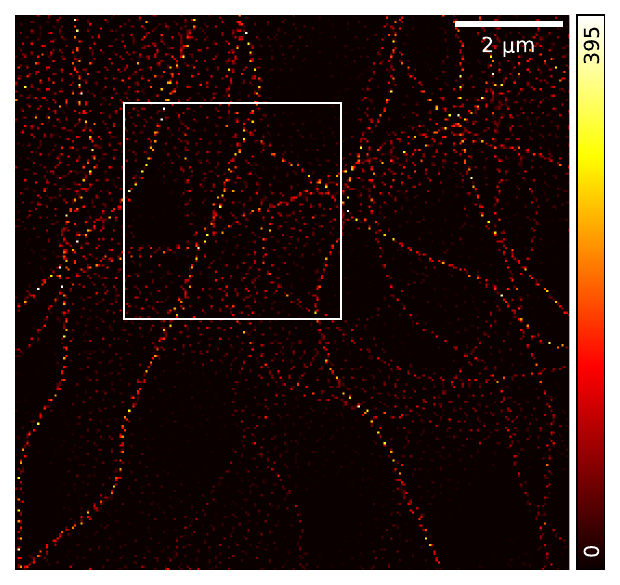} &
        \includegraphics[width=\tw]{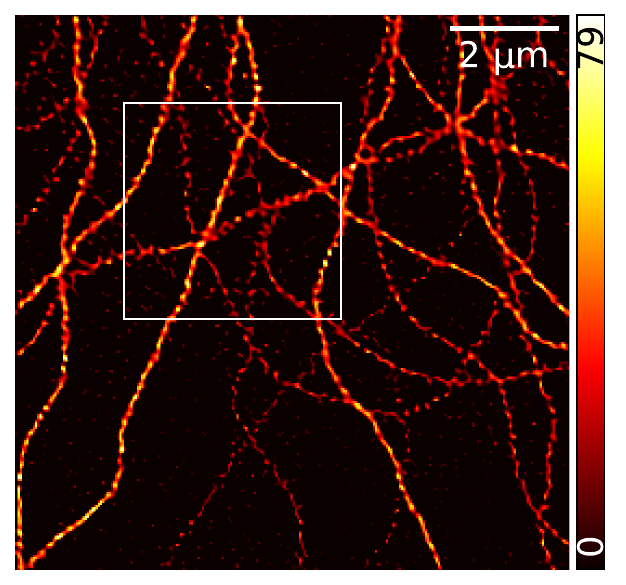} \\

        \rotatebox{90}{\scriptsize Crop} &
        \includegraphics[width=\tw]{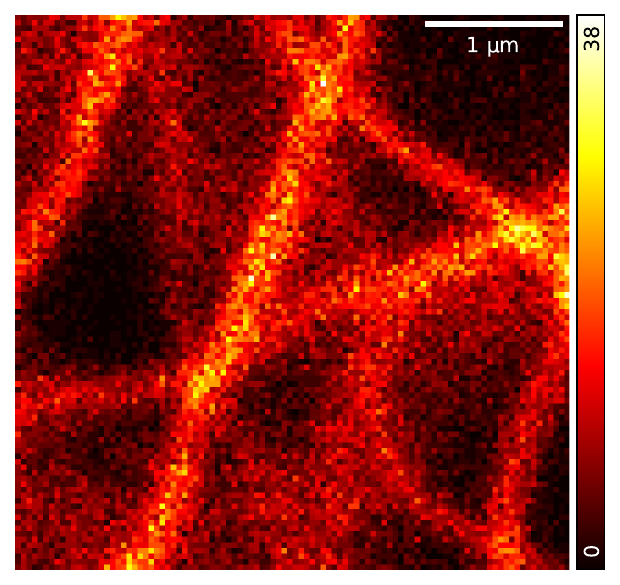} &
        \includegraphics[width=\tw]{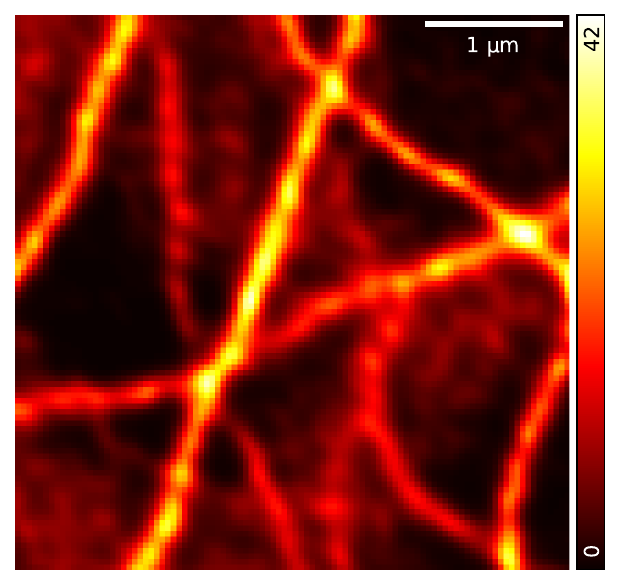} &
        \includegraphics[width=\tw]{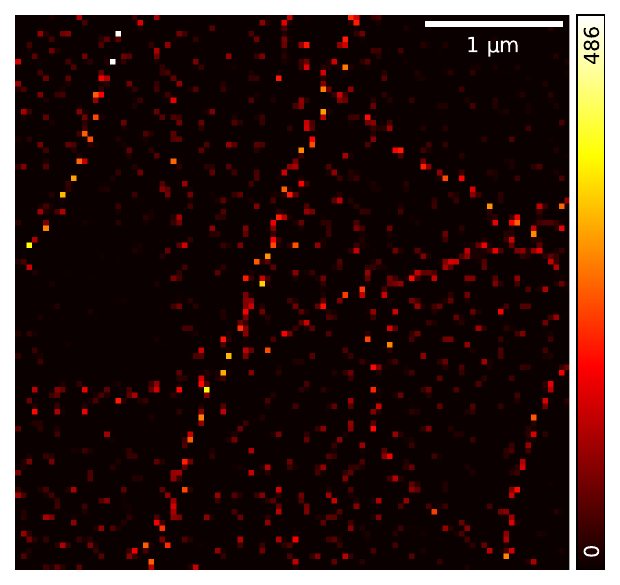} &
        \includegraphics[width=\tw]{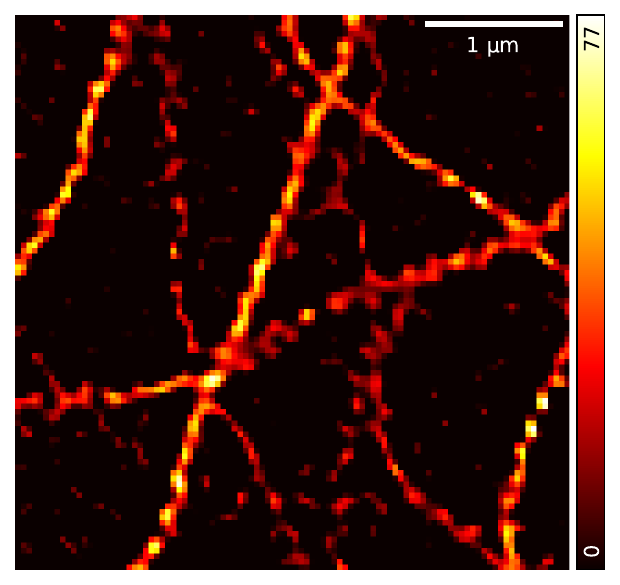} \\

        \rotatebox{90}{\scriptsize Out-of-focus} &
        &
        \includegraphics[width=\tw]{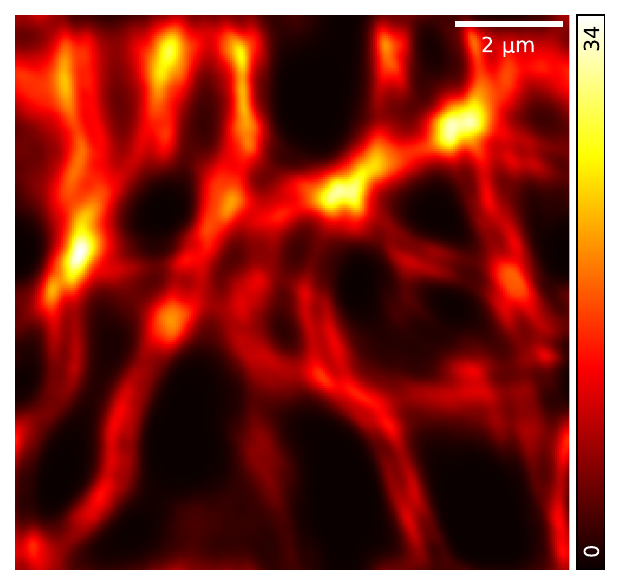} &
        \includegraphics[width=\tw]{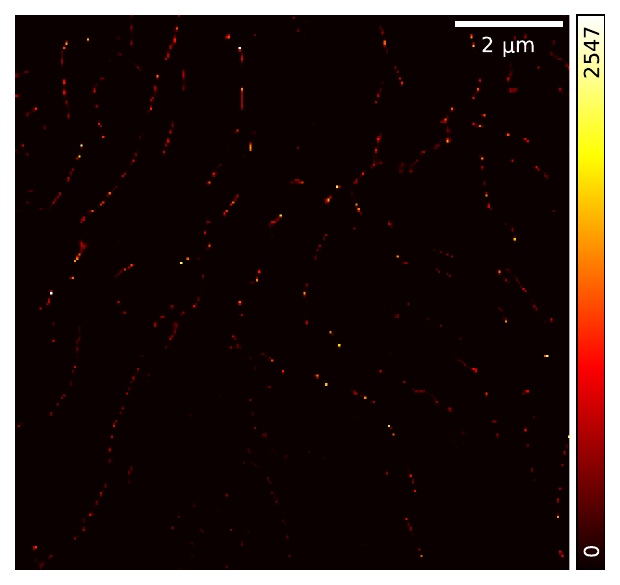} &
        \includegraphics[width=\tw]{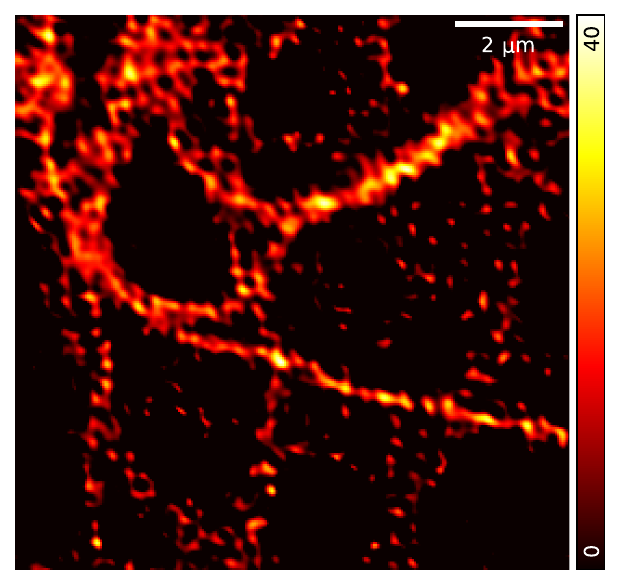} \\
    \end{tabular}

    \vspace{0.3cm}

    \begin{tabular}{@{}c cccc@{}}
        & \methodhead{ProxGD-$\ell_1$ ($\lambda=2.68\!\times\!10^{-2}$)}{20.25}{0.72}
        & \methodhead{MD-$\ell_1$ ($\lambda=9.09\!\times\!10^{-2}$) (sat)}{17.32}{0.68}
        & \methodhead{PGD TV ($\lambda=1.34\!\times\!10^{-2}$)}{18.99}{0.69}
        & \methodhead{MD TV ($\lambda=3.09\!\times\!10^{-1}$)(sat)}{16.73}{0.62} \\[3pt]

        \rotatebox{90}{\scriptsize In-focus} &
        \includegraphics[width=\tw]{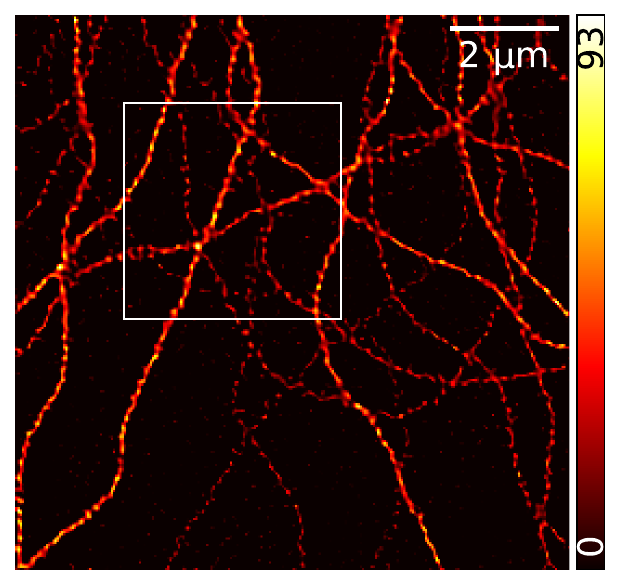} &
        \includegraphics[width=\tw]{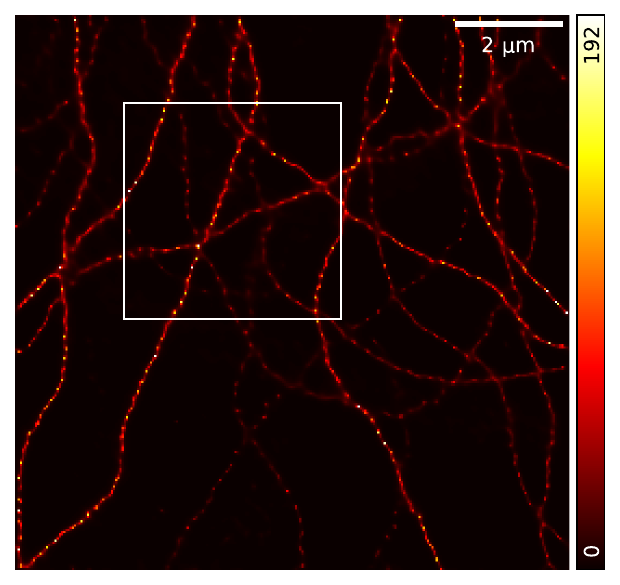} &
        \includegraphics[width=\tw]{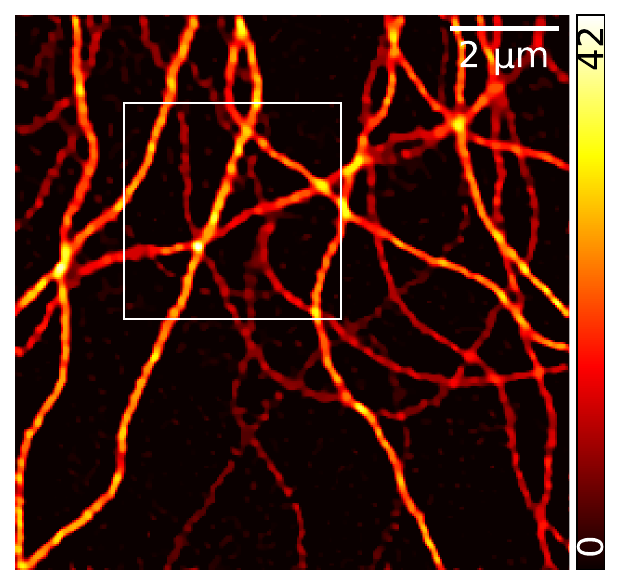} &
        \includegraphics[width=\tw]{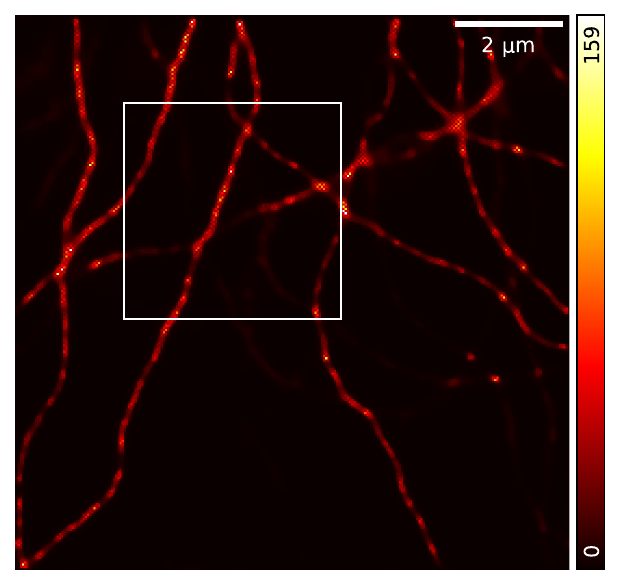} \\

        \rotatebox{90}{\scriptsize Crop} &
        \includegraphics[width=\tw]{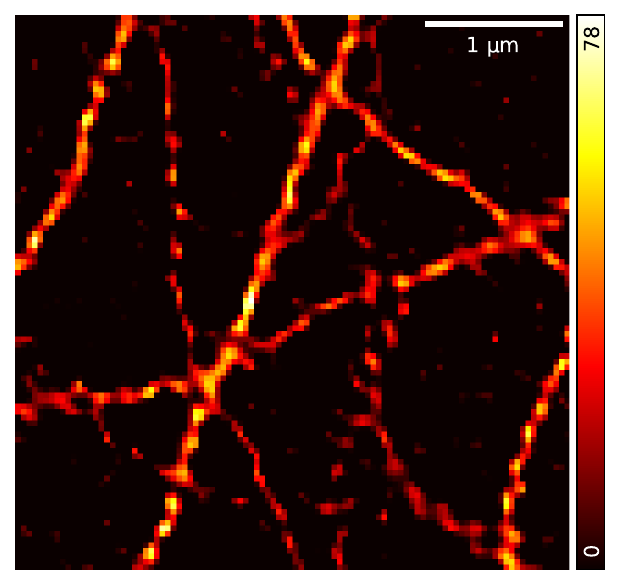} &
        \includegraphics[width=\tw]{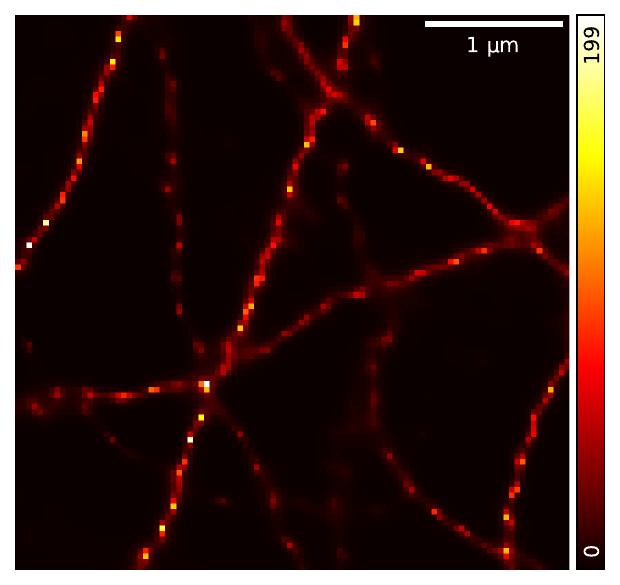} &
        \includegraphics[width=\tw]{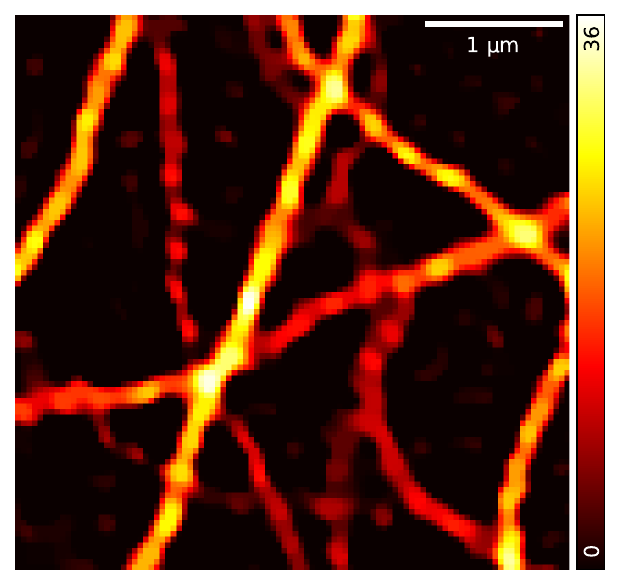} &
        \includegraphics[width=\tw]{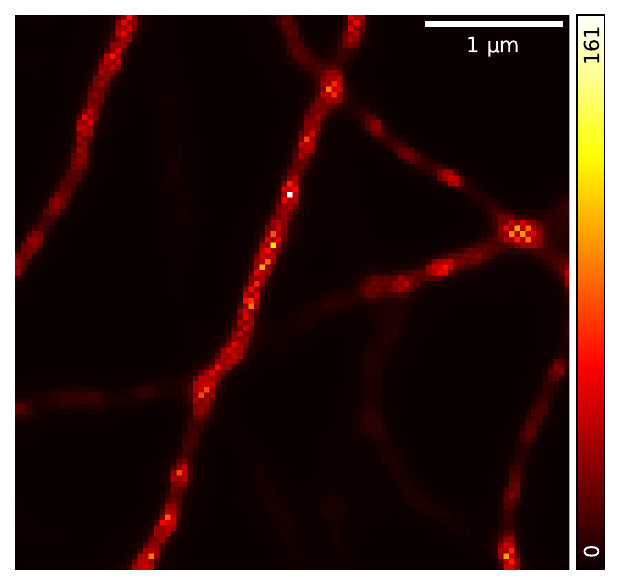} \\

        \rotatebox{90}{\scriptsize Out-of-focus} &
        \includegraphics[width=\tw]{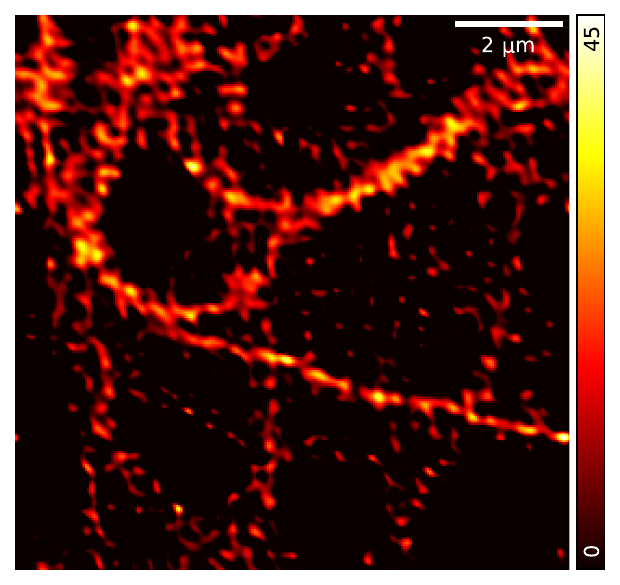} &
        \includegraphics[width=\tw]{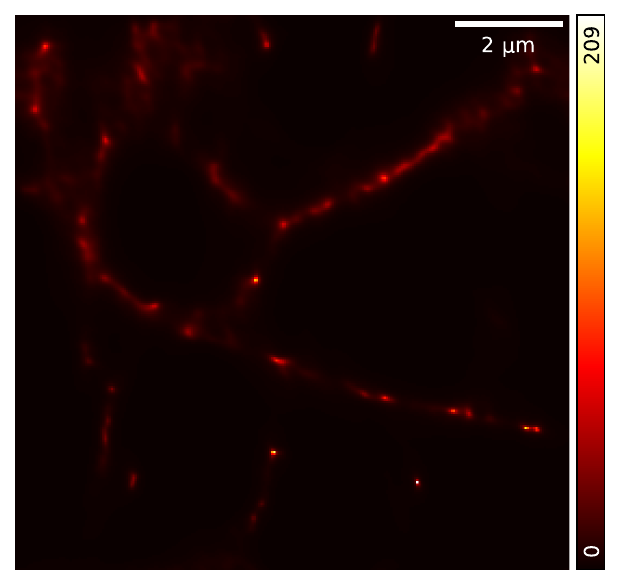} &
        \includegraphics[width=\tw]{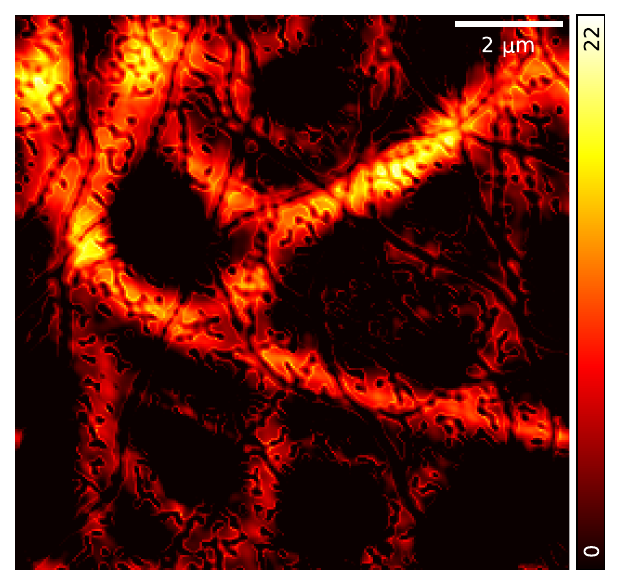} &
        \includegraphics[width=\tw]{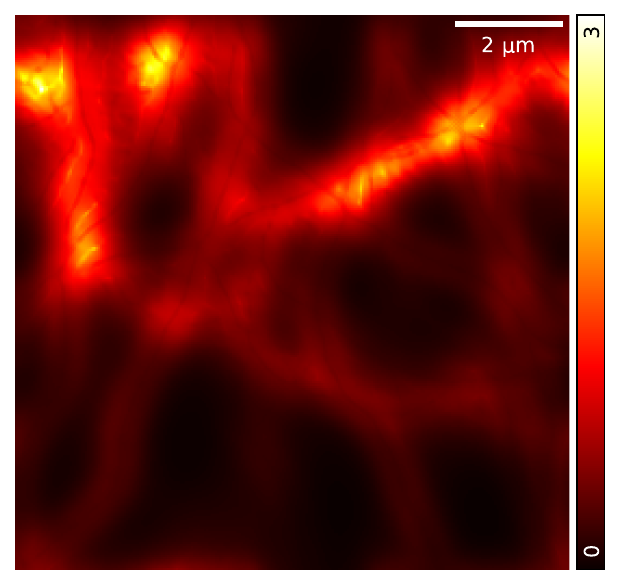} \\
    \end{tabular}

\caption{Results of simulated tubulin 3D reconstruction. Top block: 
        baseline comparisons (noisy sum, Richardson--Lucy, projected GD 
        without regularization). Bottom block: regularized iterative methods. 
        For each method the full sum-projection (Full), a central crop (Crop), 
        and the axial view (Axial) are shown, together with PSNR and 
        MicroSSIM.}
    \label{fig:sim_result_3d}
\end{figure}

\subsection{Validation of the RWP parameter selection strategy}

In this Section, we report some numerical tests validating the RWP strategy for choosing the optimal regularization weight. To start with, in Figure~\ref{fig:whiteness_graphs} we illustrate the parameter-selection process for MD-TV (Algorithm~\ref{alg:Mirror-TV-Combined}) on the simulated 3D tubulin dataset. The figure reports both the whiteness functional and the corresponding oracle PSNR values as functions of $\lambda$. Two observations are worth emphasizing. First, the whiteness functional exhibits a well-defined minimum, providing a stable and fully unsupervised parameter-selection criterion. Second, the parameter $\lambda^*$ minimizing $\mathcal W(\lambda)$ appears close to the value maximizing the PSNR, a quantity which is unavailable in practice. Similar plots can be showed for the other regularization/algorithmic strategies considered, thus validating the RWP strategy as an effective unsupervised approach for regularization parameter selection in the ISM setting.

\begin{figure}
\begin{subfigure}{0.45\textwidth}
    \centering
    \includegraphics[width=\linewidth]{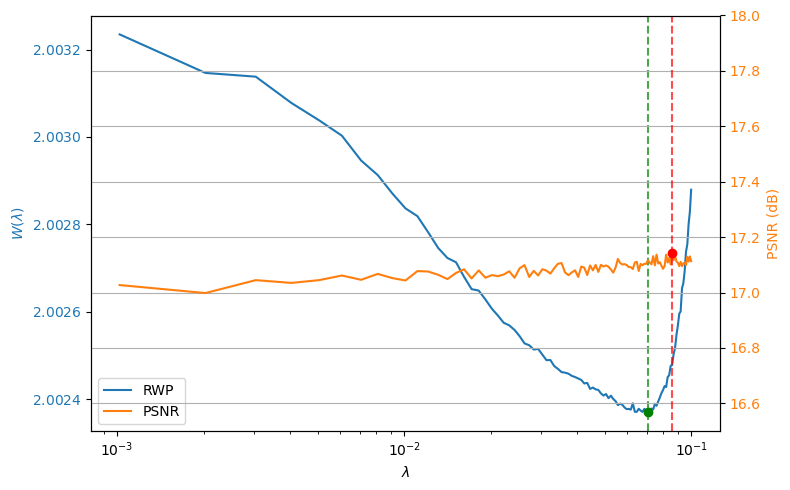}
    \caption{Evolution of $\mathcal{W}$ and PSNR  vs.~$\lambda$ for MD-TV.}
    \label{fig:whiteness_graphs}
\end{subfigure}
\hspace{0.5cm}
\begin{subfigure}{0.45\textwidth}
    \centering
    \includegraphics[width=\linewidth]{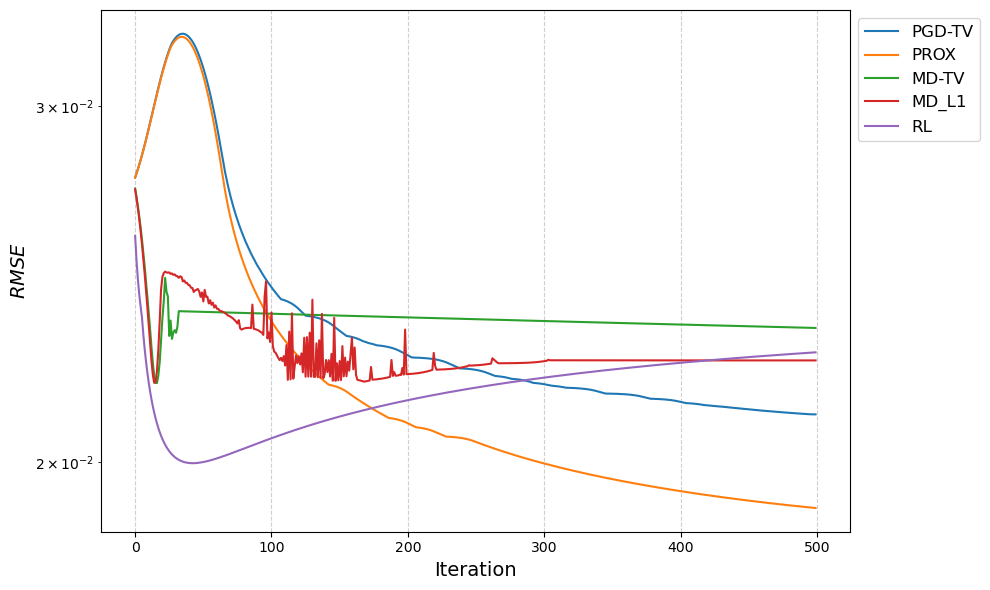}
    \caption{RMSE evolution.}
    \label{fig:conv_curv}
    \end{subfigure}
    \caption{Validation of the proposed RWP parameter selection method and of the regularization framework on the 3D dataset. Left: whiteness functional $\mathcal{W}(\lambda)$ and corresponding PSNR values as functions of the regularization parameter for MD-TV Algorithm~\ref{alg:Mirror-TV-Combined}. Right: RMSE as a function of the iteration count for all the proposed reconstruction algorithms.}
\end{figure}

\subsubsection{Reconstruction quality and convergence behavior}

Figure~\ref{fig:conv_curv} reports the Root Mean Square Error (RMSE)  as a function of the iteration count. In contrast to the RL semi-convergent behavior discussed in Section~\ref{sec: implicit reg and early stopping} and showed in Figure \ref{fig:no_reg_conv}, the explicitly regularized schemes remain substantially more stable. In particular, PGD-TV and ProxGD-$\ell_1$ exhibit a stable decrease of the RMSE, confirming that the proposed regularized formulations can be run until numerical convergence without requiring empirical early stopping rules. We observe that the behavior of MD-$\ell_1$ and MD-TV closely resemble the one of RL, which can be explained by their shared multiplicative structure. Within the logarithmic Mirror Descent geometry, the $\ell_1$ regularizer becomes smooth on the strictly positive orthant and contributes only a constant shift to the gradient of the KL terms (see Algorithm~\ref{alg:Mirror-L1-Combined}), making the resulting dynamics structurally very similar to Richardson--Lucy iterations. This is also reflected in its oscillatory behavior. 

\subsection{Robustness to different photon flux levels}

Figure~\ref{fig:flux_comparison} illustrates the behavior of the proposed framework under different photon flux levels. As expected, increasing the flux improves the quality of the measurements, resulting in progressively cleaner acquisitions and higher-quality reconstructions. This trend is reflected by the increasing PSNR and MicroSSIM values reported below each reconstruction, while the recovered foreground and background components remain well separated across all considered noise levels. To validate how the RWP criterion adapts to these different acquisition conditions, while accounting for the stochasticity of the Poisson noise, we further considered the procedure applied to select the optimal weighting for the ProxGD-$\ell_1$ model over $N = 10$ independent noise realizations for different flux levels and average the resulting optimal lambdas  showing also uncertainty. As shown in Figure~\ref{fig:adaptivity_noise}, the selected parameter decreases with increasing photon flux, indicating that, as expected, less regularization (selected by RWP) is required as the measurements become more reliable, while lower photon fluxes automatically lead to stronger regularization.

\begin{figure}
    \setlength{\tabcolsep}{1.5pt}      
    \renewcommand{\arraystretch}{0.6}  
    \newcommand{\imgw}{0.30\linewidth} 
    \begin{tabular}{@{}c ccc@{}}
        & \small $\bar{\y}$ & \small In-focus $\x^*_1$ & \small Out-of-focus $\x^*_2$\\[2pt]

        \rotatebox{90}{ Flux 10} &
        \includegraphics[width=\imgw]{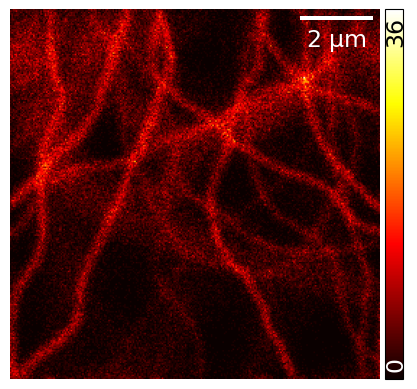}   &
        \includegraphics[width=\imgw]{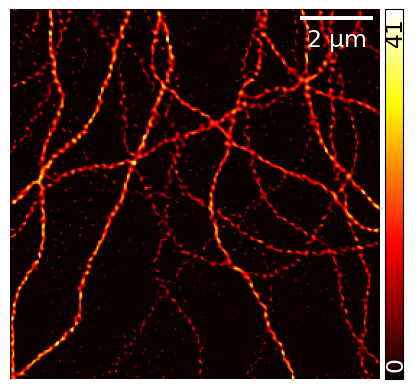} &
        \includegraphics[width=\imgw]{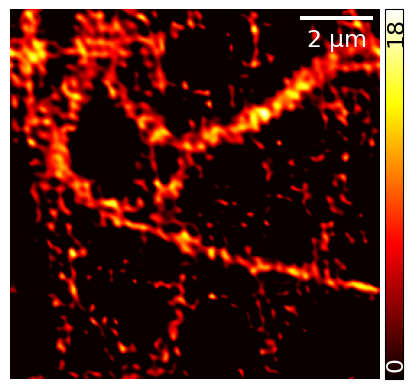}  \\
        & \multicolumn{3}{c}{
            $\lambda^\ast = 3.73\times10^{-2}$ \quad
            PSNR $= 20.15$ \quad
            MICROSSIM $= 0.670$} \\[3pt]

        \rotatebox{90}{ Flux 20} &
        \includegraphics[width=\imgw]{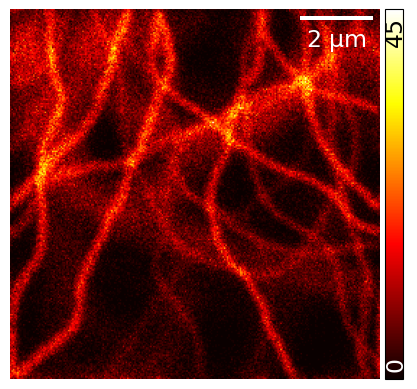}   &
        \includegraphics[width=\imgw]{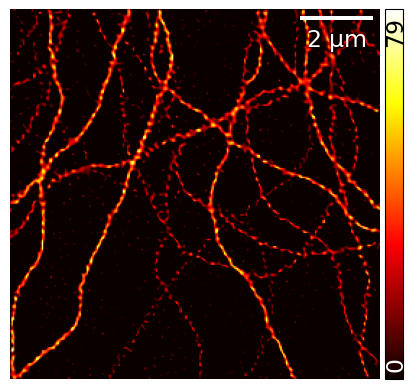} &
        \includegraphics[width=\imgw]{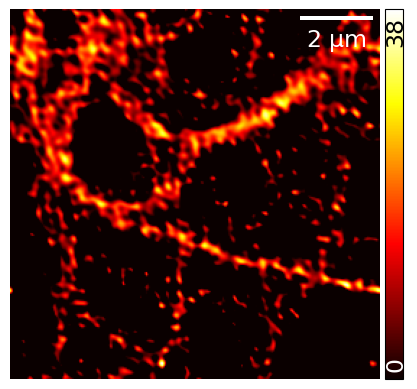}  \\
        & \multicolumn{3}{c}{
            $\lambda^\ast = 2.33\times10^{-2}$ \quad
            PSNR $= 20.74$ \quad
            MICROSSIM $= 0.727$} \\[3pt]

        \rotatebox{90}{ Flux 40} &
        \includegraphics[width=\imgw]{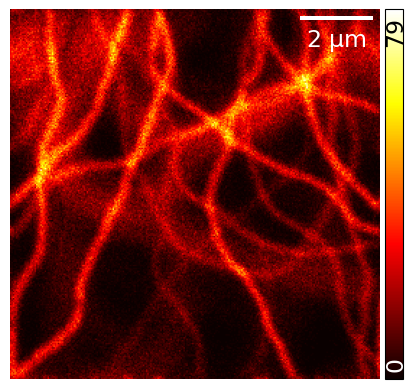}   &
        \includegraphics[width=\imgw]{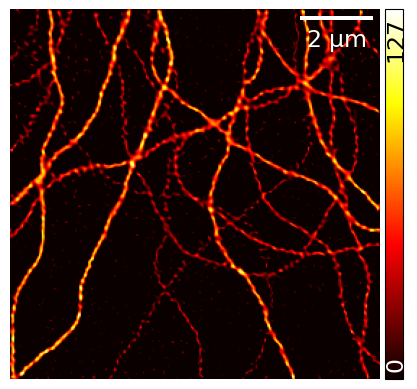} &
        \includegraphics[width=\imgw]{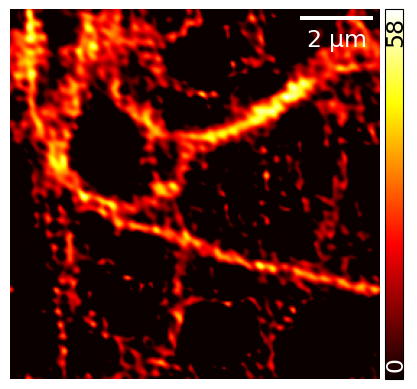}  \\
        & \multicolumn{3}{c}{
            $\lambda^\ast = 1.84\times10^{-2}$ \quad
            PSNR $= 22.08$ \quad
            MICROSSIM $= 0.789$} \\
    \end{tabular}
    
    \caption{Top: whiteness and error curves for the three flux levels.
             Bottom: noisy measurement and reconstructions (in-focus/out-of-focus) at
             the $\lambda^\ast$ selected by the RWP. For each
             flux, we report the selected $\lambda^\ast$ and the metrics
             (PSNR, MicroSSIM) of the corresponding reconstruction.}
    \label{fig:flux_comparison}
\end{figure}

\begin{figure}
    \centering
    \includegraphics[width=\linewidth]{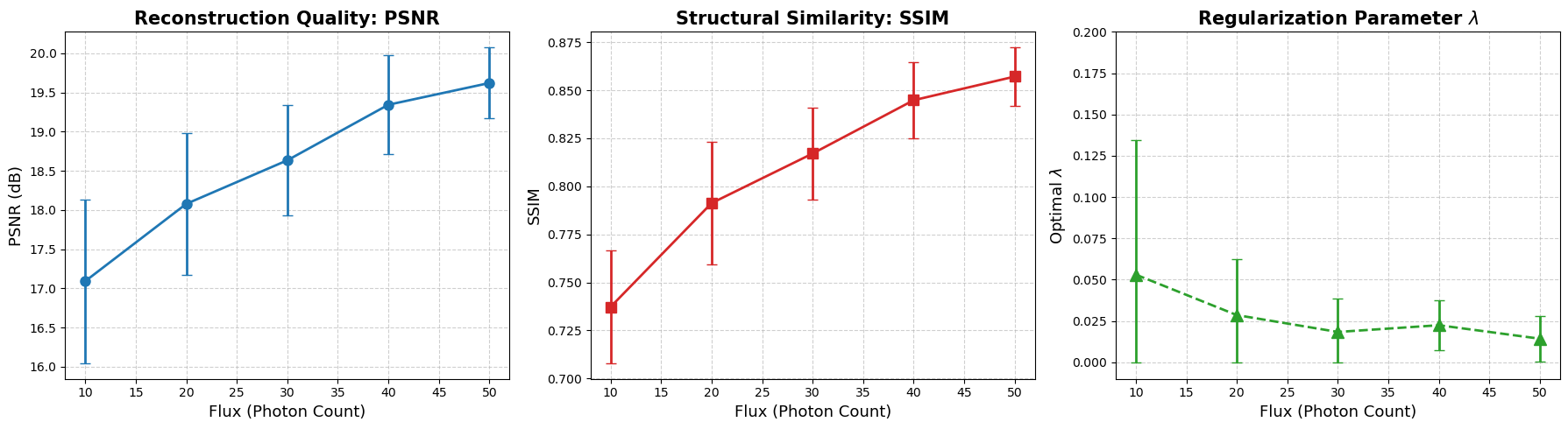}
    \caption{Adaptivity of the proposed RWP parameter-selection strategy to different photon flux levels. Each curve shows the high-pass filtered whiteness functional $\mathcal{W}(\lambda)$, averaged over $N=10$ independent noise realizations, as a function of the regularization parameter $\lambda$; the dashed vertical lines mark the  $\lambda^*$ selected by RWP. As the photon flux decreases, the criterion automatically selects progressively larger regularization parameters, reflecting the increased need of higher regularization amount due to the presence of higher noise. 
    }
    \label{fig:adaptivity_noise}
\end{figure}

\section{Numerical results on real s$^2$ISM data}
\label{real_results}

In this section, we evaluate the proposed reconstruction framework on real s$^2$ISM  acquisitions. In contrast to the simulated setting, no ground-truth image is available in this case, making the computation of quantitative error metrics such as PSNR or MicroSSIM not possible. Consequently, the assessment of the reconstruction quality relies primarily on qualitative criteria, including structural continuity, background suppression, contrast enhancement, and robustness against noise amplification artifacts. 
Figure \ref{fig:all_tomm20_meas} reports the three datasets used in the following. They show subcellular targets representing morphologically distinct biological organizations: the microtubule marker $\alpha$-tubulin, the intermediate filament vimentin and the heterochromatin mark histone H3K9me3. More details on their acquisition are given below.

\subsection{Real data acquisition}
\label{real_data_acq}
All real data were acquired using a custom-built laser-scanning microscope equipped with a $5\times5$ square single-photon avalanche diode (SPAD) array detector (PRISM Light Kit, Genoa Instruments), previously described in \cite{slenders2025array, patil2025open}, and operated in Image Scanning Microscopy (ISM) modality. Briefly, the system was designed to image samples labelled with far-red fluorophores by focusing a 640 nm laser diode (PicoQuant LDH-D-C-640S) through a high-numerical-aperture oil-immersion objective lens (Nikon SR HP Apo TIRF 100$\times$/1.49 NA oil). The emitted fluorescence was collected by the same objective, spectrally separated from the excitation light using a dichroic mirror (F48-643, AHF, Germany) and emission filters (LP01-633R-25 and Notch 642), and subsequently imaged onto a secondary image plane where the SPAD array was located. A relay-lens system provided an overall magnification of approximately $500\times$ between the sample plane and the detector plane, corresponding to an effective SPAD-array size of approximately 1 Airy Unit. The excitation focus was raster-scanned across the sample using a pair of galvanometric mirrors. Synchronization between scanning and photon detection, as well as generation of the ISM dataset, was fully controlled by the open-source BrightEyes-MCS microscope-control software \cite{donato2024brighteyes}. 
The acquisitions are performed with a pixel dwell time of 64~$\mathrm{\mu}$s (64 time bins of 1~$\mathrm{\mu}$s), a pixel size of 50~nm and a laser power of 40~$\mathrm{\mu}$W.

\subsection{Biological sample preparation}

\textit{Cell Culture.} HeLa cells were grown in DMEM with $10\%$ FBS, $1\%$ L-glutamin and $1\%$ penicillin/streptomycin (Sigma-Aldrich) at 37$^{\circ}$C and $5\%$ CO2 in a cell culture incubator. After two days in incubator, they were seeded at medium confluence on cleaned and round 18 mm diameter high resolution 1.5'' glass coverslips (Marienfield, VWR).

\textit{Cell fixation.} After 24 hours, the cells were washed one time with sterile Dulbecco PBS 1X (Sigma Life Science). A fixation solution ($0.2\%$ glutaraldehyde + $0.2\%$ formaldehyde in DPBS) is added for 15 min at room temperature. Cells were then washed 3 times in PBS 1x and stored in PBS 1x + sodium azide $0.1\%$ at 4$^{\circ}$C until labelling. 

\textit{Cell immunolabeling.} For the saturation step, cells are blocked for 1h in PBS + 3$\%$ BSA + 0.1$\%$ Triton. The cells were incubated 1h at room temperature with primary antibody in PBS + 3$\%$ BSA + 0.1$\%$ Triton. This was followed by three washing steps in PBS + 3$\%$ BSA + 0.1$\%$ Triton and incubation for 1h at room temperature with secondary antibody. Three more washes with PBS 1x are performed. The cells are then embedded in Prolong Diamond mounting medium (ThermoFisher) and sealed with two-component sealants. Samples are stored in fridge at 4$^{\circ}$C until imaging.

\textit{Primary antibodies.} The primary antibodies are used with a dilution of 1:500. Mouse anti-alpha-tubulin (Thermo Fisher DM1A + Thermo Fisher B512), chicken anti-vimentin (EnCor CPCA-Vim) and rabbit anti-H3K9me3 (Abcam ab8898). 

\textit{Secondary antibodies.} The secondary antibodies are used with a dilution of 1:500. Anti-mouse Atto 647N (Sigma 50185), anti-rabbit Atto 647N (Sigma 40839) and anti-chicken CF660C (Biotium 20371).

\begin{figure}
    \centering
    
    \begin{subfigure}{0.35\textwidth} 
        \centering
        \includegraphics[width=\textwidth]{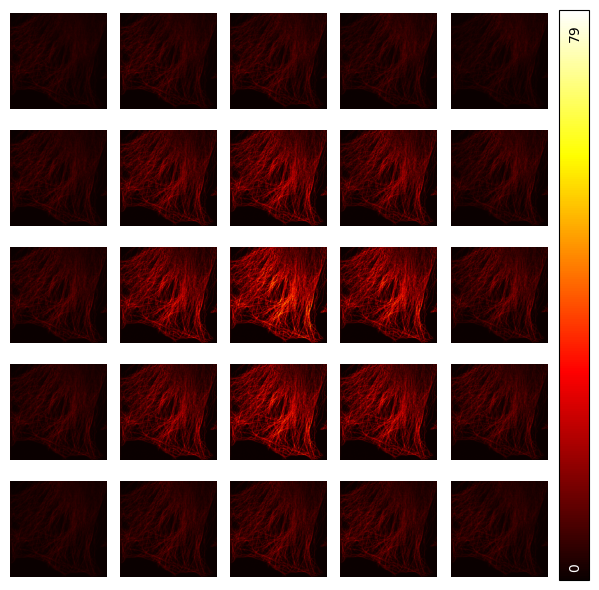}
        \caption{Measurement \emph{Tubulin}}
        \label{fig:meas_01}
    \end{subfigure}
    \hspace{1cm} 
    \begin{subfigure}{0.35\textwidth} 
        \centering
        \includegraphics[width=\textwidth]{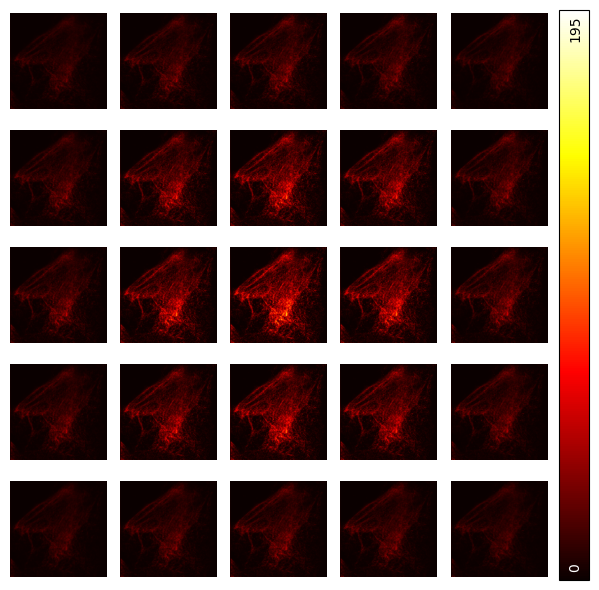}
        \caption{Measurement \emph{Vimentin}}
        \label{fig:meas_02}
    \end{subfigure}
    
    \vspace{0.5cm} 
    
    \begin{subfigure}{0.35\textwidth} 
        \centering
        \includegraphics[width=\textwidth]{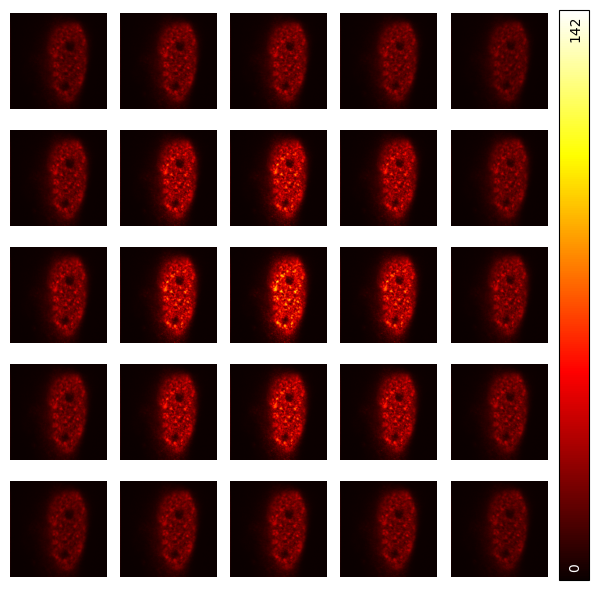}
        \caption{Measurement \emph{Histone}}
        \label{fig:meas_03}
    \end{subfigure}
    \hspace{1cm} 

    \caption{s$^2$ISM real datasets.}
    \label{fig:all_tomm20_meas}
\end{figure}

\subsection{Reconstruction stability and parameter selection}

Figure~\ref{fig:01tomm_iteration_comparison} highlights the stability of the reconstruction process along iterations for Algorithm~\ref{alg:PGD-TV} in comparison with the RL iteration~\eqref{eq: Richardson-Lucy}. The progressive amplification of noise and high-frequency artifacts ultimately deteriorate the reconstruction quality at late RL iterations, as illustrated by the reconstructions corresponding to increasing iteration counts. In contrast, the solutions computed by the regularized framework produce stable reconstructions, yielding improved structural continuity and a more favorable compromise between resolution enhancement and noise suppression. 
As previously observed on simulated data, RL tends to amplify the dynamic range of the reconstructed image. Due to the multiplicative nature of the iteration, photon counts that are initially distributed across multiple measurements become progressively concentrated into a small number of pixels, producing intensity peaks far from the expected intensity range.

\subsection{Reconstruction results}

Figures~\ref{fig:02TUB_res} through~\ref{fig:02vimentine_results} illustrate the reconstruction results obtained across four different $s^2$ISM samples. Overall, the proposed regularized frameworks provide a significant improvement in image quality and spatial resolution compared to both the raw measurements and the standard RL reconstructions. A primary effect is the substantial suppression of background fluorescence and out-of-focus contributions, leading to improved contrast of the mitochondrial structures.

Among the considered regularization models, Total Variation (TV) promotes piecewise-smooth reconstructions while preserving structural continuity. As visible in the magnified crops, the TV-based methods (PGD-TV and MD-TV) provide effective edge-preserving denoising, resulting in cleaner and more natural-looking reconstructions. In contrast, the $\ell_1$ regularization approaches (ProxGD-$\ell_1$ and MD-$\ell_1$) enforce pointwise sparsity, leading to reconstructions with a more granular appearance.


\begin{figure}
    \centering
    \includegraphics[width=0.9\linewidth]{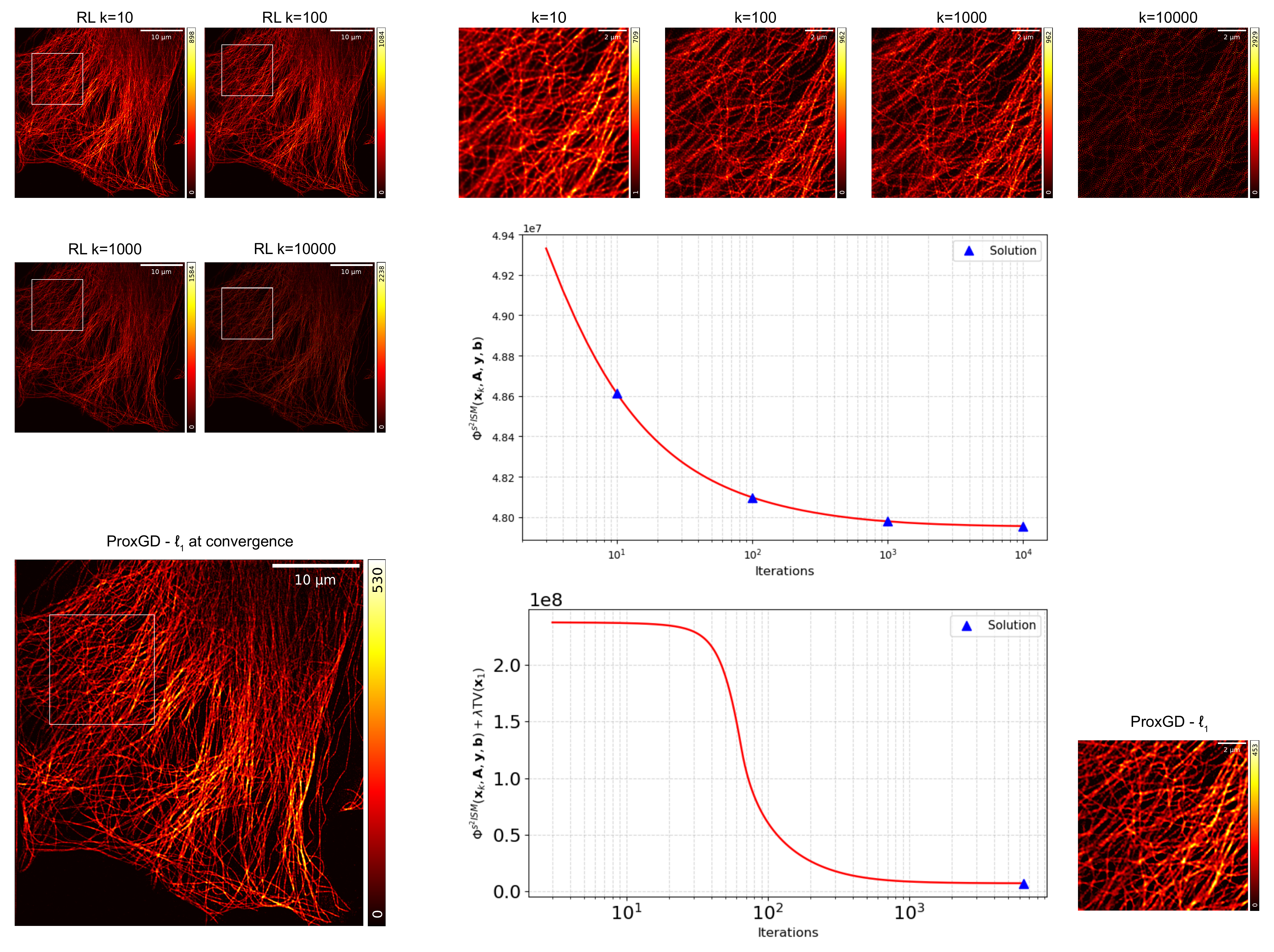}
    \caption{Comparison between RL solutions at different iteration counts and the regularized ProxGD-$\ell_1$ solution with optimal parameter selected by RWP. Figures (a), (b), (c), (d): RL leads to noise overfitting at large $k$, thus requiring early stopping. Fig (f): TV regularization ensures a stable, high-quality solution at convergence.}
    \label{fig:01tomm_iteration_comparison}
\end{figure}

To conclude, Table~\ref{tab:convergence_results} reports the computational performance of the considered algorithms in terms of execution time and number of iterations required to satisfy the convergence criterion~\eqref{conv_crit}. Overall, the multiplicative schemes exhibit a clear computational advantage over their PGD counterparts, converging in significantly fewer iterations and with substantially reduced execution times across all datasets. However, this improved computational efficiency comes at the price of reduced reconstruction quality in the considered ISM/$s^2$ISM setting as pointed out in the previous sections. 

Figure~\ref{fig:backtracking} demonstrates the effectiveness of the adaptive backtracking strategy (Algorithm~\ref{alg:backtracking_adaptive}) when applied to real data. Notably, for ProjGD and ProxGD, the stepsize behaves identically, initially increasing and then remaining stable.In contrast, the stepsize for both MD methods exhibits fluctuation, eventually dropping toward the theoretical lower bound as the algorithm converges.


\begin{table}[h]
    \centering
    \caption{Computational performance comparison: execution time (s) and number of 
    iterations to convergence ($\text{tol} = 10^{-5}$).}
    \label{tab:convergence_results}
    \small
    \begin{tabular}{lcccccc}
        \toprule
        \multirow{2}{*}{\textbf{Algorithm}} 
            & \multicolumn{2}{c}{\textbf{TUBULIN}} 
            & \multicolumn{2}{c}{\textbf{VIMENTIN}} 
            & \multicolumn{2}{c}{\textbf{HISTONE}} \\
        \cmidrule(lr){2-3} \cmidrule(lr){4-5} \cmidrule(lr){6-7} 
        & Time (s) & Iter. & Time (s) & Iter. 
        & Time (s) & Iter. \\
        \midrule
        PGD-TV   & 11.57 & 338 & 13.3 & 465 & 25.59 & 2032  \\
        ProxGD-$\ell_1$   & 39.88 & 1528 & 48.67 & 1873 & 27.75 & 2260  \\
        MD-TV    &  0.99 &  25 &  1.89 &  52 &  3.22 &  227 \\
        MD-$\ell_1$  & 6.89 &  223  &  7.46 &  240 & 2.83 &  191 \\
        \bottomrule
    \end{tabular}
\end{table}

\begin{figure}
    \centering
    \includegraphics[width=0.7\linewidth]{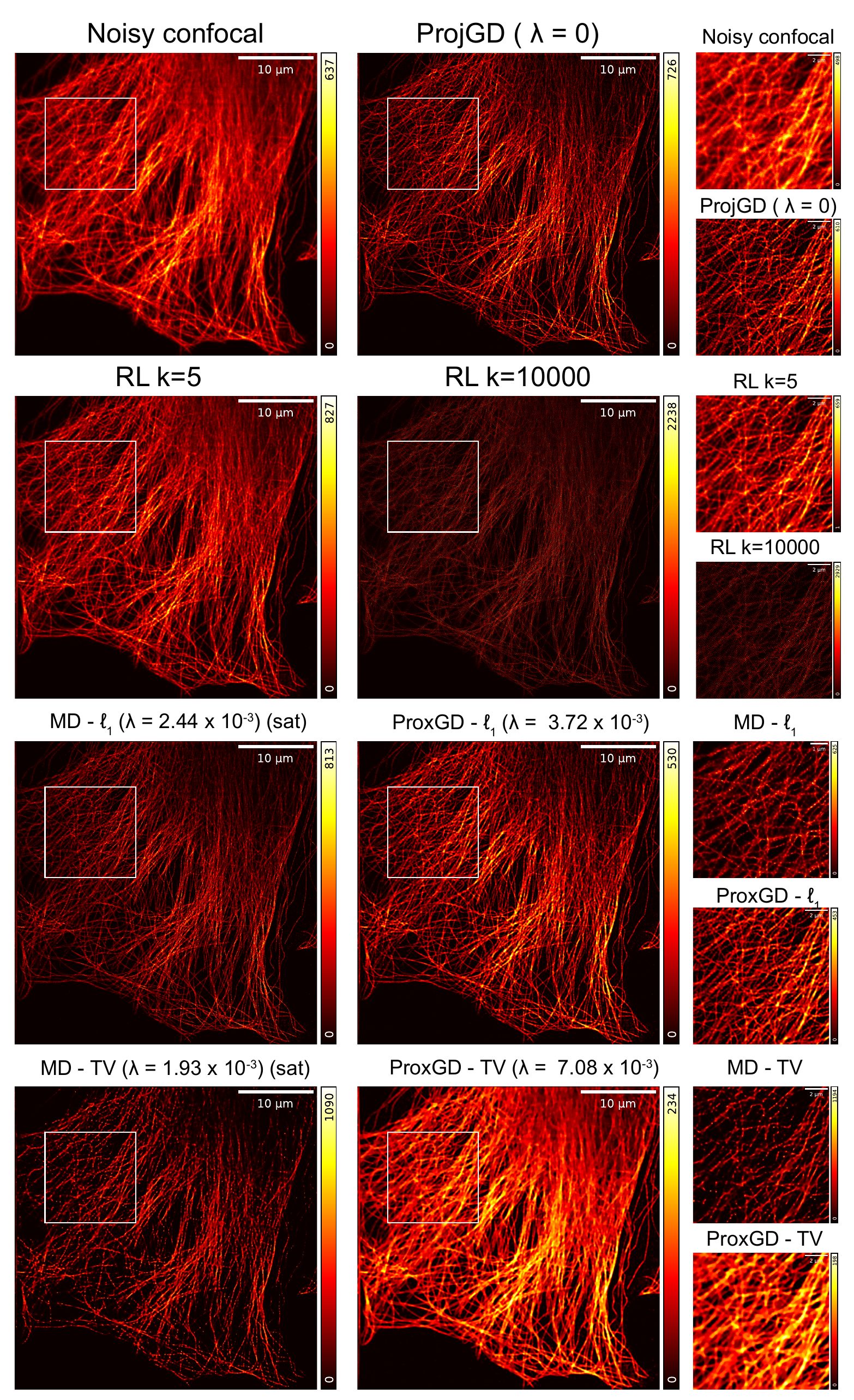}
    \caption{\emph{Tubulin} results: Full images  and corresponding crops .}
    \label{fig:02TUB_res}
\end{figure}

\begin{figure}
    \centering
    \includegraphics[width=0.7\linewidth]{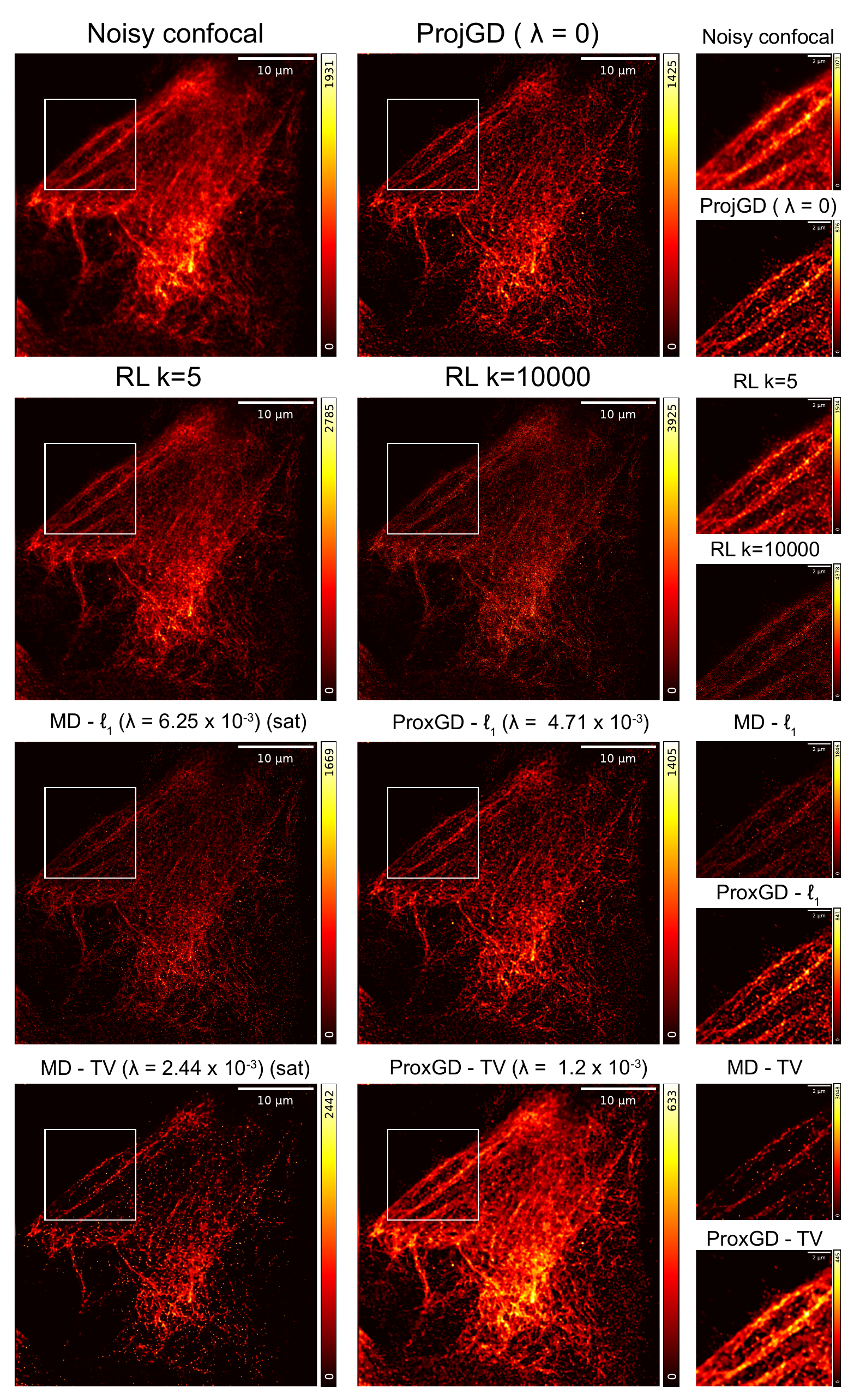}
    \caption{\emph{Vimentin} results: Full images  and corresponding crops .}
    \label{fig:02vimentine_results}
\end{figure}

\begin{figure}
    \centering
    \includegraphics[width=0.7\linewidth]{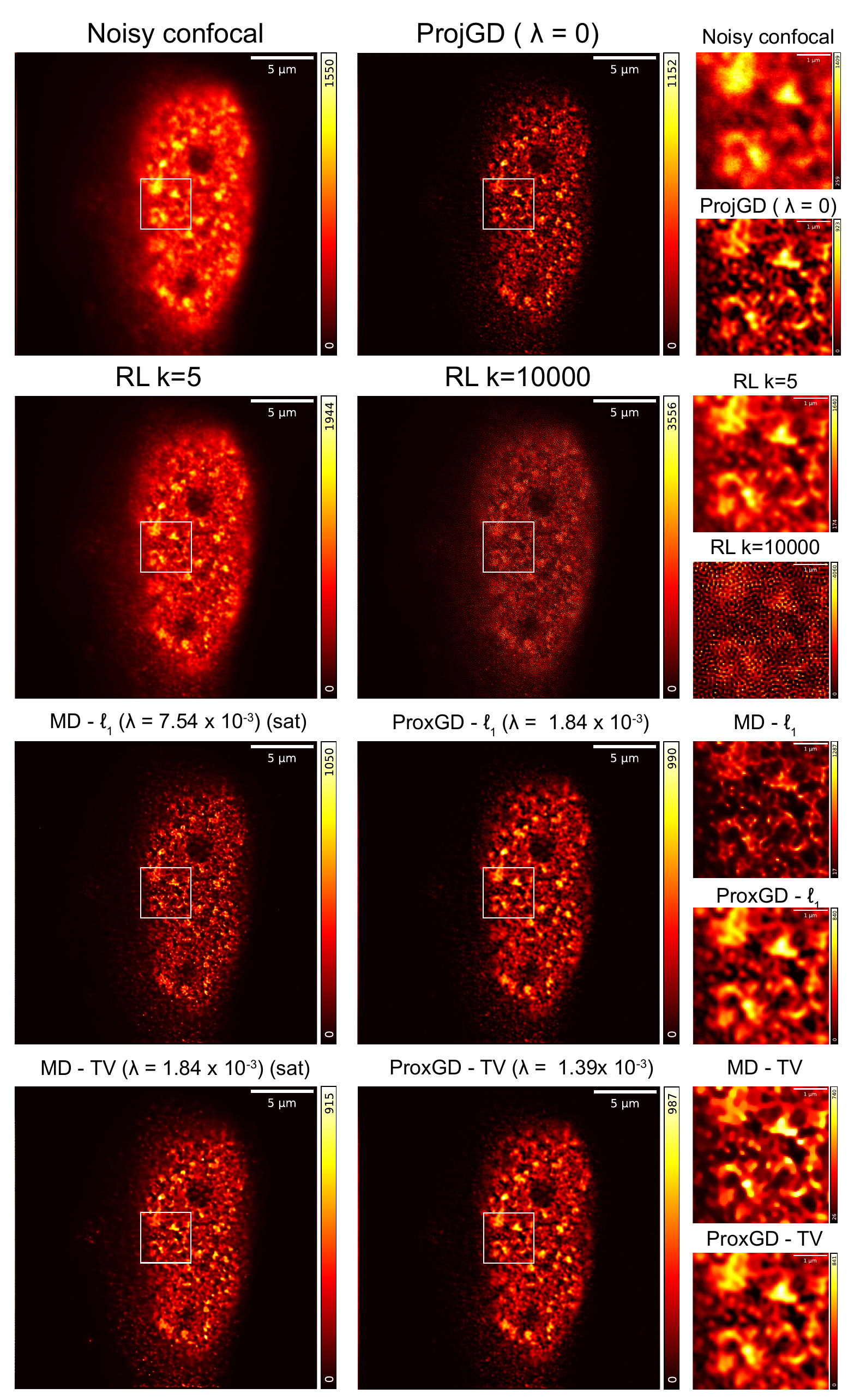}
    \caption{\emph{Histone} results: Full images  and corresponding crops .}
    \label{fig:03H3_res}
\end{figure}

\begin{figure}
    \centering
    
    \begin{subfigure}{0.47\linewidth}
        \centering
\includegraphics[width=1.1\linewidth]{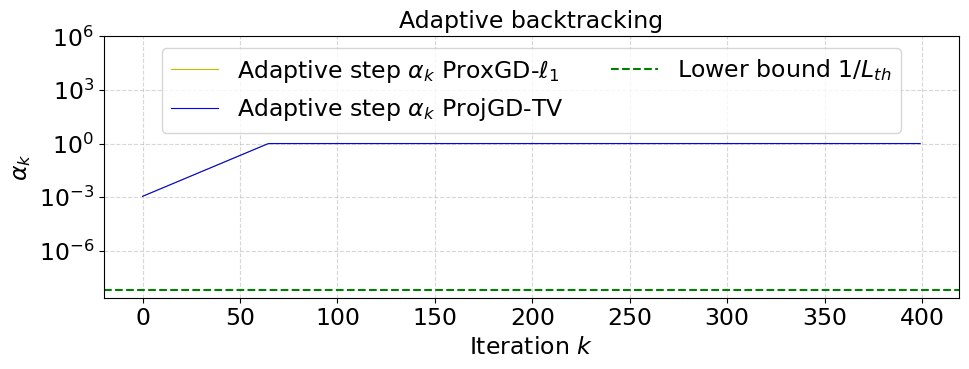} 
        \caption{Stepsize evolution for \eqref{eq:ProjGD} and \eqref{eq:PGD}.}
        \label{fig:back_pgdprox}
    \end{subfigure}
\hfill
    \begin{subfigure}{0.47\linewidth}
        \centering        \includegraphics[width=1.1\linewidth]{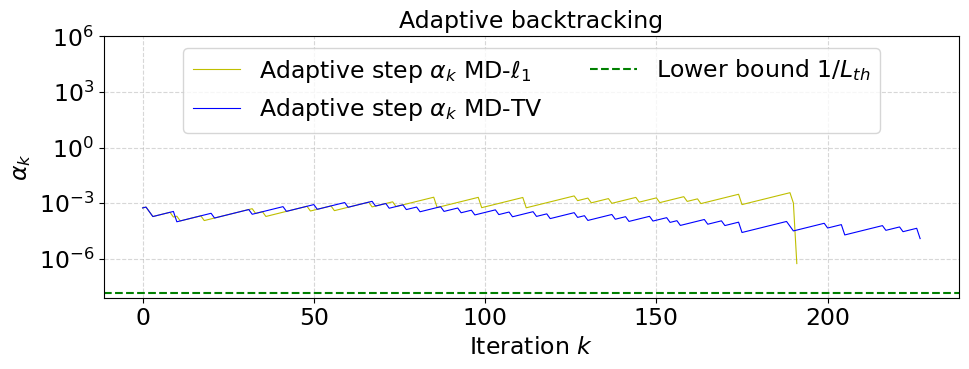}
        \caption{Stepsize evolution for \eqref{eq:explicit_md_step}.}
        \label{fig:back_md}
    \end{subfigure}
    
    \caption{Evolution of the adaptive stepsize $\alpha_k$ across first 500 iterations for the \emph{Histone} data, compared against the theoretical global lower bounds. In all cases, the adaptive backtracking strategy (Algorithm \ref{alg:backtracking_adaptive}) selects stepsizes significantly larger than the conservative theoretical bounds, thus accelerating convergence.}
    \label{fig:backtracking}
\end{figure}

\section{Conclusions}

We proposed a self-tuning explicit regularization framework for Image Scanning Microscopy (ISM) datasets. By formulating the reconstruction problem within a Bayesian Maximum A Posteriori framework, we combined multi-frame Poisson data fidelity terms with sparsity-promoting regularization strategies, including smoothed Total Variation and $\ell_1$ penalties. The proposed framework addresses one of the main limitations of existing Richardson–Lucy-based approaches, namely their semi-convergent behavior and the consequent need for empirical early stopping \cite{Liu}. By introducing explicit regularization together with automatic parameter selection based on the Residual Whiteness Principle, we obtained stable and robust reconstructions without requiring ground-truth information or manual tuning. The proposed spectral high-pass extension further enabled reliable parameter selection in the s$^2$ISM setting.

For the numerical computation of the solutions, we investigated first-order optimization strategies based on Proximal Gradient Descent and Mirror Descent corresponding to multiplicative algorithms. Both approaches can effectively solve the resulting reconstruction problems while preserving the physical non-negativity constraints. Extended experiments on simulated and real ISM datasets demonstrated improved reconstruction stability, enhanced image quality, and effective optical sectioning in challenging low-photon conditions.


An important direction for future work concerns the computational cost of optimal parameter selection. In the present work, the regularization parameter is determined through a grid-search use of RWP. A natural extension would instead consist in directly optimizing the whiteness functional within a bilevel optimization framework \cite{SantambrogioWP,PragliolaBilevelWP}, enabling simultaneous reconstruction and parameter learning in a unified procedure. Alternative definitions of the residual whiteness functional, for example based on detector-wise spatial correlations, may also be considered. More broadly, the framework introduced here opens the way to the integration of learned priors and, in particular, Plug-and-Play regularization approaches \cite{Kamilov_Pnp} tailored for Poisson data \cite{hurault2023convergent,Daniele2026,Benfenati2025} for for the reconstruction problems presented in this work.

\section*{Acknowledgments}

The work of S. Agostoni, C. Daniele and L. Calatroni is supported by the
funding received from the European Research Council (ERC) Starting project MALIN under
the European Union’s Horizon Europe programme (grant agreement No. 101117133).
The work of G. Vicidomini is supported by the
funding received from the European Research Council (ERC) Consolidator project BrightEyes under
the European Union’s Horizon Europe programme (grant agreement No. 818699).
This work represents only the view of the authors. The European Commission and the other
organizations are not responsible for any use that may be made of the information it contains.

\section*{Conflicts of Interest}

G. Vicidomini are co-founders and scientific advisors of Genoa Instruments. 

\section*{Declaration of generative AI and AI-assisted technologies}

During the preparation of this work, the authors used generative AI tools to assist with language proofreading of the manuscript, to support the review and debugging of the accompanying code, and to generate the graphical abstract. After using these tools, the authors reviewed and edited the content as needed.

\clearpage

\appendix

\section{Background on Optimization Theory}

In this section, we introduce the fundamental definitions of optimization necessary for an exhaustive understanding of the results in this work.  We report  in Section \ref{app: app_conv} the main convergence results.

\begin{definition}[Domain of a function] \label{def: dom}
    Let $f:\mathbb{R}^n \to \mathbb{R} \cup \{+\infty\}$. The (proper) domain of $f$ is the set of points where the function does not assume infinite values, that is
    $\operatorname{dom}(f)=\{x \in \mathbb{R}^n : f(x) < +\infty\}.$
\end{definition}

We now recall a classical property for smooth first-order optimization methods, that is the Lipschitz-smoothness of the gradient function.

\begin{definition}[$L$-smoothness]\label{L-smothness}
    Let $f:\mathbb{R}^n \to \mathbb{R} \cup \{+\infty\}$. $f$ is said to have Lipschitz-continuous gradient or, in short, $L$-smooth if:
    \begin{itemize}
        \item $f$ is differentiable in $\operatorname{int(dom}(f))$;
        \item $\exists L>0 : \|\nabla f(x)-\nabla f (y)\| \leq L\|x-y\|$ $\forall x,y \in \operatorname{dom}(f)$. 
    \end{itemize}
\end{definition}

Definition \ref{L-smothness} plays a crucial role in proving the convergence of the Proximal Gradient Method. Specifically, this condition is required to determine an admissible stepsize and to ensure the decrease of the target functional.

\smallskip

We now recall the notion of Legendre function, which is  used in this work to define the Mirror Descent update described in Section \ref{sec:MD}. Such functions are used to incorporate natural constraints on the desired solution, hence the definition requires some structural properties such as the growth till the boundary of the constraint set and its regularity in its interior. 
\begin{definition}[Legendre function]\label{def:Legendre function}
    Let \( h: \mathbb{R}^n \to \mathbb{R} \cup \{+ \infty\} \) have a non-empty domain. Then $h$ is Legendre-type if:
\begin{enumerate}
    \item[(i)] essentially smooth, that is \( h \) is differentiable on $ \operatorname{int}(\operatorname{dom}(h))$, with moreover $\| \nabla h(x^k) \| \to + \infty$ \ for every sequence \( \{x^k\}_{k \in \mathbb{N}} \subset \operatorname{dom}(h) \) converging to a boundary point of \( \operatorname{dom}(h)\) as \( k \to +\infty \);
    \item[(ii)]  strictly convex on \( \operatorname{int}(\operatorname{dom}(h))\).
\end{enumerate}
\end{definition}

We can now define a generalized notion of distance associated to such function in terms of the so-called Bregman divergence which we define below.
\begin{definition}[Bregman divergence] \label{Bregman divergence}
    Given a Legendre function $h: \mathbb{R}^n \to \mathbb{R} \cup \{+ \infty\}$ its associated Bregman divergence $D_h:\mathbb{R}^n \times \operatorname{int}(\operatorname{dom}(h)) \longrightarrow [0,+\infty]$ is defined in the following way:
\begin{equation}
    D_h(x,y)=h(x) - h(y) - \langle \nabla h (y), x-y  \rangle  \label{bregman distance}
\end{equation}
with $x \in \mathbb{R}^n$ and $y \in \operatorname{int}(\operatorname{dom}(h))$. If $x \notin \operatorname{dom}(h)$, then $D_h(x,y)=+\infty \ \forall y \in \operatorname{int}(\operatorname{dom}(h))$. 
\end{definition}
Note that Bregman divergences are not symmetric in general (for instance, there is no guarantee for them to be symmetric) hence are not distances in the classical sense.

Definitions \ref{def:Legendre function}, \ref{Bregman divergence} are the fundamental ingredients to define Mirror Descent algorithm described in Section \ref{sec:MD}. The following condition generalizes property \eqref{L-smothness}.
\begin{definition}[NoLip, $L$-SMAD, Relative smoothness \cite{nolip},\cite{Bolte2018},\cite{RelativeSmoothness}] \label{NoLip}
Let $h$ be a Legendre function and let $f:\mathbb{R}^n \rightarrow \mathbb{R} \cup \{+\infty\}$ be proper, lower semicontinuous with $\operatorname{dom}(f) \supset \operatorname{dom}(h)$ and differentiable on $\operatorname{int}(\operatorname{dom}(h))$.
We say that $f$ satisfies the (NoLip) condition with respect to $h$ with constant $L>0$, if the following holds:
\begin{equation}
   \text{There exists } L>0 \text{ such that } Lh-f \text{ is convex on $\operatorname{int}(\operatorname{dom}(h))$} \label{eq nolip}   
\end{equation}
\end{definition}
This condition can be seen as a generalization of the standard notion of $L$-smoothness (Definition \ref{L-smothness}) with respect to the new geometry induced by the Legendre function $h$. In close analogy to the standard Proximal Gradient Method, this property is crucial for establishing the monotonic decrease of the objective functional.




\section{Computation of the (generalized) Lipschitz constants}

In this section we provide the computation of the Lipschitz constants of the gradient for the functionals of interest $\Phi^{\text{ISM}}(\,\cdot\,; \underline{\mathbf{A}}, \underline{\mathbf{y}}, 
\underline{\mathbf{b}}), \Phi^{\text{$s^2$ISM}}(\,\cdot\,; \underline{\mathbf{A}}, \underline{\mathbf{y}}, 
\underline{\mathbf{b}})$ and for the smooth regularization functionals. Similarly, we provide an explicit computation of the generalized Lipschitz constants required for convergence of the MD algorithms

\begin{proposition}[Lipschitz constant of $ \nabla \Phi^{\text{ISM}}$]
\label{KL_L} 
Let $\Phi^{\text{ISM}}(\,\cdot\,; \underline{\mathbf{A}}, \underline{\mathbf{y}}, 
\underline{\mathbf{b}})$ be  the functional defined in \eqref{ml_min} and $\underline{\mathbf{b}} = [ \mathbf{b}_1, \dots, \mathbf{b}_{25}]$. If $\mathbf{b_d} > 0$ for all $d$ , then $\Phi^{\text{ISM}}(\,\cdot\,; \underline{\mathbf{A}}, \underline{\mathbf{y}}, 
\underline{\mathbf{b}})$ is $L_{\Phi^{\text{ISM}}}$-smooth with $L_{\Phi^{\text{ISM}}}$ given by:
\begin{equation} \label{eq:L_ISM}
    L_{\Phi^{\text{ISM}}} \leq \sum_{d=1}^{25} \frac{\max(\mathbf{y}_d)}{b_d^2} \cdot 
    \max(\mathbf{A}_d \mathbf{1}) \cdot \max(\mathbf{A}_d^\top \mathbf{1}).
\end{equation}
\end{proposition}

\begin{proof}
This proof is an extension of \cite[Lemma 1]{Harmany}.
The Hessian of $\Phi^{\text{ISM}}$ with respect to $\mathbf{x}$ is:
\begin{equation}
    \nabla^2 \Phi^{\text{ISM}}(\mathbf{x};\, \underline{\mathbf{A}}, \underline{\mathbf{y}}, 
    \underline{\mathbf{b}}) = \sum_{d=1}^{25} \mathbf{A}_d^\top \operatorname{diag}\!\left( 
    \frac{\mathbf{y}_d}{(\mathbf{A}_d \mathbf{x} + \mathbf{b}_d)^2} \right) \mathbf{A}_d.
\end{equation}
Since the Lipschitz constant of $\nabla \Phi^{\text{ISM}}$ equals the supremum of the spectral 
norm of the Hessian, it suffices to bound the largest eigenvalue of 
each summand. For each $d = 1, \dots, 25$, denoting by $L_d$ the Lipschitz constant associated to the gradient of each detector component of the data term, we have:
\begin{align*}
    L_d &= \lambda_{\max}\!\left( \mathbf{A}_d^\top \operatorname{diag}\!\left( 
    \frac{\mathbf{y}_d}{(\mathbf{A}_d \mathbf{x} + \mathbf{b}_d)^2} \right) 
    \mathbf{A}_d \right) \leq \frac{1}{b_d^2}\, \lambda_{\max}\!\left( \mathbf{A}_d^\top 
    \operatorname{diag}(\mathbf{y}_d)\, \mathbf{A}_d \right) \leq \frac{1}{b_d^2}\, \left\| \mathbf{A}_d^\top \operatorname{diag}(\mathbf{y}_d)\, 
    \mathbf{A}_d \right\|_2 \\
    &\leq \frac{\max(\mathbf{y}_d)}{b_d^2}\, \|\mathbf{A}_d\|_2^2 \leq \frac{\max(\mathbf{y}_d)}{b_d^2}\, \|\mathbf{A}_d\|_1 \|\mathbf{A}_d\|_\infty = \frac{\max(\mathbf{y}_d)}{b_d^2} \cdot \max(\mathbf{A}_d \mathbf{1}) \cdot \max(\mathbf{A}_d^\top \mathbf{1}),
\end{align*}
where the last equality follows from the fact that $\|\mathbf{A}_d\|_1 = 
\max(\mathbf{A}_d^\top \mathbf{1})$ and $\|\mathbf{A}_d\|_\infty = 
\max(\mathbf{A}_d \mathbf{1})$ for non-negative matrices. The desired constant 
is then bounded by summing over all frames:
\begin{equation}
    L_{\Phi^{\text{ISM}}} \leq \sum_{d=1}^{25} L_d = \sum_{d=1}^{25} 
    \frac{\max(\mathbf{y}_d)}{b_d^2} \cdot \max(\mathbf{A}_d \mathbf{1}) \cdot 
    \max(\mathbf{A}_d^\top \mathbf{1}). 
\end{equation}
\end{proof}

\begin{remark}[Dual--plane case]\label{rmk: lipschitz constant of KL in 2 planes ism}
    Proposition \ref{KL_L} can be generalized in a straightforward manner also to the functional $\Phi^{\text{$s^2$ISM}}(\,\cdot\,; \underline{\mathbf{A}}, \underline{\mathbf{y}}, 
\underline{\mathbf{b}})$.
In fact, similarly to the proof above, for a given $d \in \{1,\dots,25\}$, we can consider the operator $\underline{\A}_d\underline{\x} =  \A_{d,1}\x_1 + \A_{d,2}\x_2$ and repeat the same computations above.
\end{remark}

\begin{proposition}[Lipschitz constant of $\nabla \text{TV}_{\epsilon}$]
\label{prop:lip_F}

Let $\operatorname{TV}_{\epsilon}(\mathbf{x})$ as defined in \eqref{eq:TV_smooth}. If $\epsilon > 0$, then $\operatorname{TV}_{\epsilon}$ has $L_{\operatorname{TV}_{\epsilon}}$-Lipschitz continuous gradient with:
\begin{equation} \label{eq:L_ISM_TV}
L_{\operatorname{TV}_{\epsilon}} \leq  \frac{8}{\epsilon}.
\end{equation}
\end{proposition}

\begin{proof} This proof is based on \cite[Theorem 3.1]{chambolle2004algorithm}.
Similarly to Proposition , the idea is to bound the spectral norm of the Hessian of the functional of interest. In the proof, for the sake of simplicity, we omit the dependency on $\epsilon$ in the object involved. Let $\nabla = \begin{bmatrix}\nabla_h \\ \nabla_v\end{bmatrix} \in \mathbb{R}^{2N \times N}$ 
be the discrete gradient operator. The Hessian of  the $\operatorname{TV}_\epsilon$ functional is:
\begin{equation}    \mathbf{H}_{\operatorname{TV}}(\mathbf{x}) = \nabla^\top \mathbf{W}(\mathbf{x})\, \nabla 
    \in \mathbb{R}^{N \times N},
\end{equation}
where $\mathbf{W}(\mathbf{x}) \in \mathbb{R}^{2N \times 2N}$ is block-diagonal with 
$2 \times 2$ blocks given by:
\begin{equation}
    \mathbf{W}_i(\mathbf{x}) = \frac{1}{\sqrt{\|(\nabla \mathbf{x})_i\|^2 + \epsilon^2}}
    \left(\mathbf{I}_2 - \frac{(\nabla \mathbf{x})_i\,(\nabla \mathbf{x})_i^\top}
    {\|(\nabla \mathbf{x})_i\|^2 + \epsilon^2}\right), \qquad i = 1, \dots, N.
\end{equation}
Each block can be written as $\mathbf{W}_{i} = \frac{1}{\sqrt{\|(\nabla\mathbf{x})_i\|^2 
+ \epsilon^2}}\,\mathbf{P}_i$, where $\mathbf{P}_i$ is an orthogonal projection with 
eigenvalues:
\begin{equation}
    \lambda_1(\mathbf{P}_i) = \frac{\epsilon^2}{\|(\nabla \mathbf{x})_i\|^2 + \epsilon^2} 
    \leq 1, \qquad \lambda_2(\mathbf{P}_i) = 1,
\end{equation}
so $\|\mathbf{P}_i\|_2 = 1$ and therefore:
\begin{equation}
    \|\mathbf{W}_i(\mathbf{x})\|_2 \leq \frac{1}{\sqrt{\|(\nabla \mathbf{x})_i\|^2 
    + \epsilon^2}} \leq \frac{1}{\epsilon}.
\end{equation}
By submultiplicativity and the identity $\nabla^\top \nabla = -\Delta$, where 
$\|-\Delta\|_2 \leq 8$ for the  2D finite-difference Laplacian discretized with the usual five-point formula, we obtain:
\begin{equation}
    \|\mathbf{H}_{\text{TV}}(\mathbf{x})\|_2 \leq \|\nabla\|_2^2\,\|\mathbf{W}(\mathbf{x})\|_2 
    \leq \frac{8}{\epsilon}, 
\end{equation}
as desired.
\end{proof}
Thanks to the result above we deduce that for the actual regularization functional considered in this work, which reads $\lambda\operatorname{TV}_\epsilon(\mathbf{x})$ there holds consequently:
\begin{equation}
\|\nabla^2(\lambda\operatorname{TV})(\mathbf{x})\|_2 
    = \lambda\,\|\mathbf{H}_{\text{TV}}(\mathbf{x})\|_2 \leq \frac{8\lambda}{\epsilon}.
\end{equation}

\begin{remark}[Lipschitz constant of the regularization function for the dual plane problem]
    In the dual-plane problem the (smoothed) regularization function $\tilde{R}$ acts only on the first component of $\underline{\x}$ (the in--focus signal) as $\tilde{R}(\underline{\x}) = R(\x_1)$. Thus there holds $L_{\tilde{R}} = L_{R}$, which allows for using in the case the same estimates above.
\end{remark}
We can now combine the two estimates by triangular inequality to get an estimate of the Lipschitz smoothness constant of the functional $f(\cdot) = \Phi^{ISM}(\,\cdot\,; \underline{\mathbf{A}}, \underline{\mathbf{y}}, 
\underline{\mathbf{b}}) + \lambda \text{TV}_{\epsilon}(\cdot)$, which is $L_{f}$-- smooth with constant:
\begin{equation}
    L_{f} \leq \sum_{d=1}^{25} 
    \frac{\max(\mathbf{y}_d)}{b_d^2} \cdot \max(\mathbf{A}_d \mathbf{1}) \cdot 
    \max(\mathbf{A}_d^\top \mathbf{1}) + \frac{8 \lambda}{\epsilon}
    \end{equation}

In the Mirror Descent framework, computing a generalized Lipschitz constant in the sense of Definition \ref{NoLip} is often a challenging task. In our setting of interest, however, under the choice of the logarithmic barrier $h$ in \eqref{eq: burg's entropy}, we can state the following result:
\begin{proposition}
    Let $\Phi^{\text{ISM}}(\,\cdot\,; \underline{\mathbf{A}}, \underline{\mathbf{y}}, 
\underline{\mathbf{b}})$  the functional defined in \eqref{ml_min} and $h$ defined in \eqref{eq: burg's entropy}. Then any $L> 0$ such that:
\begin{equation}
    L \geq \sum^{25}_{d=1} \|\y_d\|_1,
\end{equation}
satisfies Definition \ref{NoLip}.
\end{proposition}
\begin{proof}
Recalling the definition of $\Phi^{\text{ISM}}(\cdot;\underline{\mathbf{A}}, \underline{\mathbf{y}}, \underline{\mathbf{b}})$:
$$
\Phi^{\text{ISM}}(\x;\underline{\mathbf{A}}, \underline{\mathbf{y}}, \underline{\mathbf{b}})
=
\sum_{d=1}^{25}
\operatorname{KL}(\mathbf{y}_d,\A_d\x+\mathbf{b}_d),$$
we have that for every $d = 1,...25$, we can apply \cite[Lemma 7]{nolip}, which implies that: \begin{equation}
    \|\y_d\|_1 h(\cdot) - \operatorname{KL}(\mathbf{y}_d,\A_d\cdot+\mathbf{b}_d) \text{ is convex on } \operatorname{\operatorname{int}(dom}(h)).
\end{equation}
The thesis follows by the fact the sum of convex functions is convex.
\end{proof}

The adaptation to the
dual-plane case is straightforward as described in Remark \ref{rmk: lipschitz constant of KL in 2 planes ism}.

\begin{remark}[On the generalized Lipschitz constant of the composite problem]
Let us note that deriving an explicit generalized Lipschitz constant for composite problems involving a smooth regularization functional is a difficult task as for the latter such estimate could be challenging. However, to guarantee existence of such constant it is possible to follow a general reasoning that applies to all regularization functions $R$ that have a Lipschitz gradient and are twice differentiable on the positive orthant (as in our case). Further details and a more precise analysis of this case are available in \cite[Section B.3.3]{Daniele2026}.





Note that that the lack of an explicit estimate of the generalized Lipschitz constant poses a challenge for choosing a fixed, global stepsize. However, employing a backtracking strategy, such as the one depicted in Algorithm \ref{alg:backtracking_adaptive}, allows this issue to be bypassed.
\end{remark}

\section{Convergence Properties of the Algorithms}
\label{app: app_conv}

In this section, we provide further details about the convergence of the algorithms considered in this work. 

\subsection{Proximal Gradient Descent}

We present the standard convergence result for Proximal Gradient Descent which applies to the different setups considered, see, e.g. \cite{beck}. To leave the discussion general, we consider the following composite minimization problem:
\begin{equation}
    \underset{x \in \mathbb{R}^n}{\operatorname{argmin}}~ F(x) := f(x) + g(x), \label{eq: composite problem F = f+g}
\end{equation}
where we assume the following:
\begin{assumption} \label{assumption 1} \leavevmode
    \begin{enumerate}
        \item $f : \mathbb{R}^n \to \mathbb{R} \cup \{+ \infty\}$ is proper, convex and $L_{f}$-smooth.
        \item $g : \mathbb{R}^n \to \mathbb{R} \cup \{+ \infty\}$ is proper, lower semi-continuous and convex.
        \item We assume that  $X^*:=\left\{ \text{Argmin}~F\right\} \neq \emptyset$. For $\x^*\in X^*$ we define $F(\x^*)=F_{\text{opt}}$.   
    \end{enumerate}
\end{assumption}

\begin{theorem}[Theorem 10.1 \cite{beck}]
\label{converg_pgd}
Suppose that Assumption \ref{assumption 1} holds. Let $\{  \mathbf{x}^k \}_{k \geq 0}$ be the sequence given by:
\begin{equation}
    \mathbf{x}^{k+1} = \operatorname{prox}_{\alpha g} (\mathbf{x}^k - \alpha \nabla f(\mathbf{x}^k)), \label{eq: general pgd}
\end{equation}
where $\alpha > 0$ is the stepsize.
     If $\alpha \in \left(0, \frac{2}{L_f}\right) $, then for any optimal solution \( \x^* \in X^* \) and \( k \geq 0 \):
    $$F(\x^k) - F_{\text{opt}} \leq \frac{L_f \Vert \x^0 - \x^* \Vert^2}{2k} = O(1/k)$$
\end{theorem}
Theorem \ref{converg_pgd} can be applied to establish the convergence of both \eqref{eq:PGD} and \eqref{eq:ProjGD}, which are particular cases of \eqref{eq: general pgd} for specific choices of $f$ and $g$. Table \ref{tab:scelte_f_g} summarizes such choices. Note, that when the non-smooth function $g$ reduces to the indicator function of a closed and convex constraint set, PGD becomes nothing but projected gradient descent.

\begin{table}[htbp]
\centering
\caption{Choices of $f$ and $g$ for the models considered in this work.}
\label{tab:scelte_f_g}
\vspace{2mm} 
\begin{tabular}{lll}
\toprule
\textbf{Algorithm} & \textbf{$f$} & \textbf{$g$} \\
\midrule
\textbf{PGD} & $\Phi^{\text{ISM}}(\,\cdot\,; \underline{\mathbf{A}}, \underline{\mathbf{y}}, 
\underline{\mathbf{b}})$ & $\|\cdot\|_1+ \iota_{\geq 0}(\cdot) $ \\
\addlinespace 
\textbf{ProjGD} & $\Phi^{\text{ISM}}(\,\cdot\,; \underline{\mathbf{A}}, \underline{\mathbf{y}}, 
\underline{\mathbf{b}}) + \text{TV}_{\epsilon}$       & $\iota_{\geq 0}(\cdot)$ \\
\bottomrule
\end{tabular}
\end{table}

\subsection{Mirror Descent}



We provide here a brief overview on the convergence results for the MD scheme \eqref{eq:implicit MD for KL} as used in this work. The key to establishing convergence lies in the choice of the stepsize $\alpha$, which, analogously to Proximal Gradient Descent, must be inversely proportional to the generalized Lipschitz constant defined in Definition \ref{NoLip} of the functional considered. Under this condition, convergence results are well-established \cite{nolip, Bolte2018, RelativeSmoothness}. In particular, if $h$ is chosen as in \ref{eq: burg's entropy}, then it holds that \cite[Theorem 1]{nolip}:

\begin{equation}
    F(\mathbf{x}^k) - F (\mathbf{u}) \leq \frac{2L_F}{k} D_{h}(\mathbf{u}, \mathbf{x}^0), \text{ for all } \mathbf{u} \in \operatorname{dom}(h),
\end{equation}

where $\{ \mathbf{x}^k \}_{k \geq 0}$ is the MD sequence \eqref{eq:implicit MD for KL}, $F = \Phi^{\text{ISM}}(\,\cdot\,; \underline{\mathbf{A}}, \underline{\mathbf{y}}, 
\underline{\mathbf{b}}) + R(\cdot)$, with $R$ being one of the two regularization functions considered in Section \ref{sec:reg} and  $L_F$ is the generalized Lipschitz constant of $F$ w.r.t. $h$.


\section{Statistics of zero-truncated Poisson distribution}\label{masked_wh_proof}

We review in this Section the main properties of masked/truncated Poisson process which we employ to apply the Residual Whiteness Principle in Section \ref{sec:masked_wp}.
For that, we follow \cite[Proposition 4.2]{Bevilacqua_2023}
Let $Y \sim \mathcal{P}(\nu)$ be a standard Poisson random variable with mean $\nu > 0$, 
and let $Y_+ \sim \mathcal{P}_+(\nu)$ denote its zero-truncated counterpart, i.e.\ the 
distribution of $Y$ conditioned on $Y > 0$. The probability mass function of $Y_+$ is:
\begin{equation}
    \mathbb{P}_{Y_+}(y) = \frac{1}{1 - e^{-\nu}}\,\mathbb{P}_Y(y) 
    = \frac{1}{1 - e^{-\nu}} \cdot \frac{\nu^y\, e^{-\nu}}{y!}, 
    \qquad y \in \mathbb{N}_0 := \mathbb{N} \setminus \{0\}.
\end{equation}
That is, $\mathbb{P}_{Y_+}(y) = T(\nu)\,\mathbb{P}_Y(y)$, where:
\begin{equation}
    T(\nu) := \frac{1}{1 - e^{-\nu}}, \qquad V(\nu) := \frac{1-(1+\nu)e^{-\nu}}{(1-e^{-\nu})^2}.
\end{equation}

\begin{proposition}
Let $Y \sim \mathcal{P}(\nu)$ and $Y_+ \sim \mathcal{P}_+(\nu)$, with $\nu \in \mathbb{R}_{++}$. 
Then the expected value and variance of $Y_+$ are:
\begin{equation}
    \mathbb{E}[Y_+] = T(\nu)\,\mathbb{E}[Y] = \frac{\nu}{1 - e^{-\nu}}, 
    \label{eq:zt_mean}
\end{equation}
\begin{equation}
    \operatorname{Var}[Y_+] = V(\nu)\operatorname{Var}[Y] 
    = \frac{\nu}{(1-e^{-\nu})^2}\!\left(1 - \frac{1+\nu}{e^\nu}\right).
    \label{eq:zt_var}
\end{equation}
\end{proposition}

\begin{proof}
Since $\mathbb{P}_{Y_+}(y) = T(\nu)\,\mathbb{P}_Y(y)$ for all $y \in \mathbb{N}_0$, the 
$m$-th order raw moments of $Y_+$ satisfy:
\begin{equation}
    \mathbb{E}[Y_+^m] = \sum_{y=1}^{\infty} y^m\,\mathbb{P}_{Y_+}(y) 
    = T(\nu)\sum_{y=1}^{\infty} y^m\,\mathbb{P}_Y(y) 
    = T(\nu)\sum_{y=0}^{\infty} y^m\,\mathbb{P}_Y(y) 
    = T(\nu)\,\mathbb{E}[Y^m],
\end{equation}
where the third equality holds because the $y=0$ term contributes zero. Setting $m=1$ 
and recalling that $\mathbb{E}[Y] = \nu$:
\begin{equation}
    \mathbb{E}[Y_+] = T(\nu)\,\nu = \frac{\nu}{1 - e^{-\nu}},
\end{equation}
which establishes~\eqref{eq:zt_mean}. Setting $m=2$ and using $\mathbb{E}[Y^2] = 
\nu(1+\nu)$ for $Y\sim\mathcal{P}(\nu)$:
\begin{equation}
    \mathbb{E}[Y_+^2] = T(\nu)\,\mathbb{E}[Y^2] = \frac{\nu(1+\nu)}{1 - e^{-\nu}}.
\end{equation}
The variance of $Y_+$ then follows from $\operatorname{Var}[Y_+] = \mathbb{E}[Y_+^2] - 
(\mathbb{E}[Y_+])^2$:
\begin{align}
    \operatorname{Var}[Y_+] 
    &= T(\nu)\,\mathbb{E}[Y^2] - \bigl(T(\nu)\,\mathbb{E}[Y]\bigr)^2 = T(\nu)\,\nu(1+\nu) - T(\nu)^2\,\nu^2 \notag= T(\nu)\,\nu\bigl(1 + \nu - \nu\,T(\nu)\bigr) \notag\\
    &= T(\nu)\bigl(1 + \nu - \nu\,T(\nu)\bigr)\nu = V(\nu)\operatorname{Var}[Y],
\end{align}
where the last equality uses $\operatorname{Var}[Y] = \nu$ and the identity 
$V(\nu) = T(\nu)(1 + \nu - \nu\,T(\nu))$, which follows directly from the definitions 
of $T$ and $V$. Substituting the explicit form of $T(\nu)$ yields~\eqref{eq:zt_var}.
\end{proof}

In the context of Section~\ref{sec:masked_wp}, we apply this result pixelwise: for each 
$(i,d) \in \mathcal{M}$, we set $\nu = \nu_{i,d}^\lambda$ and identify:
\begin{equation}
    \mu_{i,d}^+ = \mathbb{E}[Y_+] = \frac{\nu_{i,d}^\lambda}{1 - e^{-\nu_{i,d}^\lambda}},
    \qquad
    (\sigma_{i,d}^+)^2 = \operatorname{Var}[Y_+] 
    = \mu_{i,d}^+\!\left(1 - \mu_{i,d}^+\,e^{-\nu_{i,d}^\lambda}\right),
\end{equation}
which are precisely the quantities used to define the masked standardized residual 
$z_{i,d}^+$ in Section~\ref{sec:masked_wp}.

\bibliographystyle{plain}
\bibliography{cas-refs}



\end{document}